\documentclass[twoside,11pt]{article}

\usepackage{amsfonts}
\usepackage{amssymb} % included in jmlr2e.sty
\usepackage{amsthm}
\usepackage{authblk} % \author
\usepackage{bbm}
\usepackage{bm}
\usepackage{cancel}
\usepackage{color}
\usepackage{dsfont}
\usepackage{extarrows}
\usepackage{mathdots}
\usepackage[top=1in, bottom=1in, left=1in, right=1in]{geometry}
\usepackage{graphicx}
\usepackage{mathrsfs}
\usepackage{mdframed}
\usepackage[round]{natbib}
\usepackage{subfig} % subfloat
\usepackage{url}
\usepackage{setspace}
\usepackage{soul} %\st
\usepackage[dvipsnames]{xcolor}

\renewcommand{\d}{\mathop{}\!\mathrm{d}}

\newcommand{\mnorm}[1]{{\left\vert\kern-0.25ex\left\vert\kern-0.25ex\left\vert #1 
    \right\vert\kern-0.25ex\right\vert\kern-0.25ex\right\vert}}
\newcommand\newday[1]{\leavevmode\xleaders\hbox{*}\hfill\kern0pt\\ \centerline{\Large \textbf{#1}}}

\newcommand{\y}{\boldsymbol{y}}

\usepackage[shortlabels]{enumitem}

\numberwithin{equation}{section}

\newcommand{\refyu}[1]{#1}%\col{red}{#1}}

\newtheorem{theorem}{Theorem}
\newenvironment{thm}[1]
  {\renewcommand\thetheorem{#1}\theorem}
  {\endtheorem}

\newtheorem{lemma}[theorem]{Lemma}

  \newtheorem{corollary}[theorem]{Corollary}

  \theoremstyle{remark}
    \newtheorem{remark}{Remark}

\newtheorem{definition}{Definition}

\allowdisplaybreaks

\makeatletter
\newcommand{\pushright}[1]{\ifmeasuring@#1\else\omit\hfill$\displaystyle#1$\fi\ignorespaces}
\newcommand{\pushleft}[1]{\ifmeasuring@#1\else\omit$\displaystyle#1$\hfill\fi\ignorespaces}
\makeatother

\title{Interaction Models and Generalized Score Matching for
  Compositional Data}
\author[1]{Shiqing Yu}
\author[2]{Mathias Drton}
\author[3]{Ali Shojaie}
\affil[1]{Department of Statistics, University of Washington, Seattle, Washington, 98195, U.S.A.}
\affil[2]{Department of Mathematics, Technical University of Munich, 85748 Garching bei M\"{u}nchen, Germany}
\affil[3]{Department of Biostatistics, University of Washington, Seattle, Washington, 98195, U.S.A.}

\begin{document}
\begin{spacing}{1.5}

\maketitle

\begin{abstract}
  Applications such as the analysis of microbiome data have led to
renewed interest in statistical methods for compositional data, i.e.,
multivariate data in the form of probability vectors that contain
relative proportions. In particular, there is considerable interest in
modeling interactions among such relative proportions. To this end we
propose a class of exponential family models that accommodate general
patterns of pairwise interaction while being supported on the probability
simplex. Special cases include the family of Dirichlet distributions
as well as Aitchison's additive logistic normal distributions. 
Generally, the distributions we consider have a density that features a
difficult to compute normalizing constant. To circumvent this issue, we design effective
estimation methods based on generalized versions of score matching. A
high-dimensional analysis of our estimation methods shows that the
simplex domain is handled as efficiently as previously studied
full-dimensional domains.

\end{abstract}

\newcommand{\Deltam}{\Delta}

\section{Introduction}

%\subsection{Models for Distributions on the Simplex}

Compositional data, where each data point is an element of the
probability simplex, frequently arise in real-world
applications. Since the components of such data points
are nonnegative and add up to one, they naturally represent proportions,
percentages or the event probabilities in multinomial
distributions.  Aside from probability and statistics itself, compositional data
arise from applications such as rock composition in geology \citep{ait82,paw06},
demographic data \citep{llo12}, % in demography,
concentrations in chemistry \citep{rol92}, and more recently, very prominently as
relative abundances in measurements of microbiome composition \citep{fuk12,xia13,li15,shi16,ran18}.

For a vector $\boldsymbol{x}\in\mathbb{R}^m$, write
$\boldsymbol{x}\succ 0$ (or $\boldsymbol{x}\succeq 0$) if all of its
coordinates $x_j>0$ (or $x_j\ge 0$).  Let
$\mathbf{1}_m\in\mathbb{R}^m$ be the vector with all entries equal to 1.
Compositional data with $m$ variables are comprised of points in the
$(m-1)$-dimensional probability simplex
\begin{equation}
  \label{eq:delta-simplex}
  \Delta\equiv \Delta_{m-1}=\left\{\boldsymbol{x}\in\mathbb{R}^m:
    \boldsymbol{x}\succeq\boldsymbol{0},\,\mathbf{1}_m^{\top}\boldsymbol{x}=1\right\}.
\end{equation}
The $A^{m-1}$ model ($A^d$ in the original notation)
proposed in
\citet{ait85} yields the classical approach for statistical analysis of
compositional data in $\Deltam$. The distributions in this
model, which is described in detail in Section \ref{sec_Ad_Models},
can be parameterized as having a density proportional to
\[\exp\left(-\frac{1}{2}\log\boldsymbol{x}^{\top}\mathbf{K}\log\boldsymbol{x}+\boldsymbol{\eta}^{\top}\log\boldsymbol{x}\right),\quad\boldsymbol{x}\in\Deltam,\]
where the interaction matrix $\mathbf{K}\in\mathbb{R}^{m\times m}$,
$\mathbf{K}\mathbf{1}_m=\mathbf{0}_m$, and 
$\boldsymbol{\eta}\in\mathbb{R}^m$ are parameters. This class includes the Dirichlet distribution when $\mathbf{K}=\mathbf{0}$ (with concentration parameters $\boldsymbol{\eta}+\mathbf{1}_m\succ\mathbf{0}_m$), and the widely used additive logistic normal distribution \citep{ait82} that corresponds to the case where $\mathbf{1}_m^{\top}\boldsymbol{\eta}=-m$, positing that $(\log(x_1/x_m),\ldots,\log(x_{m-1}/x_m))\in\mathbb{R}^{m-1}$ is multivariate normal.

As a flexible extension of Aitchison's class, we propose in this paper
a \emph{power interaction model}
 with densities
\begin{equation}\label{eq_intro_a_b_models}
p_{\boldsymbol{\eta},\mathbf{K}}(\boldsymbol{x})\propto\exp\left(-\frac{1}{2a}{\boldsymbol{x}^a}^{\top}\mathbf{K}\boldsymbol{x}^a+\frac{1}{b}\boldsymbol{\eta}^{\top}\boldsymbol{x}^b\right),\quad\boldsymbol{x}\in\Deltam.
\end{equation}
Here, $a\geq 0, b\geq 0$, and for $a=b=0$ we let
${\boldsymbol{x}^a}^{\top}\mathbf{K}\boldsymbol{x}^{a}/a\equiv\log\boldsymbol{x}^{\top}\mathbf{K}\log\boldsymbol{x}$
and $\boldsymbol{x}^b/b\equiv\log\boldsymbol{x}$ (for $a=0$ or $b=0$, $\boldsymbol{x}\succ\boldsymbol{0}$ almost surely). This class, which we
also term \emph{$a$-$b$ interaction models}, contains the $A^{m-1}$
models as a special case when $a=b=0$ and
$\mathbf{K}\mathbf{1}_m=\mathbf{K}^{\top}\mathbf{1}_m=\mathbf{0}_m$,
but allows for many other possibilities such as normal distribution or
square root models \citep{ino16} truncated to the simplex.
Importantly, for this new $a$-$b$ class on the simplex, we are able to
derive consistent estimators that are computationally efficient and scalable by suitably adapting the
%technique of generalized
%score matching of \citet{yu21} and making use of
generalized score
matching methods that were developed in \citet{yu21} for domains with
positive measure.

Through its flexibility, the model class from
(\ref{eq_intro_a_b_models}) offers new tractable approaches to study interaction
structures in compositional data via estimation of the parameters
$(\eta,\mathbf{K})$.  Similarly, it offers increased flexibility in
problems such as comparing two sets of compositional data; see for
example, Chapter 7 of \citet{paw07} for traditional approaches to this
problem.

Flexibility aside, the models proposed in (\ref{eq_intro_a_b_models})
have the strong appeal that with $a>0$ and $b>0$, they allow one to
directly handle proportions that are exactly zero.  Indeed, many
modern applications feature ``sparse data,'' for which a considerable
number of the components of data vectors are zero and heuristics like
adding small positive numbers to data need to be adopted to apply
approaches such as the Aitchison additive logistic normal model, which
relies on logarithms.

%\subsection{Outline of the Paper}

\paragraph{Outline of the Paper.}
In Section~\ref{Generalized Score Matching for General Domains} we
first review the generalized score methodology that was developed in
\citet{yu21} for data from domains of positive measure.  We then show
how to modify this methodology to obtain estimators for models of
distributions supported on a simplex, and concretely, the probability
simplex $\Deltam$.  For high-dimensional applications we propose
$\ell_1$-regularization to allow one to exploit sparse interaction
structure.  In Section~\ref{$a$-$b$ Models on Standard Simplices} we
turn to our proposed class of $a$-$b$ interaction models on $\Deltam$.
For these models, we study identifiability issues as well as
conditions for the density kernels to be normalizable to a proper
density. In Section~\ref{sec:estim-power-inter}, we customize our
estimation strategy from Section~\ref{Generalized Score Matching for
  General Domains} to the $a$-$b$ interaction models.
High-dimensional consistency results for our regularized estimators
are derived in Section~\ref{Theory}.  Previous results on
$m$-dimensional domains of positive measure show success of
regularized score matching estimators when the sample size scales as
$n=\Omega(\log m)$.  We are able to show that this scaling also holds
for simplex domains.  In Section~\ref{Numerical Experiments}, we
perform numerical experiments that explore choice of tuning parameters
for our estimators---the experiments use the code we make available in
the R package \texttt{genscore}. 
Section~\ref{sec:micr-data-analys} illustrates the utility of our
proposed models and estimation methodology in an application to
microbiome data.  The paper concludes with a discussion in
Section~\ref{sec:discussion}.  Appendix~\ref{Proofs} collects the
proofs of our theorems.

% \subsection{Notation}
\paragraph{Notation.}
Upper- versus lower-case is used to denote random quantities. Regular
font is reserved for scalars (e.g.~$a$, $X$) and boldface for vectors (e.g.~$\boldsymbol{a}$, $\boldsymbol{X}$).
Matrices are written in upright bold, with constant matrices in upper-case ($\mathbf{K}$, $\mathbf{M}$) and random data matrices in lower-case ($\mathbf{x}$, $\mathbf{y}$).  Superscripts to index rows and subscripts to index columns in a data matrix $\mathbf{x}$, i.e.~$\boldsymbol{X}^{(i)}$ is the $i$-th row/sample, and $X_j^{(i)}$ is its $j$-th feature.

For  $\boldsymbol{u},\boldsymbol{v}\in\mathbb{R}^m$,
$\boldsymbol{u}\odot\boldsymbol{v}\equiv(u_1v_1,\ldots,u_mv_m)$
denotes the Hadamard  (or element-wise) product. The $\ell_a$-norm for $a\geq 1$ is denoted $\|\boldsymbol{u}\|_a=(\sum_{j=1}^m|u_j|^a)^{1/a}$, with $\|\boldsymbol{u}\|_{\infty}=\max_{j=1,\ldots,m}|u_j|$. For $a\in\mathbb{R}$, write $\boldsymbol{v}^a\equiv (v_1^a,\ldots,v_m^a)$. Similarly for vector-valued $\boldsymbol{f}:\mathbb{R}^m\to\mathbb{R}^m$, $\boldsymbol{x}\mapsto (f_1(\boldsymbol{x}),\ldots, f_m(\boldsymbol{x}))$, we write $\boldsymbol{f}^a(\boldsymbol{x})\equiv(f_1^a(\boldsymbol{x}),\ldots,f_m^a(\boldsymbol{x}))$. We also write $\boldsymbol{f}'(\boldsymbol{x})\equiv(\partial f_1(\boldsymbol{x})/\partial x_1,\ldots,\partial f_m(\boldsymbol{x})/\partial x_m)$.

The vectorization of a matrix
$\mathbf{K}=[\kappa_{ij}]_{i,j}\in\mathbb{R}^{n\times m}$ is the
vector in $\mathbb{R}^{nm}$ obtained by stacking the matrix
columns. The Frobenius norm of the matrix is $\mnorm{\mathbf{K}}_{F}=\|\mathrm{vec}(\mathbf{K})\|_2$, its max norm is $\|\mathbf{K}\|_{\infty}\equiv\|\mathrm{vec}(\mathbf{K})\|_{\infty}\equiv\max_{i,j}|\kappa_{ij}|$, and its $\ell_a$--$\ell_b$ operator norm is $\mnorm{\mathbf{K}}_{a,b}\equiv\max_{\boldsymbol{x}\neq\boldsymbol{0}}\|\mathbf{K}\boldsymbol{x}\|_b/\|\boldsymbol{x}\|_a$, with $\mnorm{\mathbf{K}}_a\equiv\mnorm{\mathbf{K}}_{a,a}$.

For a function $f$ of a vector $\boldsymbol{x}$, we may also write
$f(x_j;\boldsymbol{x}_{-j})$ to stress the dependency on $x_j$,
especially when $\boldsymbol{x}_{-j}$ is fixed and only $x_j$ is
varied; e.g., $\partial_j f(\boldsymbol{x})=\partial f(x_j;\boldsymbol{x}_{-j})/\partial x_j$. For two compatible functions $f$ and $g$, $f\circ g$ denotes their function composition.

%Unless otherwise noted, the probability density of a distribution function is with respect the Lebesgue measure on $\mathbb{R}^m$.

\section{Generalized Score Matching for General Domains}\label{Generalized Score Matching for General Domains}
  
\subsection{Generalized Score Matching}
\label{sec:gener-score-match}

Score matching is an effective method for estimation of Lebesgue densities that
are defined only up to a finite normalizing constant. Originally
proposed by \citet{hyv05} for densities supported on $\mathbb{R}^m$
and by \citet{hyv07} for support $\mathbb{R}_+^m$, it provides
consistent estimators that can be evaluated without calculation of the often
intractable normalizing constant. Possible, but ultimately less
convenient analogues for the discrete case were discussed also by
\citet{lyu09}. Again in the continuous case, \citet{yu18,yu19}
generalized score matching to allow more efficient estimation for densities on
$\mathbb{R}_+^m$. Further generalizations in \citet{liu19} and \citet{yu21} treat
densities supported on more complicated domains of positive measure.

Let $\mathcal{P}(\mathcal{D})$ be a family of distributions of
interest with twice continuously differentiable densities on a domain
$\mathcal{D}\subseteq\mathbb{R}^m$. %If $p$ is a density function, we write $p\in\mathcal{P}(\mathcal{D})$ for ease of notation if there exists some distribution $P\in\mathcal{P}(\mathcal{D})$ that has density $p$.
The main idea behind the generalized score matching methods of
\citet{yu19,yu21} is to estimate an unknown density $p_0$ by picking the distribution
$P\in\mathcal{P}(\mathcal{D})$ whose density $p$ minimizes a measure
of distance between $p$ and $p_0$.  Concretely, the distance measure is taken
to be the modified Fisher divergence
\begin{equation}
  \label{eq:fish-div}
  \frac{1}{2}\int_{\mathcal{D}}p_0(\boldsymbol{x})\Big\|\nabla\log
  p(\boldsymbol{x})\odot
  \boldsymbol{h}^{1/2}(\boldsymbol{x})-
  \nabla\log p_0(\boldsymbol{x})\odot \boldsymbol{h}^{1/2}(\boldsymbol{x})\Big\|_2^2\d\boldsymbol{x}.
\end{equation}
This divergence is half a weighted version of the $L_2(P_0)$ distance between the gradients of
the log-densities.  The weights are given by a function
$\boldsymbol{h}(\boldsymbol{x})=(h_1(x_1),\ldots,h_m(x_m))$, which is
a pre-selected almost surely positive function from $\mathcal{D}$ to
$\mathbb{R}_+^m$.  Introducing the weights gives flexibility to
efficiently cope with the effects of the boundary of the domain
$\mathcal{D}$.   To facilitate consistent estimation, the
modified Fisher divergence ought to be minimized if and only if
$p_0=p$ almost everywhere (a.e.). An estimator of $p_0$ is then
obtained by minimizing the loss function that results from a sample
version of the Fisher divergence. This estimator has the appealing
feature to not depend on normalizing constants as these fall out
when computing the gradient of the log-density.  Importantly, for
exponential families, the sample loss is quadratic in the canonical
parameters.

% Here
% $\boldsymbol{h}(\boldsymbol{x})=(h_1(x_1),\ldots,h_m(x_m))$ is a
% pre-selected almost surely positive function from $\mathcal{D}$ to
% $\mathbb{R}_+^m$, often with $h_j(x_j)$ chosen to be $x_j^{\alpha}$
% for some $\alpha>0$, and
% $\boldsymbol{\varphi}_{\boldsymbol{C},\mathcal{D}}=(\varphi_{1,C_1,\mathcal{D}},\ldots,\varphi_{m,C_m,\mathcal{D}})$
% with
% $\varphi_{j,C_j,\mathcal{D}}(\boldsymbol{x})\equiv\min\left\{\inf_{(y;\,\boldsymbol{x}_{-j})\in\mathcal{D}}|y-x_j|,C_j\right\}$,
% with $\boldsymbol{C}=(C_1,\ldots,C_m)$ pre-chosen constants. This
% theoretical loss is a weighted $\ell_2$ distance between the gradient
% of the true log density and the proposed one, and thus calculation of
% the normalizing constant is not needed. It is easy to show that we can
% rewrite the loss as the expectation of a function in $p_0$ only plus a
% constant depending $p$ only, thus enabling us to define the
% corresponding sample loss for estimation. For exponential family
% distributions, the sample loss is quadratic in the canonical
% parameter.

The divergence written in \eqref{eq:fish-div} involves the full
gradient of the log-density and the requirement of the divergence
being minimal if and only if $p$ is a.e.~equal to the true density
%a.e.~$\boldsymbol{x}\in\mathcal{D}$, the method above %in \citet{yu21}
only makes sense for $\mathcal{D}$ with positive Lebesgue measure in
$\mathbb{R}^m$.  Such domains were treated in \cite{yu21} but do not
include the case where $\mathcal{D}$ is the probability simplex $\Deltam$ from
\eqref{eq:delta-simplex}. %, the domain of compositional data.
In this paper, we show that the generalized score matching approach
can nevertheless be further extended to the simplex case and
compositional data.
% =\left\{\boldsymbol{x}\in\mathbb{R}_+^m\right|\left.\boldsymbol{x}\succeq\boldsymbol{0},\,\mathbf{1}_m^{\top}\boldsymbol{x}=1\right\}$,
%an important example of domains with zero Lebesgue measure in
%$\mathbb{R}^m$.
This is accomplished by profiling out
$x_m\equiv 1-\sum_{j=1}^{m-1} x_j$. Notably, we
revisit the $A^{m-1}$ models for simplex domains introduced in
\citet{ait85}
%corresponding to $a=b=0$ with
%$\mathbf{K}\mathbf{1}_m=\mathbf{K}^{\top}\mathbf{1}_m=\mathbf{0}_m$,
and show that estimation for models on the simplex can be done with
little additional effort compared to general $a$-$b$ interaction
models on domains
$\mathcal{D}$ with positive Lebesgue measure.  We also show that our
methods for models on simplex domains enjoy the same high-dimensional sample
complexity results as corresponding methods for models on domains of positive measure.

% We revisit the definition of the $a$-$b$ interaction models in
% Section~\ref{$a$-$b$ Models on Standard Simplices}, where we also detail
% how the class contains (possibly truncated)
% %
% % \ref{Pairwise Interaction Power $a$-$b$ Models},
% %we revisit the $a$-$b$ interaction models discussed above, %in \citet{yu19}, 
% %of which
% Gaussian graphical models (GGMs) and the $A^{m-1}$ models in
% \citet{ait85} as special examples. \citet{yu21} previously proved
% theoretical consistency results for high-dimensional settings
% requiring sample size $n=\Omega(\log m)$ for GGMs truncated to finite
% unions of convex sets as well as for general $a$-$b$ interaction
% models on bounded domains with positive Lebesgue measure (assuming
% here $a>0$).
% %In addition, for $a=0$ the sample complexity required is $\Omega(\log
% %m)$ times an unknown constant that may weakly depend on $m$.
% In this paper we show that these sample complexity results also hold for simplex domains. 

% % The rest of the paper is organized as follows. In Section \ref{Generalized Score Matching for General Domains} we revisit the generalized score matching estimator in \citet{yu21}, and in Section \ref{$a$-$b$ Models on Standard Simplices} we extend the method to densities on simplex domains. Theoretical results and numerical experiments are demonstrated in Sections \ref{Theory} and \ref{Numerical Experiments}, respectively. Longer proofs are included in the Appendix. An implementation that incorporates various types of domain $\mathcal{D}$ is available in the R package \texttt{genscore}.

\subsection{Domains of Positive Measure}%Formulation in \citet{yu21}}
\label{Formulation in yu21}

We now review in more detail the method to tackle general domains of
positive measure that was developed in \cite{yu21}.  Let
$\mathcal{D}\subseteq\mathbb{R}^m$ be a set of positive measure.  Our
interest is in modeling and estimating the joint distribution $P_0$ of
a random vector $\boldsymbol{X}\in\mathbb{R}^m$, when $P_0$ has
support $\mathcal{D}$.  The scenario of interest is the case where
$P_0$ has a twice continuously differentiable probability density
function $p_0$ with respect to Lebesgue measure on $\mathcal{D}$.  As
a model for $P_0$ we now consider a family $\mathcal{P}(\mathcal{D})$
of distributions with twice continuously differentiable densities on
$\mathcal{D}$.  These densities are assumed to be specified up to
unknown/difficult to compute normalizing constants.
% % If $p$ is a density function, we write $p\in\mathcal{P}(\mathcal{D})$ for ease of notation if there exists some distribution $P\in\mathcal{P}(\mathcal{D})$ that has density $p$.
% The goal is to estimate $p_0$ by picking the distribution $P$ from $\mathcal{P}(V)$ with density $p$ that minimizes some empirical loss that measures the distance between $p$ and $p_0$.
%

The approach taken proceeds in a coordinate-wise fashion.  For any
index $j=1,\ldots,m$, let
$\mathcal{C}_{j,\mathcal{D}}\left(\boldsymbol{x}_{-j}\right)\equiv\{y\in\mathbb{R}:(y;\boldsymbol{x}_{-j})\in\mathcal{D}\}$
be the $j$th section of $\mathcal{D}$ defined by the $(m-1)$-dimensional
vector $\boldsymbol{x}_{-j}$.  Furthermore, define the projection
$\mathcal{S}_{-j,\mathcal{D}}\equiv\left\{\boldsymbol{x}_{-j}:\mathcal{C}_{j,\mathcal{D}}\left(\boldsymbol{x}_{-j}\right)\neq\varnothing\right\}\subseteq\mathbb{R}^{m-1}$.
A measurable domain $\mathcal{D}$ is a \emph{component-wise countable
  union of intervals} if
 % it is measurable and if
for any $j=1,\ldots,m$ and fixed
$\boldsymbol{x}_{-j}\in\mathcal{S}_{-j,\mathcal{D}}$, the section
$\mathcal{C}_{j,\mathcal{D}}(\boldsymbol{x}_{-j})$ is a union of
finite or countably many intervals in $\mathbb{R}$. The idea to handle
effects of the boundary of $\mathcal{D}$ is to build a weight function
for the divergence in~(\ref{eq:fish-div}) by using transformed and
truncated coordinate-wise distances in the different sections.  

Concretely, let $\boldsymbol{C}\in\mathbb{R}^m$ be a choice of
truncation constants with $\boldsymbol{C}\succ\boldsymbol{0}_m$. 
 Define
 $\boldsymbol{\varphi}_{\boldsymbol{C},\mathcal{D}}=(\varphi_{1,C_1,\mathcal{D}},\ldots,\varphi_{m,C_m,\mathcal{D}})$
 where
 \begin{equation}
   \label{eq:trunc-dist}
     \varphi_{j,C_j,\mathcal{D}}(\boldsymbol{x})\equiv\min\Big\{\inf_{(y;\,\boldsymbol{x}_{-j})\in\mathcal{D}}|y-x_j|,C_j\Big\}.
   \end{equation}
   By the assumption of a component-wise countable union of intervals, the
   component $x_j$ of vector $\boldsymbol{x}$ lies in unique
   maximal subinterval of the section
   $\mathcal{C}_{j,\mathcal{D}}\left(\boldsymbol{x}_{-j}\right)$ and
   the infimum in~(\ref{eq:trunc-dist}) gives the distance between
   $x_j$ and the boundary of this interval. The minimum then truncates
   the distance from above by some $C_j>0$, in order to maintain
   bounded weights in the divergence from~(\ref{eq:fish-div}). In
   practice, %\citet{yu21}
   we suggest $\boldsymbol{C}$ to be chosen as sample quantiles.

   The second ingredient to the weights are transformations that
   allow one to adapt to the decay of densities at the
   boundary of $\mathcal{D}$.  Given
   a user-specified $\boldsymbol{h}:\mathbb{R}_+^m\to\mathbb{R}_+^m$,
   $\boldsymbol{x}\mapsto\left(h_1(x_1),\ldots,h_m(x_m)\right)^{\top}$
   with $h_1,\ldots,h_m:\mathbb{R}_+\to\mathbb{R}_+$ almost surely
   positive and absolutely continuous in every bounded sub-interval of
   $\mathbb{R}_+$, the \emph{generalized
     $(\boldsymbol{h},\mathcal{C},\mathcal{D})$-score matching loss}
   in $P\in\mathcal{P}(\mathcal{D})$ with density $p$ is defined as
   the divergence
\begin{multline}\label{def_loss}
L_{\boldsymbol{h},\boldsymbol{C},\mathcal{D}}(P)\equiv\frac{1}{2}\int_{\mathcal{D}}p_0(\boldsymbol{x})\Big\|\nabla\log
p(\boldsymbol{x})\odot
\left(\boldsymbol{h}\circ\boldsymbol{\varphi}_{\boldsymbol{C},\mathcal{D}}\right)^{1/2}(\boldsymbol{x})-\\
\nabla\log p_0(\boldsymbol{x})\odot \left(\boldsymbol{h}\circ\boldsymbol{\varphi}_{\boldsymbol{C},\mathcal{D}}\right)^{1/2}(\boldsymbol{x})\Big\|_2^2\d\boldsymbol{x}.
\end{multline}
This population loss is minimized at $p$ if and only if $p=p_0$ almost
surely.  In order to derive a practical sample loss, we make the following assumptions:
\begin{enumerate}[label=(A.\arabic*),leftmargin=30pt]
\item $p_0(x_j;\boldsymbol{x}_{-j})h_j(\varphi_{C_j,\mathcal{D},j}(\boldsymbol{x}))\partial_j\log p(x_j;\boldsymbol{x}_{-j})\left|_{x_j \searrow a_{k}(\boldsymbol{x}_{-j})^+}^{x_j\nearrow b_{k}(\boldsymbol{x}_{-j})^-}\right.=0$\\ for all $k=1,\ldots,K_j(\boldsymbol{x}_{-j})$ and  $\boldsymbol{x}_{-j}\in\mathcal{S}_{-j,\mathcal{D}}$ for all $j$;
\item $\int_{\mathcal{D}}p_0(\boldsymbol{x})\left\|\nabla\log
    p(\boldsymbol{x})\odot
    (\boldsymbol{h}\circ\boldsymbol{\varphi_{\boldsymbol{C},\mathcal{D}}})^{1/2}(\boldsymbol{x})\right\|_2^2\d\boldsymbol{x}<+\infty$,
  \\
$\int_{\mathcal{D}}p_0(\boldsymbol{x})\left\|\left[\nabla\log p(\boldsymbol{x})\odot(\boldsymbol{h}\circ\boldsymbol{\varphi_{\boldsymbol{C},\mathcal{D}}})(\boldsymbol{x})\right]'\right\|_1\d\boldsymbol{x}<+\infty$.
\item $\forall j=1,\ldots,m$ and a.e.~$\boldsymbol{x}_{-j}\in\mathcal{S}_{-j,\mathcal{D}}$, the component function $h_j$ of $\boldsymbol{h}$ is absolutely continuous in any bounded sub-interval of the section $\mathcal{C}_{j,\mathcal{D}}(\boldsymbol{x}_{-j})$.
\end{enumerate}
Under the mild assumptions (A.1)--(A.3), one can show that
\begin{multline}\label{eq_equivalent_loss}
L_{\boldsymbol{h},\boldsymbol{C},\mathcal{D}}(P)\equiv\frac{1}{2}\sum_{j=1}^m\int_{\mathcal{D}}p_0(\boldsymbol{x})\cdot(h_j\circ\varphi_{C_j,\mathcal{D},j})(\boldsymbol{x})\cdot\left[\partial_j\log p(\boldsymbol{x})\right]^2\d\boldsymbol{x}\\
+\sum_{j=1}^m\int_{\mathcal{D}}p_0(\boldsymbol{x})\cdot\partial_j \left[(h_j\circ\varphi_{C_j,\mathcal{D},j})(\boldsymbol{x})\cdot\partial_j\log p(\boldsymbol{x})\right]\d\boldsymbol{x}
\end{multline}
plus a constant depending on $p_0$ only (so, independent of $p$). This
facilitates consistent estimation using the empirical loss 
\begin{multline}\label{eq_empirical_loss}
\hat{L}_{\boldsymbol{h},\boldsymbol{C},\mathcal{D}}(P)=\frac{1}{2}\sum_{j=1}^m\sum_{i=1}^n\frac{1}{2}(h_j\circ\varphi_{C_j,\mathcal{D},j})\left(\boldsymbol{X}^{(i)}\right)\cdot\left[\partial_j\log p\left(\boldsymbol{X}^{(i)}\right)\right]^2
+\\\
\partial_j \left[(h_j\circ\varphi_{C_j,\mathcal{D},j})\left(\boldsymbol{X}^{(i)}\right)\cdot\partial_j\log p\left(\boldsymbol{X}^{(i)}\right)\right],
\end{multline}
where $\boldsymbol{X}^{(i)}$, $1\le i\le n$, form an i.i.d.~sample
from $P_0$.

\subsection{Extension to Simplices}

As noted in Section~\ref{sec:gener-score-match}, the generalized score
matching method from \cite{yu21} is applicable only to domains
$\mathcal{D}$ with positive Lebesgue measure in
$\mathbb{R}^m$. However, for some domains that are null sets of
dimension $k<m$, this issue may be resolved by transforming
$\mathcal{D}$ to a full-dimensional subset of $\mathbb{R}^k$.  This is
particularly easy and tractable for the important case of the
probability simplex
$\Deltam\equiv\{\boldsymbol{x}\in\mathbb{R}^m:\boldsymbol{x}\succeq\boldsymbol{0},\,\mathbf{1}_m^{\top}\boldsymbol{x}=1\}$.
Indeed, we may drop the last coordinate $x_m$ (or any other coordinate
as we discuss later; here we choose $m$ for simplicity), substituting
it with $1-\mathbf{1}_{m-1}^{\top}\boldsymbol{x}_{-m}$, and work
instead with the full-dimensional simplex
\begin{equation}
  \label{eq:Delta-m}
  \Deltam_{-m}\equiv
  \left\{\boldsymbol{x}_{-m}\in\mathbb{R}^{m-1}:\boldsymbol{x}_{-m}\succeq\boldsymbol{0},\,\mathbf{1}_{m-1}^{\top}\boldsymbol{x}_{-m}\le
    1\right\}
\end{equation}
in $\mathbb{R}^{m-1}$.  Throughout the rest of the paper we thus
consider the case of domain $\mathcal{D}\equiv\Delta_{-m}$, and in our
notation we remove the dependency of $L$ and $\boldsymbol{\varphi}$ on
$\mathcal{D}$.

Let $\boldsymbol{x}\in\mathbb{R}^m$, and let $j\in\{1,\ldots,m-1\}$.
Write $\boldsymbol{x}_{-\{j,m\}}$ for the vector in $\mathbb{R}^{m-2}$
obtained by removing $x_m$ and $x_j$.  Then $\Deltam_{-m}$ has 
$j$th section
\[
  \mathcal{C}_j(\boldsymbol{x}_{-m})=\left[0,1-\mathbf{1}^{\top}_{m-2}\boldsymbol{x}_{-\{j,m\}}\right].
\]
Hence, we have the coordinate-wise distance
\[
  \varphi_{C_j,j}(\boldsymbol{x})=\min\left\{C_j,x_j,1-\mathbf{1}_{m-1}^{\top}\boldsymbol{x}_{-m}\right\}=\min\left\{C_j,x_j,x_m\right\}.
\]
The role of the truncation constants $C_j$ is to ensure boundedness of
the coordinate-wise distances.  As the simplex is naturally bounded by
the unit cube, and it is natural to not use any truncation here
and simply use the coordinate-wise distance
\[
  \varphi_{j}(\boldsymbol{x})=\min\left\{x_j,1-\mathbf{1}_{m-1}^{\top}\boldsymbol{x}_{-m}\right\}=\min\left\{x_j,x_m\right\}.
\]
A plot of $\varphi_{j}$ in the case of $m=3$ is given in Figure
\ref{plot_phi}.

\begin{figure}[!t]
\centering
%\vspace{-0.7in}
\subfloat[$\varphi_{1}$]
{\includegraphics[scale=0.4]{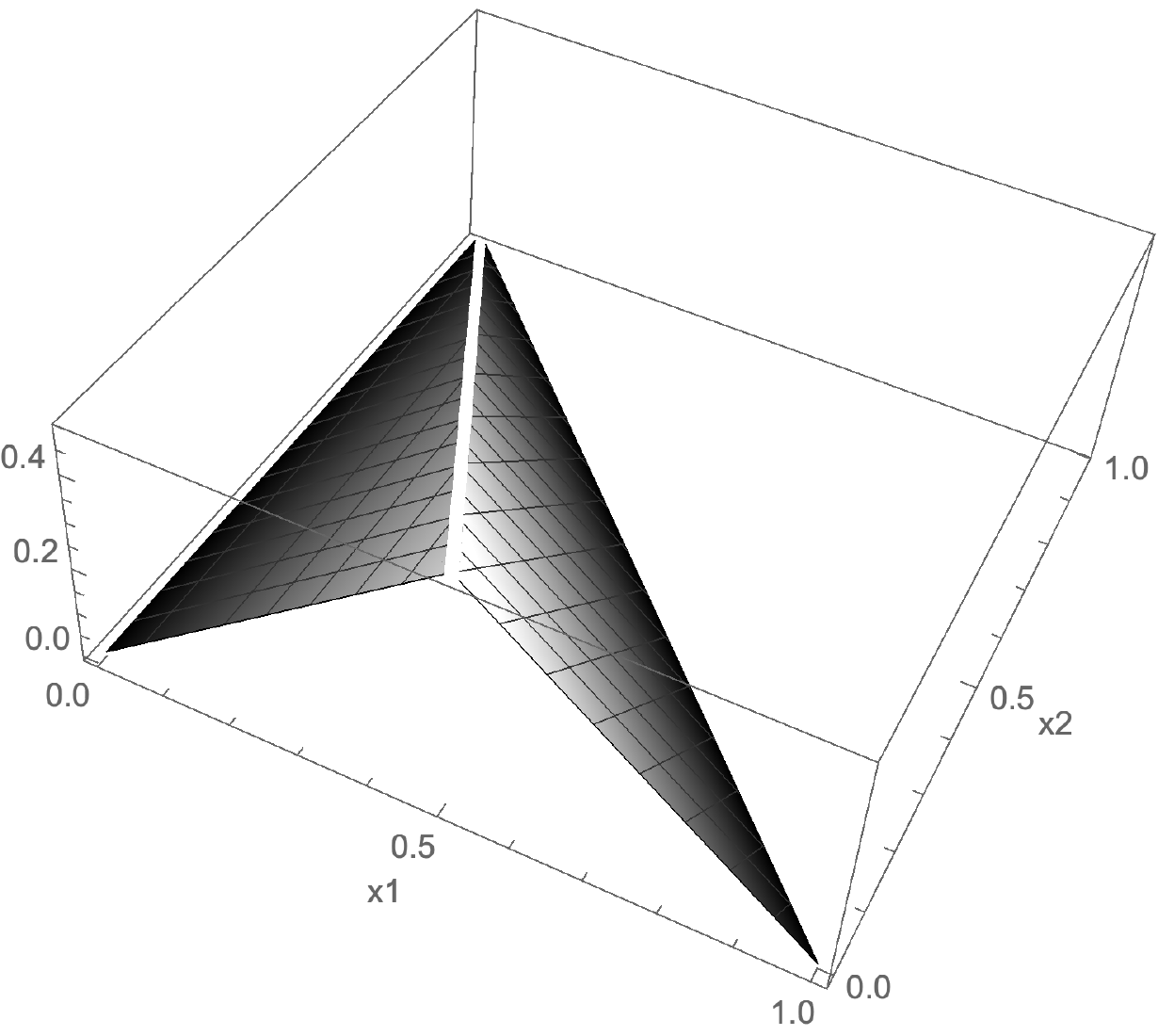}\hspace{-0.02in}}\hspace{0.2in}
\subfloat[$\varphi_{2}$]
{\includegraphics[scale=0.4]{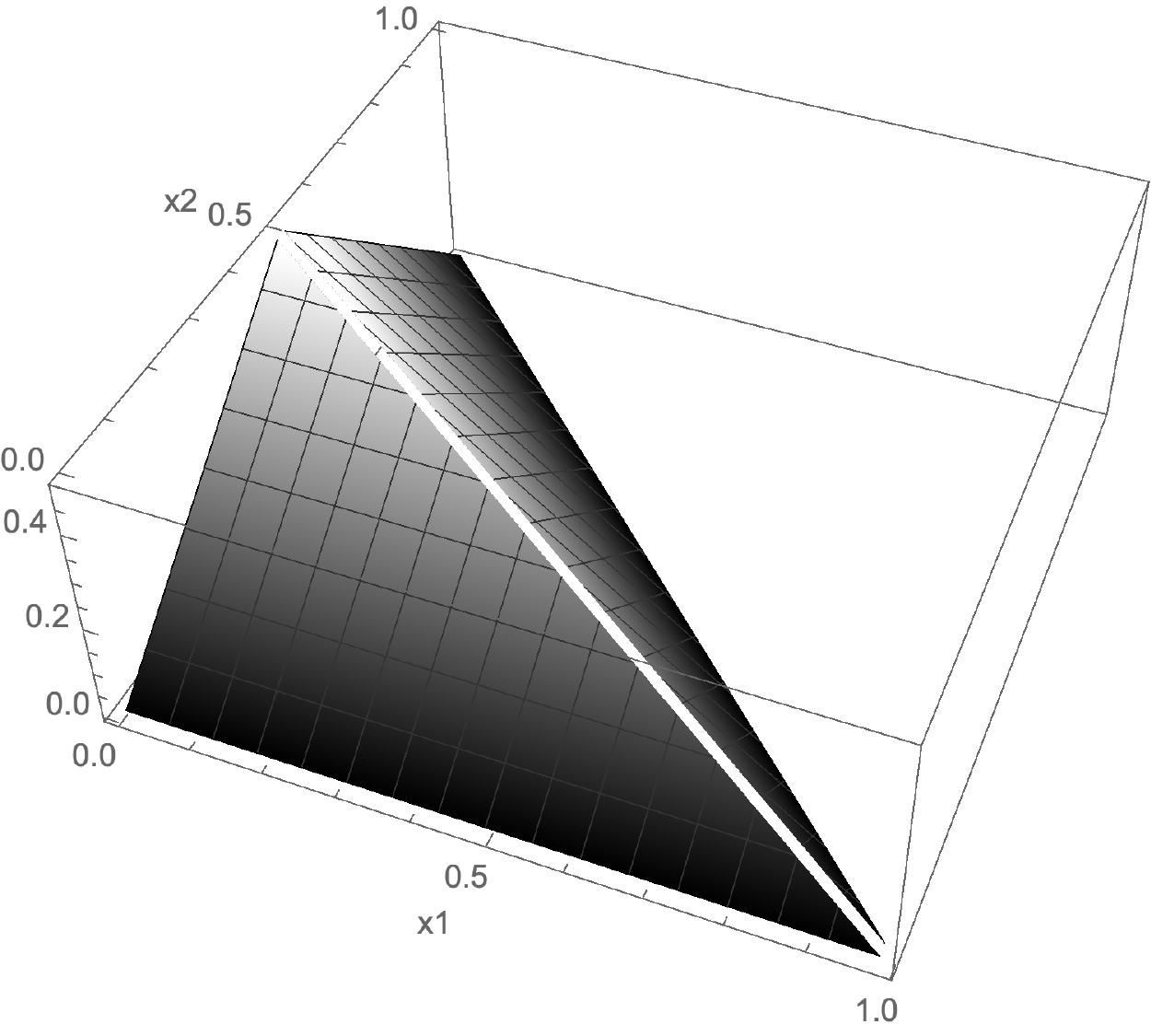}\hspace{-0.02in}}
\hspace{0.2in}
\subfloat[$\varphi_{1}\vee\varphi_{2}$]{\includegraphics[scale=0.4]{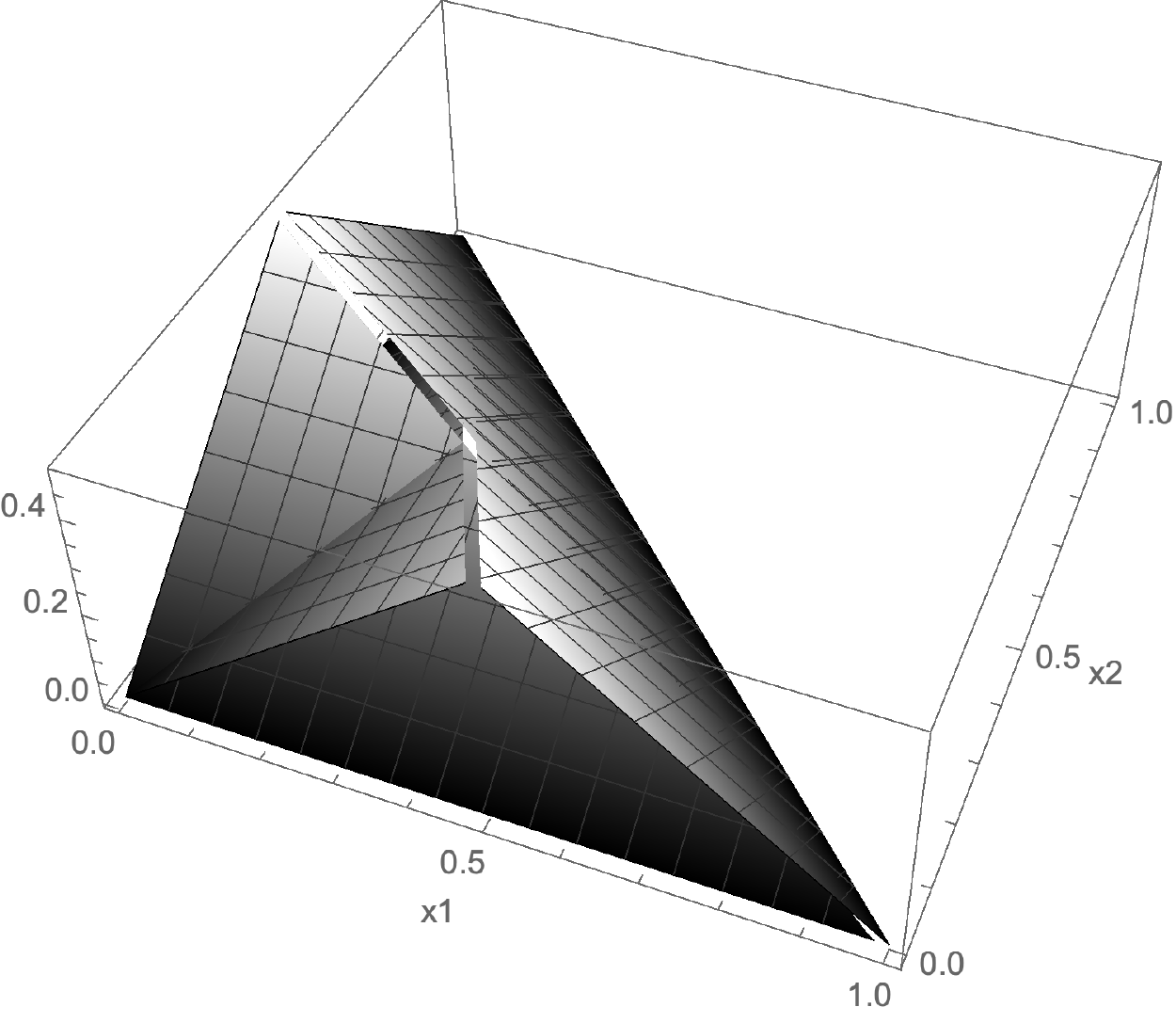}\hspace{-0.02in}}
\caption{Plots of $\varphi_{1}$ and $\varphi_{2}$ for $\Delta_{-3}\equiv\{(x_1,x_2):x_1\geq 0,x_2\geq 0,x_1+x_2\leq 1\}$.}
\label{plot_phi}
\end{figure}

\subsection{Exponential Families and Regularization}\label{Exponential Families and Regularized Generalized Score Matching}

Let
$\mathcal{P}\equiv\{P_{\boldsymbol{\theta}}:\boldsymbol{\theta}\in\boldsymbol{\Theta}\subset\mathbb{R}^r\}$
be an $r$-dimensional exponential family, in which the distributions
$P_{\boldsymbol{\theta}}$ are continuous with support $\Delta$.  Let
$\boldsymbol{\theta}\in\mathbb{R}^r$ be the canonical parameters, and
let the distributions have densities of the form
\[
  \log
  p_{\boldsymbol{\theta}}(\boldsymbol{x})=\boldsymbol{\theta}^{\top}\boldsymbol{t}(\boldsymbol{x})-\psi(\boldsymbol{\theta})+b(\boldsymbol{x}), \quad
  \boldsymbol{x}\in\Deltam.
\]
It is not difficult to show that in this case the empirical loss
$\hat{L}_{\boldsymbol{h},\boldsymbol{C}}$ from~(\ref{eq_empirical_loss}) can be written as a quadratic function in the
canonical parameter, i.e.,
\begin{equation}
  \label{eq:exp-fam-quadratic}
  \hat{L}_{\boldsymbol{h},\boldsymbol{C}}(p_{\boldsymbol{\theta}})=\frac{1}{2}\boldsymbol{\theta}^{\top}\boldsymbol{\Gamma}(\mathbf{x})\boldsymbol{\theta}-\boldsymbol{g}(\mathbf{x})^{\top}\boldsymbol{\theta}+\mathrm{const.}
\end{equation}
with $\boldsymbol{\Gamma}(\mathbf{x})\in\mathbb{R}^{r\times r}$ and
$\boldsymbol{g}(\mathbf{x})\in\mathbb{R}^r$ sample averages of known
functions in $\mathbf{x}$ only.  We detail the derivation
of~(\ref{eq:exp-fam-quadratic}) for our specific models of interest in
Section \ref{Estimation-Simplex}, where we then give explicit expressions for
$\boldsymbol{\Gamma}$ and $\boldsymbol{g}$.

In high-dimensional settings where the number of parameters $r$ is
large compared to the sample size $n$, we add an $\ell_1$
regularization on $\boldsymbol{\theta}$ and consider the \emph{regularized generalized
  score matching loss} 
\begin{equation}\label{eq_loss_regularized}
\hat{L}_{\boldsymbol{h},\boldsymbol{C},\lambda,\delta}(p_{\boldsymbol{\theta}})\equiv\frac{1}{2}\boldsymbol{\theta}^{\top}\boldsymbol{\Gamma}_{\delta}(\mathbf{x})\boldsymbol{\theta}-\boldsymbol{g}(\mathbf{x})^{\top}\boldsymbol{\theta}+\lambda\|\boldsymbol{\theta}\|_1.
\end{equation}
Here $\boldsymbol{\Gamma}_{\delta}(\mathbf{x})$ equals
$\boldsymbol{\Gamma}(\mathbf{x})$ except that the diagonal entries are
multiplied by $\delta>1$.  The \emph{diagonal multiplier} $\delta$ is
introduced to avoid possible unboundedness of the loss when
$\boldsymbol{\Gamma}(\mathbf{x})$ is singular (due to high dimension)
and the regularization parameter $\lambda$ is small; cf.~Section 4 of
\citet{yu19}.  We minimize the regularized
loss from~(\ref{eq_loss_regularized}) using coordinate-descent.
% ; see Section 5.3 of \citet{yu19}.
In some problems, it is natural to exclude a subset of the
parameters from the regularization term.  In particular, in our later
estimation of an interaction matrix $\mathbf{K}$, we typically
only penalize its off-diagonal entries,
i.e., the penalty is~$\lambda\|\mathrm{vec}(\mathbf{K}_{\mathrm{off}})\|_1$.

\subsection{Dependency on the Removed Coordinate}\label{Dependency on
  the Removed Coordinate}

The loss
function in (\ref{eq_loss_regularized}) depends on the choice of the removed
coordinate $x_m$ through $\boldsymbol{\Gamma}_{\delta}(\mathbf{x})$ and
$\boldsymbol{g}(\mathbf{x})$ and is no longer symmetric in all
coordinates. One way to mitigate the asymmetry is to calculate the
loss $\hat{L}_j$ for each removed coordinate $x_j$ and then optimize
the average loss, which is symmetric in $j=1,\dots,m$.  In
high dimensions, averaging over all $j$ is computationally expensive, but we may
nevertheless average the losses $\hat{L}_j$ over a set of (say 10
randomly chosen) coordinates $\mathcal{J}\subseteq\{1,\dots,m\}$.
Importantly, the averaged loss is still a quadratic form as in
(\ref{eq_loss_regularized}), just with
$\boldsymbol{\Gamma}_{\delta}(\mathbf{x})$ and
$\boldsymbol{g}(\mathbf{x})$ replaced by averages.
% of functions in $\mathbf{x}$.
% Choosing $\mathcal{J}=\{1,\ldots,m\}$ the resulting loss
% would be symmetric in all $j=1,\ldots,m$.

\section{Power Interaction Models on the Probability Simplex}\label{$a$-$b$ Models on Standard Simplices}

%\subsection{Pairwise Interaction Power $a$-$b$ Models}\label{Pairwise
%Interaction Power $a$-$b$ Models}

Our interest is in models for compositional data that flexibly
accommodate interactions. To
this end, we propose the following class of power interaction models
in the form of exponential families:
\begin{equation}\label{eq_interaction_density2}
p_{\boldsymbol{\eta},\mathbf{K}}(\boldsymbol{x})\;\propto\;\exp\left(-\frac{1}{2a}{\boldsymbol{x}^a}^{\top}\mathbf{K}\boldsymbol{x}^a+\frac{1}{b}\boldsymbol{\eta}^{\top}\boldsymbol{x}^b\right)\mathds{1}_{\Deltam}(\boldsymbol{x}),
\end{equation}
where $a\geq 0$ and $b\geq 0$ are known constants and the interaction
matrix $\mathbf{K}\in\mathbb{R}^{m\times m}$ and the vector
$\boldsymbol{\eta}\in\mathbb{R}^m$ are unknown parameters of interest.
If $a=0$ or $b=0$, we take
$\boldsymbol{x}^0\equiv\log\boldsymbol{x}$ and $1/0\equiv1$.  We write
$\mathbf{K}_0$ and $\boldsymbol{\eta}_0$ for the parameters of a true data-generating distribution in the model
from~(\ref{eq_interaction_density2}).

When $a=b=0$ with
$\mathbf{K}\mathbf{1}_m=\mathbf{K}^{\top}\mathbf{1}_m=\mathbf{0}_m$,
the model from~(\ref{eq_interaction_density2}) specializes to the
$A^{m-1}$ model introduced in \citet{ait85}, on which we elaborate in
Section~\ref{sec_Ad_Models}.  The models determined by
(\ref{eq_interaction_density2}) are also compositional versions of
graphical models on $\mathbb{R}^m$ or the orthant $\mathbb{R}^m_+$
such as (truncated) Gaussian models ($a=b=1$), the exponential
square-root model from \citet{ino16} ($a=b=1/2$), or the gamma model
from \cite{yu19}.  However, on the simplex domain,
sparsity in $\mathbf{K}$ does not directly relate to conditional
independences among the components of the compositional observation. Indeed,
any two components are perfectly correlated when all others are fixed,
due to the sum-to-one constraint.  

Reasoning as for \refyu{Theorem 4.1} in \citet{yu21}, we obtain the
following conditions on $a$ and $b$ for the density to be proper in
the simplex case.  Refined conditions for the case $a=b=0$ are obtained in Section \ref{sec_Ad_Models}. 

\begin{theorem}[Finite normalizing constant]\label{thm_norm_const}
If one of the following conditions holds, then the right-hand side of
\eqref{eq_interaction_density2} is integrable over $\Deltam$ and
defines a proper density:

\begin{enumerate}[label=(CC\arabic*),leftmargin=35pt]
\item $a>0$, $b>0$;
\item $a>0$, $b=0$, $\eta_j>-1$ for all $j$;
\item $a=0$, $b=0$, $\log(\boldsymbol{x})^{\top}\mathbf{K}\log(\boldsymbol{x})>0$ $\forall \boldsymbol{x}\in\Deltam$;
\item $a=0$, $b>0$, $\log(\boldsymbol{x})^{\top}\mathbf{K}\log(\boldsymbol{x})\geq 0$ $\forall\boldsymbol{x}\in\Deltam$.
\end{enumerate}
In the case where $\boldsymbol{\eta}=\boldsymbol{0}$ is known, the conditions on $b$ and $\boldsymbol{\eta}$ can be ignored.
\end{theorem}

% However, one may still be interested in the pattern as a measure of
% how two components interact, as well as estimation of $\mathbf{K}$ and
% $\boldsymbol{\eta}$ with minimized error for the sole purpose of
% recovering the density itself.
As a prerequisite for our subsequent discussion of estimation, the
following theorem gives the conditions for the identifiability of
$\mathbf{K}$ and $\boldsymbol{\eta}$ from a given $a$-$b$ density on
the simplex. In particular, the parameters are identifiable if $a\neq 1$
and $2a=b> 0$
does not hold.  The theorem is proven in the appendix.
% ; in these cases, the goal of recovering
% the underlying true parameters $\mathbf{K}_0$ and
% $\boldsymbol{\eta}_0$ is well-defined. The proof is in the appendix.

\begin{theorem}[Identifiability]
  \label{thm_simplex_identifiability}
Suppose there exist $\mathbf{K}_1$, $\mathbf{K}_2$, $\boldsymbol{\eta}_1$, $\boldsymbol{\eta}_2$ such that
\[\exp\left(-\frac{1}{2a}{\boldsymbol{x}^a}^{\top}\mathbf{K}_1\boldsymbol{x}^a+\frac{1}{b}\boldsymbol{\eta}_1^{\top}\boldsymbol{x}^b\right)=\exp\left(-\frac{1}{2a}{\boldsymbol{x}^a}^{\top}\mathbf{K}_2\boldsymbol{x}^a+\frac{1}{b}\boldsymbol{\eta}_2^{\top}\boldsymbol{x}^b\right)\]
for all $\boldsymbol{x}\in\Deltam$, where $\boldsymbol{x}^0\equiv\log(\boldsymbol{x})$ and $1/0\equiv 1$. Then $\mathbf{K}_1=\mathbf{K}_2$ and $\boldsymbol{\eta}_1=\boldsymbol{\eta}_2$, or else one of the following must hold: (I) $a=b=1$, (II) $a=1$, $b=2$, (III) $a=1$ and $\boldsymbol{\eta}_1=\boldsymbol{\eta}_2$, (IV) $2a=b> 0$ and $\mathbf{K}_1-\mathbf{K}_2=2\boldsymbol{\eta}_1-2\boldsymbol{\eta}_2$.
\end{theorem}

Our approach to estimation by score matching is to profile out the last component of $\boldsymbol{x}$ using
$x_m=1-\sum_{j=1}^{m-1} x_j$.   The density
in \eqref{eq_interaction_density2} becomes
% \begin{multline}\label{eq_density_simplex}
% p_{\boldsymbol{\eta},\mathbf{K}}(\boldsymbol{x}_{-m})\propto\exp\left[-\frac{1}{2a}{\boldsymbol{x}_{-m}^a}^{\top}\mathbf{K}_{-m,-m}\boldsymbol{x}_{-m}^a-\frac{1}{a}{\boldsymbol{x}_{-m}^a}^{\top}\boldsymbol{\kappa}_{-m,m}\left(1-\mathbf{1}_{m-1}^{\top}\boldsymbol{x}_{-m}\right)^a\right.\\
% \left.-\frac{1}{2a}\kappa_{m,m}\left(1-\mathbf{1}_{m-1}^{\top}\boldsymbol{x}_{-m}\right)^{2a}+\frac{1}{b}\boldsymbol{\eta}_{-m}^{\top}\boldsymbol{x}_{-m}^b+\frac{\eta_m}{b}\left(1-\mathbf{1}_{m-1}^{\top}\boldsymbol{x}_{-m}\right)^b\right]
% \end{multline}
\begin{multline}\label{eq_density_simplex}
p_{\boldsymbol{\eta},\mathbf{K}}(\boldsymbol{x}_{-m})\propto\exp\left[-\frac{1}{2a}{\boldsymbol{x}_{-m}^a}^{\top}\mathbf{K}_{-m,-m}\boldsymbol{x}_{-m}^a-\frac{1}{a}{\boldsymbol{x}_{-m}^a}^{\top}\boldsymbol{\kappa}_{-m,m}\Big(1-\sum_{j=1}^{m-1} x_j\Big)^a\right.\\
\left.-\frac{1}{2a}\kappa_{m,m}\Big(1-\sum_{j=1}^{m-1} x_j\Big)^{2a}+\frac{1}{b}\boldsymbol{\eta}_{-m}^{\top}\boldsymbol{x}_{-m}^b+\frac{\eta_m}{b}\Big(1-\sum_{j=1}^{m-1} x_j\Big)^b\right]
\end{multline}
on $\Deltam_{-m}\subseteq\mathbb{R}^{m-1}$.
%, where again by $\boldsymbol{x}_{-m}^{0}$ we mean $\log(\boldsymbol{x}_{-m})$, and $1/0\equiv 1$. 
The next theorem gives sufficient conditions for the
assumptions (A.1)--(A.3) in Section~\ref{Formulation in yu21} to
hold. Under these assumptions, the generalized score matching loss
from \eqref{def_loss} has the equivalent form in
\eqref{eq_equivalent_loss}, and the empirical loss stated in
\eqref{eq_empirical_loss} is valid. The theorem is proven in the appendix.

\begin{theorem}[Assumptions for score matching]\label{thm_A_simp}
Suppose one of (CC1) through (CC4) holds, and $\boldsymbol{h}(\boldsymbol{x})=\left(x_1^{\alpha_1},\ldots,x_m^{\alpha_m}\right)$, where
\begin{enumerate}[(I)]
\item if $a>0$ and $b>0$, $\alpha_j>\max\{0,1-a,1-b\}$;
\item if $a>0$ and $b=0$, $\alpha_j>1-\eta_{0,j}$;
\item if $a=0$, $\alpha_j\geq 0$.
\end{enumerate}
Then conditions (A.1)--(A.3) in Section \ref{Formulation in yu21} are
satisfied.
%, which allows us to rewrite Equation \ref{def_loss} as Equation
%\ref{eq_equivalent_loss}.
In the case with $\boldsymbol{\eta}\equiv\boldsymbol{0}$ known, it suffices to have $a>0$ and $\alpha_j>\max\{0,1-a\}$, or $a=0$ and $\alpha_j\geq 0$.
\end{theorem}

\section{Estimation for Power Interaction Models}
\label{sec:estim-power-inter}

\subsection{Estimation for General $a$ and $b$}\label{Estimation-Simplex}
%We again write
%$\boldsymbol{\Psi}\equiv\left(\mathbf{K}^{\top},\boldsymbol{\eta}\right)^{\top}\in\mathbb{R}^{(m+1)\times
%m}$.

In the equations in this section, the case $a=0$ (and similarly for
$b=0$) is covered by the following convention.  If $a=0$, substitute
coefficients ``$a$'' with ``1'' and coefficients ``$(a-1)$'' with
``$-1$''.  As before, we let $x^a\equiv\log x$ if $a=0$.

For notational simplicity, we consider again dropping the last
coordinate $x_m$; the results for dropping other coordinates $x_j$'s
would be analogous.  Having substituted
$x_m\equiv 1-\mathbf{1}_{m-1}^{\top}\boldsymbol{x}_{-m}$ and working
on $\Deltam_{-m}$, the partial derivative
$\partial_j\log p(\boldsymbol{x}_{-m})$ of the density
$p(\boldsymbol{x}_{-m})\equiv p_{\boldsymbol{\eta},\mathbf{K}}(\boldsymbol{x}_{-m})$
in~(\ref{eq_density_simplex}) now depends on both
$(\boldsymbol{\kappa}_{,j},\eta_j)$ and
$(\boldsymbol{\kappa}_{,m},\eta_m)$. Thus, unlike in the case of
$a$-$b$ models on domains with positive Lebesgue
measure, % in \citet{yu21}
the subvectors $(\boldsymbol{\kappa}_{,j},\eta_j)$ and
$(\boldsymbol{\kappa}_{,m},\eta_m)$ are no longer isolated in the
score-matching loss.  Instead, we have
\begin{align}
\partial_j&\log p(\boldsymbol{x}_{-m})=-\left(\boldsymbol{\kappa}_{,j}^{\top}\boldsymbol{x}^a\right)x_j^{a-1}+\left(\boldsymbol{\kappa}_{,m}^{\top}\boldsymbol{x}^a\right)x_m^{a-1}+\eta_j x_j^{b-1}-\eta_m x_m^{b-1},\label{eq_grad_logp_simplex}\\
\partial_{jj}&\log p(\boldsymbol{x}_{-m})=-(a-1)\left[\left(\boldsymbol{\kappa}_{,j}^{\top}\boldsymbol{x}^a\right)x_j^{a-2}+\left(\boldsymbol{\kappa}_{,m}^{\top}\boldsymbol{x}^a\right)x_m^{a-2}\right]\\
&-a\left[\kappa_{jj}x_j^{2a-2}+\kappa_{mm}x_m^{2a-2}+2\kappa_{jm}x_j^{a-1}x_m^{a-1}\right]+(b-1)\left[\eta_j x_j^{b-2}+\eta_m x_m^{b-2}\right].\nonumber
\end{align}
These derivatives yield the penalized loss, which we may write as
\begin{equation}\label{penalized_loss}
\frac{1}{2}\left(\mathrm{vec}(\mathbf{K}),\,\boldsymbol{\eta}\right)^{\top}\boldsymbol{\Gamma}\left(\mathrm{vec}(\mathbf{K}),\,\boldsymbol{\eta}\right)-\boldsymbol{g}^{\top}\left(\mathrm{vec}(\mathbf{K}),\,\boldsymbol{\eta}\right)+\lambda_{\mathbf{K}}\|\mathrm{vec}(\mathbf{K}_{\mathrm{off}})\|_1+\lambda_{\boldsymbol{\eta}}\|\boldsymbol{\eta}\|_1,
\end{equation}
with matrix $\boldsymbol{\Gamma}$ and vector $\boldsymbol{g}$ naturally partitioned as
\begin{align}
  \label{eq:Gamma}
\boldsymbol{\Gamma}\equiv\begin{bmatrix}\boldsymbol{\Gamma}_{\mathbf{K}} & \boldsymbol{\Gamma}_{\mathbf{K},\boldsymbol{\eta}} \\ \boldsymbol{\Gamma}_{\mathbf{K},\boldsymbol{\eta}}^{\top} & \boldsymbol{\Gamma}_{\boldsymbol{\eta}}\end{bmatrix}\in\mathbb{R}^{(m^2+m)\times (m^2+m)},
\quad\boldsymbol{g}\equiv\left(\mathrm{vec}(\mathbf{g}_{\mathbf{K}}),\,\boldsymbol{g}_{\boldsymbol{\eta}}\right)\in\mathbb{R}^{m^2+m},
\end{align}
where the respective blocks are
$\boldsymbol{\Gamma}_{\mathbf{K}}\in\mathbb{R}^{m^2\times m^2}$,
$\boldsymbol{\Gamma}_{\mathbf{K},\boldsymbol{\eta}}\in\mathbb{R}^{m^2\times
  m}$, $\boldsymbol{\Gamma}_{\boldsymbol{\eta}}\in\mathbb{R}^{m\times
  m}$, and $\mathbf{g}_{\mathbf{K}}\in\mathbb{R}^{m\times m}$, $\boldsymbol{g}_{\boldsymbol{\eta}}\in\mathbb{R}^m$.
The estimates $\hat{\mathbf{K}}$ and $\hat{\boldsymbol{\eta}}$ that
minimize (\ref{penalized_loss}) are our estimates for $\mathbf{K}$ and
$\boldsymbol{\eta}$.

We note that in comparison to the full-dimensional setting of
\cite{yu19,yu21}, the blocks
$\boldsymbol{\Gamma}_{\mathbf{K}}\in\mathbb{R}^{m^2\times m^2}$,
$\boldsymbol{\Gamma}_{\mathbf{K},\boldsymbol{\eta}}\in\mathbb{R}^{m(m+1)}$
and $\boldsymbol{\Gamma}_{\boldsymbol{\eta}}\in\mathbb{R}^{m\times m}$
are no longer block-diagonal with $m$ blocks, due to the substitution
of $x_m$. Instead,
% each of them can be partitioned into $m\times m$ blocks with $m$ ``block-columns'' and $m$ ``block-rows'', and each matrix not  only has $m$ blocks on the diagonal, but also have the entire right-most ``block-column'' and bottom-most ``block-row'' filled, representing the interaction between the $(\boldsymbol{\kappa}_{,j},\eta_j)$ and $(\boldsymbol{\kappa}_{,m},\eta_m)$. In particular,
\begin{align*}
\boldsymbol{\Gamma}_{\mathbf{K}}&\equiv
\begin{bmatrix}
\boldsymbol{\Gamma}_{\mathbf{K},1} & \boldsymbol{0} & \cdots & \boldsymbol{0} & \boldsymbol{\Gamma}_{\mathbf{K},(1,m)} \\ 
\boldsymbol{0} & \boldsymbol{\Gamma}_{\mathbf{K},2} & \cdots & \boldsymbol{0} & \boldsymbol{\Gamma}_{\mathbf{K},(2,m)} \\ 
\vdots & \vdots & \ddots & \vdots & \vdots \\ 
\boldsymbol{0} & \boldsymbol{0} & \cdots & \boldsymbol{\Gamma}_{\mathbf{K},m-1} & \boldsymbol{\Gamma}_{\mathbf{K},(m-1,m)}  \\
\boldsymbol{\Gamma}_{\mathbf{K},(1,m)}^{\top} & \boldsymbol{\Gamma}_{\mathbf{K},(2,m)}^{\top} & \cdots & \boldsymbol{\Gamma}_{\mathbf{K},(m-1,m)}^{\top} & \boldsymbol{\Gamma}_{\mathbf{K},m}
\end{bmatrix}
                                                                                                                                                           \in\mathbb{R}^{m^2\times m^2},\\
  \intertext{with each block of size $m\times m$, and}
\boldsymbol{\Gamma}_{\mathbf{K},\boldsymbol{\eta}}&\equiv
\begin{bmatrix}
\boldsymbol{\gamma}_{\mathbf{K},\boldsymbol{\eta},1} & \boldsymbol{0} & \cdots & \boldsymbol{0} & \boldsymbol{\gamma}_{\mathbf{K},\boldsymbol{\eta},(1,m)} \\ 
\boldsymbol{0} & \boldsymbol{\gamma}_{\mathbf{K},\boldsymbol{\eta},2} & \cdots & \boldsymbol{0} & \boldsymbol{\gamma}_{\mathbf{K},\boldsymbol{\eta},(2,m)} \\ 
\vdots & \vdots & \ddots & \vdots & \vdots \\ 
\boldsymbol{0} & \boldsymbol{0} & \cdots & \boldsymbol{\gamma}_{\mathbf{K},\boldsymbol{\eta},m-1} & \boldsymbol{\gamma}_{\mathbf{K},\boldsymbol{\eta},(m-1,m)}  \\
\boldsymbol{\gamma}_{\mathbf{K},\boldsymbol{\eta},(m,1)} & \boldsymbol{\gamma}_{\mathbf{K},\boldsymbol{\eta},(m,2)} & \cdots & \boldsymbol{\gamma}_{\mathbf{K},\boldsymbol{\eta},(m,m-1)} & \boldsymbol{\gamma}_{\mathbf{K},\boldsymbol{\eta},m}
\end{bmatrix}
                                                                                                                                                                                            \in\mathbb{R}^{m^2\times m},\\
  \intertext{with each block a vector of size $m$, and}
\boldsymbol{\Gamma}_{\boldsymbol{\eta}}&\equiv
\begin{bmatrix}
\gamma_{\boldsymbol{\eta},1} & 0 & \cdots & 0 & \gamma_{\boldsymbol{\eta},(1,m)} \\ 
0 & \gamma_{\boldsymbol{\eta},2} & \cdots & 0 & \gamma_{\boldsymbol{\eta},(2,m)} \\ 
\vdots & \vdots & \ddots & \vdots & \vdots \\ 
0 & 0 & \cdots & \gamma_{\boldsymbol{\eta},m-1} & \gamma_{\boldsymbol{\eta},(m-1,m)}  \\
\gamma_{\boldsymbol{\eta},(1,m)} & \gamma_{\boldsymbol{\eta},(2,m)} & \cdots & \gamma_{\boldsymbol{\eta},(m-1,m)} & \gamma_{\boldsymbol{\eta},m}
\end{bmatrix}
                                                                                                                    \in\mathbb{R}^{m\times m}.
                                                                                                                    %\quad\text{each block }\in\mathbb{R},
\end{align*}
The specific form of the blocks appearing in the preceding displays
is as follows.  Using the shorthand $\tilde h_j\equiv
h_j\circ\varphi_j$, we have for $j=1,\ldots,m-1$,
\begin{align*}
\boldsymbol{\Gamma}_{j} &\equiv \begin{bmatrix}
\boldsymbol{\Gamma}_{\mathbf{K},j} & \boldsymbol{\gamma}_{\mathbf{K},\boldsymbol{\eta},j} \\ \boldsymbol{\gamma}_{\mathbf{K},\boldsymbol{\eta},j}^{\top} & \gamma_{\boldsymbol{\eta},j}
\end{bmatrix}\\
  &\equiv\frac{1}{n}\sum_{i=1}^n\tilde h_j\left(\boldsymbol{X}^{(i)}\right)\begin{bmatrix}{X_j^{(i)}}^{a-1}{\boldsymbol{X}^{(i)}}^a \\ -{X_j^{(i)}}^{b-1} \end{bmatrix}\begin{bmatrix}{{X_j^{(i)}}^{a-1}{\boldsymbol{X}^{(i)}}^a} \\ -{X_j^{(i)}}^{b-1} \end{bmatrix}^{\top},\\
\boldsymbol{\Gamma}_{m} &\equiv \begin{bmatrix}
\boldsymbol{\Gamma}_{\mathbf{K},m} & \boldsymbol{\gamma}_{\mathbf{K},\boldsymbol{\eta},m} \\ \boldsymbol{\gamma}_{\mathbf{K},\boldsymbol{\eta},m}^{\top} & \gamma_{\boldsymbol{\eta},m}
\end{bmatrix}\\
&\equiv\frac{1}{n}\sum_{i=1}^n\sum_{k=1}^{m-1}\tilde h_k\left(\boldsymbol{X}^{(i)}\right)\begin{bmatrix}{X_m^{(i)}}^{a-1}{\boldsymbol{X}^{(i)}}^a \\ -{X_m^{(i)}}^{b-1} \end{bmatrix}\begin{bmatrix}{X_m^{(i)}}^{a-1}{{\boldsymbol{X}^{(i)}}^a}\\ -{X_m^{(i)}}^{b-1} \end{bmatrix}^{\top},\\
\boldsymbol{\Gamma}_{(j,m)} &\equiv \begin{bmatrix}
\boldsymbol{\Gamma}_{\mathbf{K},(j,m)} & \boldsymbol{\gamma}_{\mathbf{K},\boldsymbol{\eta},(j,m)} \\ \boldsymbol{\gamma}_{\mathbf{K},\boldsymbol{\eta},(m,j)} & \gamma_{\boldsymbol{\eta},(j,m)}
\end{bmatrix}\\
&\equiv-\frac{1}{n}\sum_{i=1}^n\tilde h_j\left(\boldsymbol{X}^{(i)}\right)\begin{bmatrix}{X_j^{(i)}}^{a-1}{\boldsymbol{X}^{(i)}}^a \\ -{X_j^{(i)}}^{b-1} \end{bmatrix}\begin{bmatrix}{X_m^{(i)}}^{a-1}{{\boldsymbol{X}^{(i)}}^a} \\ -{X_m^{(i)}}^{b-1} \end{bmatrix}^{\top}.
\end{align*}
In addition, 
\begin{align*}
\boldsymbol{g}_{\mathbf{K},j}&\equiv\frac{1}{n}\sum_{i=1}^n\left[\partial_j\tilde h_j\left(\boldsymbol{X}^{(i)}\right){X_j^{(i)}}^{a-1}+(a-1)\tilde h_j\left(\boldsymbol{X}^{(i)}\right){X_j^{(i)}}^{a-2}\right]{\boldsymbol{X}^{(i)}}^a\\
&\quad\quad\quad+a\tilde h_j\left(\boldsymbol{X}^{(i)}\right){X_j^{(i)}}^{2a-2}\boldsymbol{e}_{j,m}-a\tilde h_j\left(\boldsymbol{X}^{(i)}\right){X_j^{(i)}}^{a-1}{X_m^{(i)}}^{a-1}\boldsymbol{e}_{m,m},\\
\boldsymbol{g}_{\mathbf{K},m}&\equiv\frac{1}{n}\sum_{i=1}^n\sum_{k=1}^{m-1}\left[-\partial_k\tilde h_k\left(\boldsymbol{X}^{(i)}\right){X_m^{(i)}}^{a-1}+(a-1)\tilde h_k\left(\boldsymbol{X}^{(i)}\right){{X_m^{(i)}}^{a-2}}\right]{\boldsymbol{X}^{(i)}}^a\\
&\quad\quad\quad+a\tilde h_k\left(\boldsymbol{X}^{(i)}\right){X_m^{(i)}}^{2a-2}\boldsymbol{e}_{m,m}-a\tilde h_k\left(\boldsymbol{X}^{(i)}\right){X^{(i)}_k}^{a-1}{X^{(i)}_m}^{a-1}\boldsymbol{e}_{k,m},\\
\boldsymbol{g}_{\boldsymbol{\eta},j}&\equiv\frac{1}{n}\sum_{i=1}^n-\partial_j\tilde h_j\left(\boldsymbol{X}^{(i)}\right){X_j^{(i)}}^{b-1}-(b-1)\tilde h_j\left(\boldsymbol{X}^{(i)}\right){X_j^{(i)}}^{b-2},\\
\boldsymbol{g}_{\boldsymbol{\eta},m}&\equiv\frac{1}{n}\sum_{i=1}^n\sum_{k=1}^{m-1}\partial_k\tilde h_k\left(\boldsymbol{X}^{(i)}\right){X_m^{(i)}}^{b-1}-(b-1)\tilde h_k\left(\boldsymbol{X}^{(i)}\right){X_m^{(i)}}^{b-2}.
\end{align*}

\begin{remark}
  As noted in Section \ref{Dependency on the Removed Coordinate},
  we may average the losses obtained by removing in turn each one of
  the coordinates in a set
  $\mathcal{J}\subseteq\{1,\ldots,m\}$, instead of only $m$, to
  mitigate the dependence on the choice of the coordinate
  removed. This yields a quadratic loss obtained by averaging the
  respective matrices $\boldsymbol{\Gamma}$ and vectors
  $\boldsymbol{g}$.  The time complexity of calculating
  $\boldsymbol{\Gamma}$ and $\boldsymbol{g}$ becomes linear in
  $|\mathcal{J}|$. On the other hand, similar to the existence of
  $\boldsymbol{\Gamma}_{(j,m)}$ in the formulas above, all blocks
  corresponding to indices in $\mathcal{J}$
  would become non-zero, which then makes the time complexity of
  coordinate descent methods for computing the loss minimizers
  $\hat{\mathbf{K}}$ and $\hat{\boldsymbol{\eta}}$ also linear in
  $|\mathcal{J}|$. As a result, instead of attempting to make the loss
  independent of the choice of removed coordinate by choosing
  $\mathcal{J}=\{1,\ldots,m\}$, randomly sampling 5-10 coordinates is
  more practical for high-dimensional problems. In our implementation in the R package
  \texttt{genscore}, the user has the option to specify the set
  $\mathcal{J}$.
\end{remark}

\subsection{$\log$--$\log$ Models on the Standard Simplex}\label{sec_Ad_Models}
 We now treat the special case of $a=0$ and $b=0$, namely, models with density proportional to 
\begin{equation}\label{eq_density_loglog}
\exp\left(-\frac{1}{2}\log\boldsymbol{x}^{\top}\mathbf{K}\log\boldsymbol{x}+\boldsymbol{\eta}^{\top}\log\boldsymbol{x}\right)
\end{equation}
supported on the $(m-1)$-dimensional standard simplex $\Deltam$.  This
class encompasses the $A^{m-1}$ class of distributions in Equation
(2.7) of \citet{ait85}, which have parameters $\boldsymbol{\beta}\equiv(\beta_j)_{j=1,\ldots,m}$ and $(\gamma_{jk})_{1\leq j\neq k\leq m}$, $\gamma_{jk}=\gamma_{kj}$, and density proportional to
\begin{equation}
  \label{eq:density:Am}
  \exp\bigg(-\frac{1}{2}\sum_{j=1}^m\sum_{k\not=j}
  % \sum_{\substack{j\neq k\\ j,k=1,\ldots,m}}
  \gamma_{jk}(\log x_j-\log
  x_k)^2+(\boldsymbol{\beta}-\mathbf{1}_m)^{\top}\log\boldsymbol{x}
  \bigg)\mathds{1}_{\Deltam}(\boldsymbol{x}).
\end{equation}
Indeed, expanding the exponent, the last display can be rewritten as
\[\exp\bigg(-\sum_{j=1}^m(\log x_j)^2\bigg(\sum_{k\not= j}
    %\sum_{k\in\{1,\ldots,m\}\backslash\{j\}}
  \gamma_{jk}\bigg) +\sum_{j=1}^m\sum_{k\not=j}
  % \sum_{\substack{j\neq k\\ j,k=1,\ldots,m}}
  \gamma_{jk}\log x_j\log
  x_k+\left(\boldsymbol{\beta}-\mathbf{1}_m\right)^{\top}\log\boldsymbol{x}
  \bigg).\]
Letting $\boldsymbol{\eta}\equiv\boldsymbol{\beta}-\mathbf{1}_m$ and taking
\[
  \kappa_{jj}=2\sum_{i\not=j}\gamma_{ji},
  \quad
  \kappa_{kj}=\kappa_{jk}=-2\gamma_{kj}, \quad 1\le j\neq k\le m,
\]
the $A^{m-1}$ model with densities as in (\ref{eq:density:Am})
translates to the $a$-$b$ model from~(\ref{eq_density_loglog}) 
%\[\exp\left(-\frac{1}{2}\log\boldsymbol{x}^{\top}\mathbf{K}\log\boldsymbol{x}+\boldsymbol{\eta}^{\top}\log\boldsymbol{x}\right)\]
%with
% when we take
% \[
%   \boldsymbol{\eta}\equiv\boldsymbol{\beta}-\mathbf{1}_m, \quad
%   \kappa_{jj}=2\sum_{i\in\{1,\ldots,m\}\backslash\{j\}}\gamma_{ji},
%   \quad
%   \kappa_{kj}=\kappa_{jk}=-2\gamma_{kj} \; 1\le j\neq k\le m.
% \]
% The $A^{m-1}$ model thus translates to the $a$-$b$ model
for $a=b=0$ and under the constraint that $\mathbf{K}\mathbf{1}_m=\mathbf{0}_m$ and $\mathbf{K}=\mathbf{K}^{\top}$. 

Next, we show that under simple
conditions the density is again proper and that the population version
of the generalized score matching loss can still be rewritten in the
form given in \eqref{eq_equivalent_loss}. %For estimation, for notational simplicity we again consider profiling out the last coordinate $x_m$ in the density and work with $\Deltam_{-m}\equiv\left\{\boldsymbol{x}_{-m}\in\mathbb{R}_+^{m-1}\right|\left.\boldsymbol{x}_{-m}\succ\boldsymbol{0},\,\mathbf{1}_{m-1}^{\top}\boldsymbol{x}_{-m}<1\right\}\subseteq\mathbb{R}_+^{m-1}$. 
The proof is given in the appendix.

\begin{theorem}\label{thm_assumption_Ad}
Suppose $\mathbf{K}$ is symmetric, and one of the following holds: 
\begin{enumerate}[(I)]
\item $\mathbf{K}$ is positive definite, or
\item $\mathbf{K}\mathbf{1}_m=\mathbf{0}_m$, $\mathbf{K}_{-k,-k}$ is positive definite for some $k=1,\ldots,m$, and $\mathbf{1}_m^{\top}\boldsymbol{\eta}+m\geq 0$, or
\item $\mathbf{K}\mathbf{1}_m=\mathbf{0}_m$, $\mathbf{K}$ is positive semi-definite, and $\boldsymbol{\eta}\succ-\mathbf{1}_m$.
\end{enumerate}
Then the density in \eqref{eq_density_loglog} has a finite normalizing
constant. Note that (II) implies that $\mathbf{K}$ is positive
semi-definite and (III) implies that for all $k$ the submatrix
$\mathbf{K}_{-k,-k}$ is positive semi-definite (but not necessarily
positive definite).

For all $j=1,\ldots,m-1$, let $h_j(x)=x^{\alpha_j}$ with $\alpha_j>0$.
If (I) or (II) hold, or if (III) holds with
$\alpha_j>\max\{1-\eta_{0,j},1-\eta_{0,m}\}$, then conditions (A.1)--(A.3) in Section \ref{Formulation in yu21} are satisfied.
\end{theorem}

We highlight in the log-log models obtained from $a=b=0$, the
parameters $\mathbf{K}$ and $\boldsymbol{\eta}$ are exactly
identifiable from the density (whether assuming
$\mathbf{K}\mathbf{1}_m=\mathbf{0}_m$ or not).  This follows in our
context as a corollary of Theorem \ref{thm_simplex_identifiability}.

\begin{corollary}\label{cor_simplex_identifiability}
Suppose there exist $\mathbf{K}_1$, $\mathbf{K}_2$, $\boldsymbol{\eta}_1$, $\boldsymbol{\eta}_2$ such that
\begin{align*}
  &\exp\left(-\frac{1}{2}\log(\boldsymbol{x})^{\top}\mathbf{K}_1\log(\boldsymbol{x})+\boldsymbol{\eta}_1^{\top}\log(\boldsymbol{x})\right)\\
  =&\exp\left(-\frac{1}{2}\log(\boldsymbol{x})^{\top}\mathbf{K}_2\log(\boldsymbol{x})+\boldsymbol{\eta}_2^{\top}\log(\boldsymbol{x})\right)
\end{align*}
for all $\boldsymbol{x}\in\Deltam$. Then $\mathbf{K}_1=\mathbf{K}_2$ and $\boldsymbol{\eta}_1=\boldsymbol{\eta}_2$.
\end{corollary}

% \subsubsection{Inference}

% As discussed in the beginning of Section \ref{$a$-$b$ Models on Standard Simplices}, conditional dependence between two components cannot be inferred from $\mathbf{K}$ due to the nature of the simplex domain. However, one can still be interested in recovering the sparsity of $\mathbf{K}$ to see how different components interact, and recovering $\mathbf{K}$ and $\boldsymbol{\eta}$ to estimate the density itself. For this purpose, we have the following corollary of Theorem \ref{thm_simplex_identifiability}, i.e.~$\mathbf{K}$ and $\boldsymbol{\eta}$ are exactly identifiable from the density (whether assuming $\mathbf{K}\mathbf{1}_m=\mathbf{0}_m$ or not):
% \begin{corollary}\label{cor_simplex_identifiability}
% Suppose there exist $\mathbf{K}_1$, $\mathbf{K}_2$, $\boldsymbol{\eta}_1$, $\boldsymbol{\eta}_2$ such that
% \[\exp\left(-\frac{1}{2}\log(\boldsymbol{x})^{\top}\mathbf{K}_1\log(\boldsymbol{x})+\boldsymbol{\eta}_1^{\top}\log(\boldsymbol{x})\right)=\exp\left(-\frac{1}{2}\log(\boldsymbol{x})^{\top}\mathbf{K}_2\log(\boldsymbol{x})+\boldsymbol{\eta}_2^{\top}\log(\boldsymbol{x})\right)\]
% for all $\boldsymbol{x}\in\Deltam$. Then $\mathbf{K}_1=\mathbf{K}_2$ and $\boldsymbol{\eta}_1=\boldsymbol{\eta}_2$.
% \end{corollary}

Now define the \emph{additive log-ratio transformation}
\[
  \boldsymbol{y}_{-m}\equiv%\mathrm{alt}(\boldsymbol{x})\equiv
   \log
  \boldsymbol{x}_{-m}-(\log x_{m})\mathbf{1}_{m-1}=\left(\log
    (x_1/x_m),\ldots,\log(x_{m-1}/x_m)\right).
  \]
  The $A^{m-1}$ model, corresponding to (II) and (III) in Theorem
  \ref{thm_assumption_Ad}, is proposed in \citet{ait85} as a
  generalization of both the Dirichlet distribution and the
  \emph{additive logistic normal model} \citep{ait82}. In particular,
  using the formula (\ref{eq_density_loglog}), when
  $\mathbf{K}=\mathbf{0}$ we have the Dirichlet distribution with
  parameters $\boldsymbol{\eta}+\mathbf{1}_m$, which belongs to case
  (III) in Theorem \ref{thm_assumption_Ad}. On the other hand, if
  $\mathbf{1}_m^{\top}\boldsymbol{\eta}=-m$, we get the normal density
  in $\boldsymbol{y}_{-m}$ with inverse covariance
  $\mathbf{K}_{-m,-m}$ and mean
  $\mathbf{K}_{-m,-m}^{-1}\boldsymbol{\eta}_{-m}$, which belongs to
  case (II) in Theorem \ref{thm_assumption_Ad}.  The generalization
  uses only one additional parameter when compared to the additive
  logistic normal model, namely $\mathbf{1}_m^{\top}\boldsymbol{\eta}$
  is no longer assumed to be equal to $-m$.
  % This brings more
  % flexibility to the two basic models over the simplex.

By the nature of the simplex domain, any two proportions $X_j$ and
$X_k$ are perfectly conditionally correlated given all other
$\boldsymbol{X}_{-j,-k}$. On the other hand, under the additive
logistic normal model, $Y_j=\log(X_j/X_m)$ and $Y_k=\log(X_k/X_m)$ are
conditionally independent given all other $\log(X_{\ell}/X_m)$,
$\ell\neq j,k,m$ if and only if $\kappa_{jk}=\kappa_{kj}=0$.  As we
make clear in the proof of Theorem \ref{thm_assumption_Ad}, this is
true only for the additive logistic normal model
($\mathbf{1}_m^{\top}\boldsymbol{\eta}=-m$).
% We note that the statement in this paragraph still holds when we replace the ratios with respect to $X_m$ by those w.r.t.~any $X_{\ell}$. 

% Assuming the more general $A^{m-1}$ model, we can thus perform the following two-step test of conditional independence between $X_{j}/X_{\ell}$ and $X_{k}/X_{\ell}$ for triplet $(j,k,\ell)$ with $j$, $k$, $\ell$ all different. First, test if $\mathbf{1}_m^{\top}\boldsymbol{\eta}+m$ is significantly different from $0$ by comparing the BIC/eBIC for fitted $A^{m-1}$ and additive logistic normal models with regularization discussed in Section \ref{Exponential Families and Regularized Generalized Score Matching}; if we reject the null, then we cannot establish conditional independence for any such triplets. Otherwise, we claim that for all $\ell\neq j,k$, $X_{j}/X_{\ell}$ and $X_{k}/X_{\ell}$ are conditionally independent given all other $X_{i}/X_{\ell}$ ($i\neq j,k,\ell$) if and only if $\kappa_{jk}=\kappa_{kj}=0$ in either fitted model. While the additive logistic normal model can be fitted as a Gaussian to the additive log-ratio transformed data, the $A^{m-1}$ model on the simplex can also be easily fitted with the help of generalized score matching, as discussed below.

\subsection{Estimation for $A^{m-1}$
  Models}\label{Estimation for $A^{m-1}$ Models}
In Section \ref{Estimation-Simplex}, we described how to estimate to
form generalized score matching estimators of $\mathbf{K}$ and
$\boldsymbol{\eta}$.  The discussion there applies to $\log$--$\log$
models ($a=b=0$), but only under the setting of assumption (I) from
Theorem~\ref{thm_assumption_Ad}, where $\mathbf{K}$ is unconstrained
except for positive definiteness.  For the $A^{m-1}$ models \citep{ait85} which
impose the additional constraint that $\mathbf{K}\mathbf{1}_m=\mathbf{0}_m$
and $\mathbf{K}=\mathbf{K}^{\top}$ as in (II) and (III) of Theorem
\ref{thm_assumption_Ad}, we need the following
modification. %Instead of dropping out the last column and the last row of $\mathbf{K}$, w
We marginalize out the diagonals of $\mathbf{K}$ with
$\kappa_{jj}=-\boldsymbol{\kappa}_{-j,j}^{\top}\mathbf{1}_{m-1}$ and
estimate all off-diagonal elements
$\mathbf{K}_{\mathrm{off}}\equiv[\boldsymbol{\kappa}_{-1,1},\ldots,\boldsymbol{\kappa}_{-m,m}]$. Under
the additional constraint, for matrices $\mathbf{A}$ with $m$ rows
and $\mathbf{B}$ with $m$ columns, we can write
\begin{align*}
\boldsymbol{\kappa}_{\cdot,j}^{\top}\mathbf{A}&=\boldsymbol{\kappa}_{-j,j}^{\top}\mathbf{A}_{-j,\cdot}+\kappa_{jj}\boldsymbol{a}_{j,\cdot}^{\top}=\boldsymbol{\kappa}_{-j,j}^{\top}\left(\mathbf{A}_{-j,\cdot}-\mathbf{1}_{m-1}\boldsymbol{a}_{j,\cdot}^{\top}\right)=\boldsymbol{\kappa}_{-j,j}^{\top}\left(
\mathbf{C}(j)\mathbf{A}
\right),\\
\mathbf{B}\boldsymbol{\kappa}_{\cdot,j}&=\mathbf{B}_{\cdot,-j}\boldsymbol{\kappa}_{-j,j}+\boldsymbol{b}_{\cdot,j}\kappa_{jj}=\left(\mathbf{B}_{\cdot,-j}-\boldsymbol{b}_{\cdot,j}\mathbf{1}_{m-1}^{\top}\right)\boldsymbol{\kappa}_{-j,j}=\big(\mathbf{B}{\mathbf{C}(j)}^{\top}\big)\boldsymbol{\kappa}_{-j,j},
\end{align*}
where $\mathbf{C}(j)\in\mathbb{R}^{(m-1)\times m}$ has its $j$-th
column all equal to $-1$, and the entries
$(1,1),\ldots,(j-1,j-1),(j,j+1),\ldots,(m-1,m)$ equal to 1, and all
other entries zero. Let $\mathbf{C}\in\mathbb{R}^{(m-1)m\times m^2}$
be the block-diagonal matrix with blocks
$\mathbf{C}(1),\ldots,\mathbf{C}(m)$. The unpenalized generalized
score-matching loss given by the first two terms in
(\ref{penalized_loss}) thus becomes
%$\frac{1}{2}\mathrm{vec}\left(\left[\mathbf{K}^{\top} \,\, \boldsymbol{\eta}\right]\right)^{\top}    \boldsymbol{\Gamma}\mathrm{vec}\left(\left[\mathbf{K}^{\top} \,\, \boldsymbol{\eta}\right]\right)-\boldsymbol{g}^{\top}\mathrm{vec}\left(\left[\mathbf{K}^{\top} \,\, \boldsymbol{\eta}\right]\right)$ 
\begin{multline*}%\label{eq_loss_Ad_unpenalized}
\frac{1}{2}\mathrm{vec}\left(\begin{bmatrix}\mathbf{K}_{\mathrm{off}}\\ \boldsymbol{\eta}^{\top}\end{bmatrix}\right)^{\top}\begin{bmatrix}
\mathbf{C}\boldsymbol{\Gamma}_{\mathbf{K}}\mathbf{C}^{\top} & \mathbf{C}\boldsymbol{\Gamma}_{\mathbf{K},\boldsymbol{\eta}} \\ \boldsymbol{\Gamma}_{\mathbf{K},\boldsymbol{\eta}}^{\top}\mathbf{C}^{\top} & \boldsymbol{\Gamma}_{\boldsymbol{\eta}}
\end{bmatrix}\mathrm{vec}\left(\begin{bmatrix}\mathbf{K}_{\mathrm{off}}\\
    \boldsymbol{\eta}^{\top}\end{bmatrix}\right)-
\begin{bmatrix}
\mathbf{C}\mathrm{vec}(\boldsymbol{g}_{\mathbf{K}}) \\ \boldsymbol{g}_{\boldsymbol{\eta}}
\end{bmatrix}
^{\top}\mathrm{vec}\left(\begin{bmatrix}\mathbf{K}_{\mathrm{off}}\\ \boldsymbol{\eta}^{\top}\end{bmatrix}\right).
\end{multline*}
For more compact notation, let
\begin{align*}
  \tilde{\boldsymbol{\Gamma}}\equiv&\,\begin{bmatrix}
\mathbf{C}\boldsymbol{\Gamma}_{\mathbf{K}}\mathbf{C}^{\top} & \mathbf{C}\boldsymbol{\Gamma}_{\mathbf{K},\boldsymbol{\eta}} \\ \boldsymbol{\Gamma}_{\mathbf{K},\boldsymbol{\eta}}^{\top}\mathbf{C}^{\top} & \boldsymbol{\Gamma}_{\boldsymbol{\eta}}
\end{bmatrix}\in\mathbb{R}^{m^2\times m^2},
\quad\quad\tilde{\boldsymbol{g}}\equiv\begin{bmatrix}
\mathbf{C}\mathrm{vec}(\boldsymbol{g}_{\mathbf{K}}) \\ \boldsymbol{g}_{\boldsymbol{\eta}}
\end{bmatrix}\in\mathbb{R}^{m^2}\nonumber.
\end{align*}
It holds that $\tilde{\boldsymbol{\Gamma}}$ is positive
(semi-)definite if and only if $\boldsymbol{\Gamma}$
from~(\ref{eq:Gamma}) is positive (semi-)definite.

% Denote the new components as $\tilde{\boldsymbol{\Gamma}}\in\mathbb{R}^{m^2\times m^2}$ and $\tilde{\boldsymbol{g}}\in\mathbb{R}^{m^2}$. For any $m^2$-dimensional nonzero $\boldsymbol{v}\equiv(\boldsymbol{v}_{\mathbf{K},1},\ldots,\boldsymbol{v}_{\mathbf{K},m},\boldsymbol{v}_{\boldsymbol{\eta}})$ with $\boldsymbol{v}_{\mathbf{K},j}\in\mathbb{R}^{m-1}$ and $\boldsymbol{v}_{\boldsymbol{\eta}}\in\mathbb{R}^m$, form $\tilde{\mathbf{V}}\in\mathbb{R}^{m\times m}$ whose $j$-th column has $\tilde{\boldsymbol{v}}_{-j,j}=\boldsymbol{v}_{\mathbf{K},j}$ and  $\tilde{v}_{jj}=-\mathbf{1}_{m-1}^{\top}\boldsymbol{v}_{\mathbf{K},j}$ (analogous to the inverse operation of $\mathbf{K}\mapsto\mathrm{vec}(\mathbf{K}_{\mathrm{off}})$), and let $\tilde{\boldsymbol{v}}\equiv (\mathrm{vec}(\tilde{\mathbf{V}}),\boldsymbol{v}_{\boldsymbol{\eta}})\in\mathbb{R}^{m(m+1)}$. Then, by definition, $\boldsymbol{v}^{\top}\tilde{\boldsymbol{\Gamma}}\boldsymbol{v}=\tilde{\boldsymbol{v}}^{\top}\boldsymbol{\Gamma}\tilde{\boldsymbol{v}}$. Since $\tilde{\boldsymbol{v}}\neq \boldsymbol{0}_{m(m+1)}$ if and only if $\boldsymbol{v}\neq\boldsymbol{0}_{m^2}$, $\tilde{\boldsymbol{\Gamma}}$ is positive (semi-)definite if and only if $\boldsymbol{\Gamma}$ is positive (semi-)definite.

When applying to the diagonal multiplication operation
$(\cdot)_{\delta}$ to form a loss as in~(\ref{eq_loss_regularized}),
we simply operate directly on the matrix
$\tilde{\boldsymbol{\Gamma}}$, rather than the matrix
$\boldsymbol{\Gamma}$.  Since $\tilde{\boldsymbol{\Gamma}}$ is merely
a linear transformation of $\boldsymbol{\Gamma}$, later high
probability bounds on are not affected, except in constants. 
%
%Since $\tilde{\boldsymbol{\Gamma}}$ is merely a linear transformation
%of $\boldsymbol{\Gamma}$, we may apply diagonal multipliers directly
%to $\tilde{\boldsymbol{\Gamma}}$ for simplicity, and the
%concentration properties of the original $\boldsymbol{\Gamma}$ (high
%probability bounds for deviation of $\boldsymbol{\Gamma}$ from its
%expectation) will be inherited, subject to different constant
%multiplicative factors.
The penalized generalized score-matching loss for the $A^{m-1}$ models is thus defined as
\begin{multline*}%\label{eq_loss_Ad_regularized}
\,\hat{L}_{\boldsymbol{h},\boldsymbol{C},\lambda,\delta}(p_{\mathbf{K},\boldsymbol{\eta}})
\equiv\,\frac{1}{2}\mathrm{vec}\left(\begin{bmatrix}\mathbf{K}_{\mathrm{off}}\\
    \boldsymbol{\eta}^{\top}\end{bmatrix}\right)^{\top}\tilde{\boldsymbol{\Gamma}}_{\delta}(\mathbf{x})\mathrm{vec}\left(\begin{bmatrix}\mathbf{K}_{\mathrm{off}}\\
    \boldsymbol{\eta}^{\top}\end{bmatrix}\right)-\tilde{\boldsymbol{g}}(\mathbf{x})^{\top}\mathrm{vec}\left(\begin{bmatrix}\mathbf{K}_{\mathrm{off}}\\
    \boldsymbol{\eta}^{\top}\end{bmatrix}\right)\\
+\lambda_{\mathbf{K}}\|\mathrm{vec}(\mathbf{K}_{\mathrm{off}})\|_1+\lambda_{\boldsymbol{\eta}}\|\boldsymbol{\eta}\|_1,
\end{multline*}
where
\begin{align*}%\label{eq_loss_Ad_regularized}
\tilde{\boldsymbol{\Gamma}}_{\delta}\equiv&\,\begin{bmatrix}
\left(\mathbf{C}\boldsymbol{\Gamma}_{\mathbf{K}}\mathbf{C}^{\top}\right)_{\delta} & \mathbf{C}\boldsymbol{\Gamma}_{\mathbf{K},\boldsymbol{\eta}} \\ \boldsymbol{\Gamma}_{\mathbf{K},\boldsymbol{\eta}}^{\top}\mathbf{C}^{\top} & \boldsymbol{\Gamma}_{\boldsymbol{\eta}}
\end{bmatrix}\in\mathbb{R}^{m^2\times m^2},
\quad\quad\tilde{\boldsymbol{g}}\equiv\begin{bmatrix}
\mathbf{C}\mathrm{vec}(\boldsymbol{g}_{\mathbf{K}}) \\ \boldsymbol{g}_{\boldsymbol{\eta}}
\end{bmatrix}\in\mathbb{R}^{m^2}.
\end{align*}
%where $(\cdot)_{\delta}$ denotes the operation of multiplying the diagonals of a matrix by $\delta>1$.

\section{Theoretical Properties}\label{Theory}
In this section, we present theoretical guarantees for our generalized
score matching estimators when applied to the pairwise interaction
power $a$-$b$ models on the simplex.  We consider high-dimensional
settings under $\ell_1$ regularization and derive bounds on the deviation of our
estimates $\hat{\mathbf{K}}$ and $\hat{\boldsymbol{\eta}}$ (minimizer
of (\ref{penalized_loss})) from their true values
$\mathbf{K}_0$ and $\boldsymbol{\eta}_0$ that hold with high
probability.

We begin by restating
% In particular, similar to
% \citet{yu19}, with high probability we bound the deviation of our
% estimates $\hat{\mathbf{K}}$ and $\hat{\boldsymbol{\eta}}$ (minimizer
% of Equation (\ref{penalized_loss})) from their true values
% $\mathbf{K}_0$ and $\boldsymbol{\eta}_0$. We first restate
Definition
12 from \citet{yu19}.

\begin{definition}\label{def_constants_centered}
Let
$\boldsymbol{\Gamma}_0\equiv\mathbb{E}_0\boldsymbol{\Gamma}(\mathbf{x})$
and $\boldsymbol{g}_0\equiv\mathbb{E}_0\boldsymbol{g}(\mathbf{x})$ be
the expectations of $\boldsymbol{\Gamma}(\mathbf{x})$ and
$\boldsymbol{g}(\mathbf{x})$ under the distribution given by a true
parameter matrix
$\mathbf{\Psi}_0\equiv\left[\mathbf{K}_0,\boldsymbol{\eta}_0\right]^{\top}\in\mathbb{R}^{m(m+1)}$,
or $\mathbf{\Psi}_0\equiv\mathbf{K}_0\in\mathbb{R}^{m^2}$ in the
``centered'' case with $\boldsymbol{\eta}_0\equiv\boldsymbol{0}$.  The
support of a matrix $\mathbf{\Psi}=(\psi_{ij})$ is
$S(\mathbf{\Psi})\equiv\{(i,j):\psi_{ij}\neq 0\}$, and we let
$S_0=S(\mathbf{\Psi}_0)$.  Furthermore, let $d_{\mathbf{\Psi}_0}$ be
the maximum number of non-zero entries in any column of
$\mathbf{\Psi}_0$, and let $c_{\mathbf{\Psi}_0}\equiv\mnorm{\mathbf{\Psi}_0}_{\infty,\infty}$. Writing $\mathbf{\Gamma}_{0,AB}$ for the $A\times B$ submatrix of $\boldsymbol{\Gamma}_0$, we define
\begin{equation*}%\label{def_c_gamma}
c_{\boldsymbol{\Gamma}_0}\equiv\mnorm{(\boldsymbol{\Gamma}_{0,S_0S_0})^{-1}}_{\infty,\infty}.
\end{equation*}
Then $\boldsymbol{\Gamma}_0$ satisfies \emph{the irrepresentability condition with incoherence parameter $\alpha\in(0,1]$ and support set $S_0$} if
\begin{equation}\label{irrepresentability}
\mnorm{\boldsymbol{\Gamma}_{0,S_0^cS_0}(\boldsymbol{\Gamma}_{0,S_0S_0})^{-1}}_{\infty,\infty}\leq (1-\alpha).
\end{equation}
\end{definition}

For simplicity, in the proof of the theoretical results presented in
this section, we assume the last coordinate $x_m$ is removed. Of
course, the argument is the same for any other coordinate. In
addition, the results generalize to the case where we average the loss
functions (i.e.,~$\boldsymbol{\Gamma}(\mathbf{x})$ and
$\boldsymbol{g}(\mathbf{x})$) over multiple coordinates, one removed
at a time, as discussed in Sections \ref{Dependency on the Removed
  Coordinate} and \ref{Estimation-Simplex}. This follows by applying
the triangle inequality since the theorems all follow from
probabilistic bounds on the deviation of $\boldsymbol{\Gamma}$ from
$\boldsymbol{\Gamma}_0$ and $\boldsymbol{g}$ from $\boldsymbol{g}_0$.

\subsection{Models on the Standard Simplex}

For models with $a>0$ on the simplex, the fact that each coordinate is
in $[0,1]$ allows us to derive the following result as a corollary of \refyu{Theorem
  5.3 in \citet{yu21}}.

\begin{theorem}\label{theorem_simplex_nonlog_ggm}
Suppose $a>0$ and $b\geq 0$. Suppose further that the true parameters $\mathbf{K}_0$ and $\boldsymbol{\eta}_0$ satisfy the conditions in Theorem \ref{thm_norm_const} (so that the density is proper). Assume $\boldsymbol{h}(\boldsymbol{x})\equiv(x_1^{\alpha_1},\ldots,x_m^{\alpha_m})$ with $\alpha_1,\ldots,\alpha_m\geq \max\{1,2-a,2-b\}$. %, and suppose $\boldsymbol{\varphi}_{\boldsymbol{C}}$ has truncation points $\boldsymbol{C}=(C_1,\ldots,C_m)$ with $0<C_j< +\infty$ for $j=1,\ldots,m$. 
Define
\begin{align*}
\varsigma_{\boldsymbol{\Gamma}}&\,\equiv1,\quad\quad\varsigma_{\boldsymbol{g}}\equiv\max_{j=1,\ldots,m}\alpha_j+\max\{|a-1|+2a,|b-1|\}.
\end{align*}
Suppose, without loss of generality, that $\lambda\equiv\lambda_{\boldsymbol{K}}=\lambda_{\boldsymbol{\eta}}$; otherwise replace $\boldsymbol{\eta}$ by $(\lambda_{\boldsymbol{\eta}}/\lambda_{\mathbf{K}})\boldsymbol{\eta}$. Suppose that $\boldsymbol{\Gamma}_{0,S_0S_0}$ is invertible and satisfies the irrepresentability condition in Equation \ref{irrepresentability} with $\omega\in(0,1]$. Suppose for $\tau>0$ the sample size, the regularization parameter and the diagonal multiplier $\delta$ from Section \ref{Exponential Families and Regularized Generalized Score Matching} satisfy
\begin{align}
n&>72c_{\boldsymbol{\Gamma}_0}^2d_{\boldsymbol{\Psi}_0}^2\varsigma_{\boldsymbol{\Gamma}}^2(\tau\log m+\log 4)/\omega^2,\label{eq_bounded_nonlog_n}\\
\lambda&>\frac{3(2-\omega)}{\omega}\max\left\{c_{\boldsymbol{\Psi}_0}\varsigma_{\boldsymbol{\Gamma}}\sqrt{2(\tau\log m+\log 4)/n},\varsigma_{\boldsymbol{g}}\sqrt{(\tau\log m+\log 4)/(2n)}\right\},\\
1&<\delta<C_{\text{bounded}}(n,m,\tau)\equiv 1+\sqrt{(\tau\log m+\log 4)/(2n)}.\label{eq_bounded_nonlog_delta}
\end{align}
Then the following statements hold with probability $1-m^{-\tau}$:
\begin{enumerate}[(a)]
\item The regularized generalized $\boldsymbol{h}$-score matching estimator $\hat{\boldsymbol{\Psi}}$ that minimizes Equation \ref{eq_loss_regularized} is unique, has its support included in the true support, $\hat{S}\equiv S(\hat{\boldsymbol{\Psi}})\subseteq S_0$, and satisfies
\begin{alignat*}{3}
\max\big\{\|\hat{\mathbf{K}}-\mathbf{K}_0\|_{\infty},\|\hat{\boldsymbol{\eta}}-\boldsymbol{\eta}_0\|_{\infty}\big\}&\leq\frac{c_{\boldsymbol{\Gamma}_0}}{2-\omega}\lambda,\\
\max\big\{\mnorm{\hat{\mathbf{K}}-\mathbf{K}_0}_{F},\mnorm{\hat{\boldsymbol{\eta}}-\boldsymbol{\eta}_0}_{F}\big\}&\leq\frac{c_{\boldsymbol{\Gamma}_0}}{2-\omega}\lambda\sqrt{|S_0|},\\
\max\big\{\mnorm{\hat{\mathbf{K}}-\mathbf{K}_0}_{2},
\mnorm{\hat{\boldsymbol{\eta}}-\boldsymbol{\eta}_0}_{2}\big\}&\leq\frac{c_{\boldsymbol{\Gamma}_0}}{2-\omega}\lambda\min\left(\sqrt{|S_0|},d_{\boldsymbol{\Psi}_0}\right).
\end{alignat*}
% \begin{alignat*}{3}
% \|\hat{\mathbf{K}}-\mathbf{K}_0\|_{\infty}&\leq\frac{c_{\boldsymbol{\Gamma}_0}}{2-\omega}\lambda,&&\|\hat{\boldsymbol{\eta}}-\boldsymbol{\eta}_0\|_{\infty}\leq\frac{c_{\boldsymbol{\Gamma}_0}}{2-\omega}\lambda,\\
% \mnorm{\hat{\mathbf{K}}-\mathbf{K}_0}_{F}&\leq\frac{c_{\boldsymbol{\Gamma}_0}}{2-\omega}\lambda\sqrt{|S_0|},&&\mnorm{\hat{\boldsymbol{\eta}}-\boldsymbol{\eta}_0}_{F}\leq\frac{c_{\boldsymbol{\Gamma}_0}}{2-\omega}\lambda\sqrt{|S_0|},\\
% \mnorm{\hat{\mathbf{K}}-\mathbf{K}_0}_{2}&\leq\frac{c_{\boldsymbol{\Gamma}_0}}{2-\omega}\lambda\min\left(\sqrt{|S_0|},d_{\boldsymbol{\Psi}_0}\right),\quad&&
% \mnorm{\hat{\boldsymbol{\eta}}-\boldsymbol{\eta}_0}_{2}\leq\frac{c_{\boldsymbol{\Gamma}_0}}{2-\omega}\lambda\min\left(\sqrt{|S_0|},d_{\boldsymbol{\Psi}_0}\right).
% \end{alignat*}
\item Moreover, if
\[\min_{j,k:(j,k)\in S_0}|\kappa_{0,jk}|>\frac{c_{\boldsymbol{\Gamma}_0}}{2-\omega}\lambda\quad\quad\text{and}\quad\quad\min_{j:(m+1,j)\in S_0}|\eta_{0,j}|>\frac{c_{\boldsymbol{\Gamma}_0}}{2-\omega}\lambda,\]
then $\hat{S}=S_0$ and $\mathrm{sign}(\hat{\kappa}_{jk})=\mathrm{sign}(\kappa_{0,jk})$ for all $(j,k)\in S_0$ and $\mathrm{sign}(\hat{\eta}_j)=\mathrm{sign}(\eta_{0j})$ for $(m+1,j)\in S_0$.
\end{enumerate}

\end{theorem}
We note that the requirement on $\alpha_j\geq 1$ is only used for
bounding the two $\partial_j(h_j\circ\varphi_j)$ terms in
$\boldsymbol{g}(\mathbf{x})$. Our simulations studies indicate that
the method also works for smaller $\alpha_j$ and that it might not be
necessary to enforce the constraint $\alpha_j\geq 1$ in practice.  We
further note that in
the proof of Theorem \ref{theorem_simplex_nonlog_ggm} we show that we
can give tighter constant bounds $\varsigma_{\boldsymbol{\Gamma}}$ and
$\varsigma_{\boldsymbol{g}}$ for entries in $\boldsymbol{\Gamma}$ and
$\boldsymbol{g}$, respectively, which may be much smaller but have
rather complicated forms.

For models with $a=0$ on simplex domains, including the $A^{m-1}$
models discussed in Section \ref{sec_Ad_Models}, we first derive the following lemma to bound $\log X_j$ with high probability.

\begin{lemma}\label{theorem_subexp_simplex}
Suppose $\boldsymbol{X}$ has the density from \eqref{eq_interaction_density2} on $\Deltam$ with true parameters $\mathbf{K}_0$ and $\boldsymbol{\eta}_0$ satisfying the conditions in Theorem \ref{thm_norm_const} for $a>0$ or $b>0$, or in Theorem \ref{thm_assumption_Ad} for $a=b=0$. Then for all $j=1,\ldots,m$, $X_j^{2a}$ is sub-exponential for $a>0$, and $\log X_j$ is sub-exponential for $a=0$.
\end{lemma}

We then have the following corollary.

\begin{corollary}\label{theorem_simplex_log_ggm}
Suppose $a=0$. %with true density \[\exp\left(-\frac{1}{2}\log\boldsymbol{x}^{\top}\mathbf{K}_0\log\boldsymbol{x}+\boldsymbol{\eta}_0^{\top}\log\boldsymbol{x}\right)\]
%on $\mathcal{D}$, as discussed in Section \ref{sec_Ad_Models}. 
Also suppose $b=0$ and the conditions for $\mathbf{K}_0$ and $\boldsymbol{\eta}_0$ in Theorem \ref{thm_assumption_Ad} hold, or $b>0$ and the condition in Theorem \ref{thm_norm_const} hold. Let $\boldsymbol{h}(\boldsymbol{x})\equiv(x_1^{\alpha_1},\ldots,x_m^{\alpha_m})$ with $\alpha_1,\ldots,\alpha_m\geq 2$. Then Theorem \ref{theorem_simplex_nonlog_ggm} holds with $\log 4$ replaced by $\log 6$ in \eqref{eq_bounded_nonlog_n}-\eqref{eq_bounded_nonlog_delta}, and
\begin{align*}
\varsigma_{\boldsymbol{\Gamma}}&\equiv \max\left\{1,c_{\log,\mathbf{K}_0,\boldsymbol{\eta}_0}^2\right\},\\
  \varsigma_{\boldsymbol{g}}&\equiv\max\left\{\left(\max_{j=1,\ldots,m}\alpha_j+1\right)c_{\log,\mathbf{K}_0,\boldsymbol{\eta}_0}+2,\max_j\alpha_j+|b-1|\right\},
                              \intertext{where}
c_{\log,\mathbf{K}_0,\boldsymbol{\eta}_0}&\equiv\max_j\mathbb{E}_0\log
                                           X_j\\
  &+\max\left\{2\sqrt{2}e\max_{j=1,\ldots,m}\|\log X_j\|_{\psi_1}\sqrt{\log 3+\log n+(\tau+1)\log m},\right.\\
&\quad\quad\quad\quad\left. 4e\max_{j=1,\ldots,m}\|\log
                                                                                                                 X_j\|_{\psi_1}(\log 3+\log n+(\tau+1)\log m)\right\},
\end{align*}
and $\|\log X_j\|_{\psi_1}\equiv\sup_{q\geq 1}\left(\mathbb{E}_0|\log X_j|^q\right)^{1/q}/q\geq -\mathbb{E}_0\log X_j$.
\end{corollary}
The results are written in terms of the maximum of the sub-exponential
norms of $\log X_1,\allowbreak\ldots,\allowbreak\log X_m$, and they indicate the
sample size requirement that
$n=\Omega\left(\log m\right)\mathcal{O}\left(\max_j\left\|\log
    X_j\right\|_{\psi_1}^{2}\right)$.  We expect that the maximum of
the
sub-exponential norms scales as
$\Omega\left(\left(\log m\right)^c\right)$ for some $c$ small,
although we cannot currently offer an exact result on this behavior.
%and have to leave this point for further research.

% , an unknown constant, %but we expect it to scale like $\mathcal{O}\left(\left(\log m\right)^c\right)$ for some $c$ small. The exact scaling is left for further research. %Again, the result involves the sub-exponential norm $\left\|\log X_j\right\|_{\psi_1}$ of $\log X_j$, and w
% and we have roughly $n=\Omega\left(\log m\right)\mathcal{O}\left(\max_j\left\|\log X_j\right\|_{\psi_1}\right)^{2}$. We expect the sub-exponential norm to scale as $\Omega\left(\left(\log m\right)^c\right)$ for some $c$ small, and the exact dependency on $m$ is left for further research.

\section{Numerical Studies}\label{Numerical Experiments}

\subsection{Numerical Experiments for $A^{m-1}$ Models on the Simplex}
In this section we present results from numerical experiments using our regularized generalized score matching estimator for simplices. %, as well as our extension to \citet{liu19} presented in \ref{Extension of g0}.
%The plots we report are similar to those in \refyu{Section 6 of \citet{yu21}}.
\subsubsection{Choices of $\boldsymbol{h}$ and $\boldsymbol{C}$}\label{sec_simulation_h_and_C}
We consider removing $x_m=1-x_1-\cdots-x_{m-1}$. Recall that multiplication of $\nabla\log p(\boldsymbol{x})$ with $(\boldsymbol{h}\circ\boldsymbol{\varphi}_{\boldsymbol{C}})^{1/2}(\boldsymbol{x})$ is key to our method; here, the $j$-th component of $\boldsymbol{\varphi}_{\boldsymbol{C}}(\boldsymbol{x})=(\varphi_{C_1,1}(\boldsymbol{x}),\ldots,\varphi_{C_{m-1},{m-1}}(\boldsymbol{x}))$ is the distance of $x_j$ to the boundary of its domain holding $\boldsymbol{x}_{-j}$ fixed, with this distance truncated from above by a constant $C_j>0$. Thus, $\varphi_{C_j,j}=\min\{C_j,x_j,x_m\}$.

We use the same function $h$ for all components of
$\boldsymbol{h}(\boldsymbol{x})=(h(x_1),\ldots,h(x_m))$ and compare
the performance of our method using various choices of $h$ of the
form $h(x)=x^{c}$ for some power $c\geq 0$ along with various
truncation points $\boldsymbol{C}$. In particular, we choose $c=i/4$
for $i=0,1,\ldots,8$. Instead of pre-specifying constants
$\boldsymbol{C}$, following \citet{yu21}, we choose a probability
$\pi\in(0,1]$ and set each $C_j$ to be the $\pi$ sample quantile of $\varphi_{1,j}$ applied to each row of the data matrix $\mathbf{x}$, namely $\{\varphi_{1,j}(\boldsymbol{x}^{(1)}),\ldots,\varphi_{1,j}(\boldsymbol{x}^{(n)})\}$, assuming there are $n$ samples in the data. %This dynamic way of choosing the truncation points allows us to automatically adapt to the scale of data, and is more informative than fixing the constant to a grid from 0.5 to 3 as done in \citet{yu19}. 
In our experiments, we choose $\pi\in\{0.2,0.4,0.6,0.8,1\}$, where $\pi=1$ means no truncation for all finite $\varphi_j$ values.

%We remind that with power $c=0$, $(\boldsymbol{h}\circ\boldsymbol{\varphi}_{\boldsymbol{C},\mathcal{D}})(\boldsymbol{x})\equiv 1$ and it corresponds to the original score-matching for $\mathbb{R}^m$ of \citet{hyv05}, while with $c=2$, $\boldsymbol{C}=+\infty^m$ and $\mathcal{D}\equiv\mathbb{R}_+^m$, $(\boldsymbol{h}\circ\boldsymbol{\varphi}_{\boldsymbol{C},\mathcal{D}})(\boldsymbol{x})\equiv \boldsymbol{x}^2$ corresponds to the estimator of \citet{hyv07} and \citet{lin16}. The case where $\mathcal{D}\equiv\mathbb{R}_+^m$ corresponds to \citet{yu18,yu19}.

%Finally, we also include results for our extension to the method proposed in \citet{liu19} using $g_0(\boldsymbol{x})$ as opposed to $(\boldsymbol{h}\circ\boldsymbol{\varphi}_{\boldsymbol{C},\mathcal{D}})^{1/2}(\boldsymbol{x})$, taking the $\ell_2$ distance to the boundary $\partial\mathcal{D}$ upper truncated by a constant $C$; see Section \ref{Extension of g0}. The constant $C$ in this case is also determined using quantiles of the untruncated $\ell_2$ distances of the given data sample to $\partial\mathcal{D}$. For $C=+\infty$ ($\pi=1$) there is no truncation and the estimator corresponds to \citet{liu19}.

\subsubsection{Experimental Setup}
As the most prominent example of models on the simplex% $\Deltam\equiv\left\{\boldsymbol{x}\in\mathbb{R}_+^m\right|\left.\boldsymbol{x}\succ\boldsymbol{0},\,\mathbf{1}_m^{\top}\boldsymbol{x}=1\right\}$
, we consider the $A^{m-1}$ models discussed in Section~\ref{sec_Ad_Models}. In these models, $a=b=0$ and $\mathbf{K}\mathbf{1}_m=\mathbf{0}_m$ with $\mathbf{K}=\mathbf{K}^{\top}$. %As in \citet{yu21}
We consider dimension $m=100$, sample sizes $n=80$ and $n=1000$, and
assume $\boldsymbol{\eta}_0\equiv \boldsymbol{0}$ is known for
simplicity. The density is then proportional to $\exp\left(-{\log \boldsymbol{x}}^{\top}\mathbf{K}\log\boldsymbol{x}/2\right)$.

We set the true interaction matrix $\mathbf{K}_0$ to be a banded matrix
with bandwidths $s=7$ for $n=1000$ and $s=2$ for $n=80$; the
bandwidth, defined as $\max\{|i-j|:\kappa_{0,i,j}>0\}$, is chosen so that $n/(d_{\mathbf{K}_0}^2\log m)$ is roughly constant, where $d$ is the maximum node degree; this quantity is suggested to be linked to the probability of successful support recovery by our consistency theory in Section~\ref{Theory}.
We set $\kappa_{0,i,j}$ to $1-|i-j|/(s+1)$ for $1\leq |i-j|\leq s$, and set
the diagonals so that $\mathbf{K}_0\mathbf{1}_m=\boldsymbol{0}_m$. Each $n$ is thus associated with only one graph, for which we run 50 trials. 

We investigate recovery of the interaction pattern given by the
support of $\mathbf{K}_0$, as well as estimation of the entries of
$\mathbf{K}_0$.  The comparison to $\mathbf{K}_0$ is motivated by the discussion in Section~\ref{sec_Ad_Models}, in particular the identifiability result in Corollary \ref{cor_simplex_identifiability}.
The off-diagonal entries in the support of $\mathbf{K}_0$ naturally
define the edges of a graph.  In the sequel, we will thus refer to
edge recovery and plot estimation results as graphs.

% As discussed in Section \ref{sec_Ad_Models}, although conditional dependence between two components does not correspond to the zero/nonzero pattern of the entries in $\mathbf{K}$, one may still be interested in successful recovery of such pattern, as well as minimization of the estimation error of $\mathbf{K}$ and $\boldsymbol{\eta}$. This is backed up by Corollary \ref{cor_simplex_identifiability}, according to which $\mathbf{K}$ and $\boldsymbol{\eta}$ are exactly identifiable from the distribution.

\subsubsection{AUCs and Estimation Errors}\label{AUCs and Estimation Errors}
To investigate the recovery of interaction patterns, we consider the \emph{areas under the ROC curve}  (AUCs); the ROC curve plots the \emph{true positive rate (TPR)} against the \emph{false positive rate (FPR)}, defined as
\[\mathrm{TPR}\equiv\frac{|\hat{S}_{\text{off}}\cap S_{0,\text{off}}|}{|S_{0,\text{off}}|}\quad\quad\text{and}\quad\quad\mathrm{FPR}\equiv\frac{|\hat{S}_{\text{off}}\backslash S_{0,\text{off}}|}{m(m-1)-|S_{0,\text{off}}|},\]
where $\hat{S}_{\mathrm{off}}$ and $S_{0,\mathrm{off}}$ are the sets
of estimated and true edges, i.e., pairs of distinct indices $(i,j)$ in the support of the estimated and the true interaction matrix, respectively.

Average AUCs over 50 trials are shown in Figure~\ref{plot_Ad} as functions of $\pi\in\{0.2,0.4,0.6,\penalty 0  0.8,1\}$, which correspond to the column-wise sample quantiles used as the truncation points $\boldsymbol{C}$ for $\boldsymbol{\varphi}_{\boldsymbol{C}}$ (c.f.~Section \ref{sec_simulation_h_and_C}). 
Each curve in the figure represents $h(x)=1$, % \citep{hyv05}, %$g_0(\boldsymbol{x})$ from Section \ref{Extension of g0}, 
or $h(x)=x^c$ with $c=1/4,1/2,\dots,2$. The $y$-ticks on the
right-hand side denote the corresponding AUCs, while those on the left
are the AUCs divided by the AUC for $h(x)=1$ as a reference, measuring
the relative performance of each method compared with the estimator for
densities on $\mathbb{R}^m$ first given by \citet{hyv05}; the dotted line corresponds to the AUC for $h(x)=1$. We fix the diagonal multiplier to the upper bound in (\ref{eq_bounded_nonlog_delta}).

As discussed in Section~\ref{Dependency on the Removed Coordinate} and
at the end of Section~\ref{Estimation-Simplex}, to reduce the effect
of the choice of the removed coordinate, one can randomly sample a set
of coordinates $\mathcal{J}$, calculate $\boldsymbol{\Gamma}$ and $\boldsymbol{\eta}$ by removing one coordinate $x_j,\,j\in\mathcal{J}$ at a time, and take the average. In the first row, we plot the results for such $\boldsymbol{\Gamma}$ and $\boldsymbol{\eta}$ constructed with $\mathcal{J}$ randomly sampled from $\{1,\dots,m\}$, where $|\mathcal{J}|=5$. In order to investigate the benefit of using $|\mathcal{J}|>1$ over $|\mathcal{J}|=1$ (e.g.,~removing $x_m$ only), in the second row we also present the average of the 5 AUC curves over 5 separate runs, where in each run we construct $\boldsymbol{\Gamma}$ and $\boldsymbol{\eta}$ by removing one $j\in\mathcal{J}$ only. The third row shows the point-wise maximum of the 5 AUC curves.

\begin{figure}[!htp]
\centering
\vspace{-0.1in}
{\includegraphics[scale=0.55]{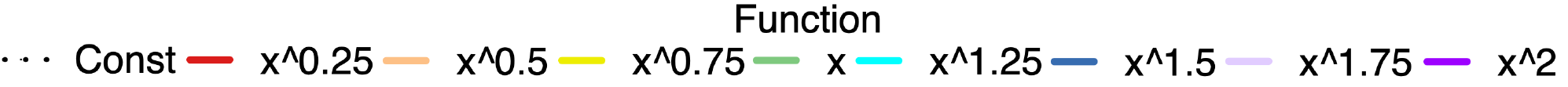}}
\vspace{-0.15in}

\subfloat[$n=80$, $|\mathcal{J}|=5$]
{\includegraphics[scale=0.28]{{Plots/AUC/Compare_AUC/log_log_sum0_n80p100s2e0.1_Gband__eta0_0_c_dm1.3518_AUCs.dat_compareAUCs_All}.pdf}\hspace{0.15in}}
\subfloat[$n=1000$, $|\mathcal{J}|=5$]
{\includegraphics[scale=0.28]{{Plots/AUC/Compare_AUC/log_log_sum0_n1000p100s7e0.1_Gband__eta0_0_c_dm1.0995_AUCs.dat_compareAUCs_All}.pdf}\hspace{-0.02in}}
\\ \vspace{-0.1in}
\subfloat[$n=80$, average of 5 runs with each $j\in |\mathcal{J}|$]
{\includegraphics[scale=0.28]{{Plots/AUC/Compare_AUC/log_log_sum0_n80p100s2e0.1_Gband__eta0_0_c_dm1.3518_AUCs.dat_compareAUCs_Avg}.pdf}\hspace{0.15in}}
\subfloat[$n=1000$, average of 5 runs]% with each $j\in |\mathcal{J}|$]
{\includegraphics[scale=0.28]{{Plots/AUC/Compare_AUC/log_log_sum0_n1000p100s7e0.1_Gband__eta0_0_c_dm1.0995_AUCs.dat_compareAUCs_Avg}.pdf}\hspace{-0.02in}}
\\ \vspace{-0.1in}
\subfloat[$n=80$, max of 5 runs]% with each $j\in |\mathcal{J}|$]
{\includegraphics[scale=0.28]{{Plots/AUC/Compare_AUC/log_log_sum0_n80p100s2e0.1_Gband__eta0_0_c_dm1.3518_AUCs.dat_compareAUCs_Max}.pdf}\hspace{0.15in}}
\subfloat[$n=1000$, max of 5 runs]% with each $j\in |\mathcal{J}|$]
{\includegraphics[scale=0.28]{{Plots/AUC/Compare_AUC/log_log_sum0_n1000p100s7e0.1_Gband__eta0_0_c_dm1.0995_AUCs.dat_compareAUCs_Max}.pdf}\hspace{-0.02in}}
\caption{AUCs averaged over 50 trials for edge recovery for the $A^{m-1}$ models on the simplex.}\label{plot_Ad}
\end{figure}

\begin{figure}[!htp]
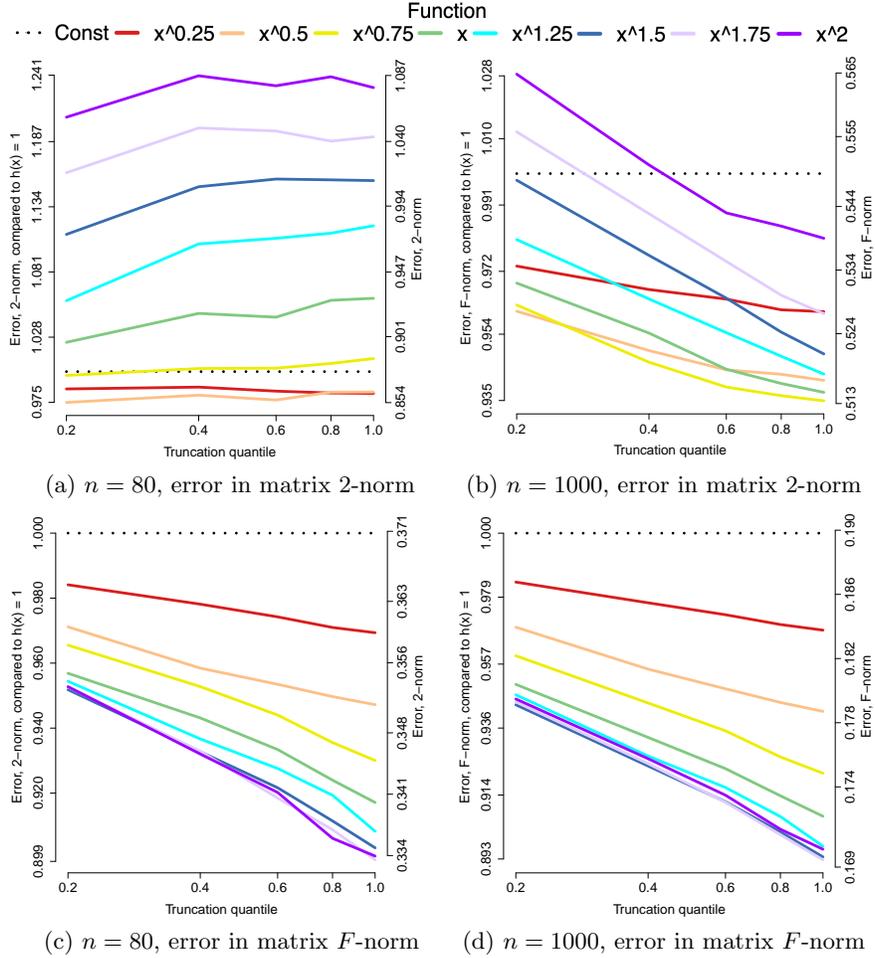

\centering
\vspace{-0.1in}
{\includegraphics[scale=0.55]{Plots/AUC/function_legend.pdf}}
\vspace{-0.15in}

\subfloat[$n=80$, error in matrix 2-norm]
{\includegraphics[scale=0.28]{{Plots/AUC/CV_norm/log_log_sum0_n80p100s2e0.1_Gband__eta0_0_c_dm1.3518_others.dat_cv_2}.pdf}\hspace{0.15in}}
\subfloat[$n=1000$, error in matrix 2-norm]
{\includegraphics[scale=0.28]{{Plots/AUC/CV_norm/log_log_sum0_n80p100s2e0.1_Gband__eta0_0_c_dm1.3518_others.dat_cv_F}.pdf}\hspace{-0.02in}}
\\ \vspace{-0.1in}
\subfloat[$n=80$, error in matrix $F$-norm]
{\includegraphics[scale=0.28]{{Plots/AUC/CV_norm/log_log_sum0_n1000p100s7e0.1_Gband__eta0_0_c_dm1.0995_others.dat_cv_2}.pdf}\hspace{0.15in}}
\subfloat[$n=1000$, error in matrix $F$-norm]
{\includegraphics[scale=0.28]{{Plots/AUC/CV_norm/log_log_sum0_n1000p100s7e0.1_Gband__eta0_0_c_dm1.0995_others.dat_cv_F}.pdf}\hspace{-0.02in}}
\caption{Averaged error in spectral and $F$ norms over 50 trials normalized by the corresponding norms of the true $\mathbf{K}_0$; sparsity chosen by cross validation; $|\mathcal{J}|=5$.}\label{plot_Ad_error}
\end{figure}

Based on the plots, for edge recovery %and error measured by the Frobenius norm, 
$h(x)=x^2$ is among the best performers, with $C_j\equiv 1$ (no truncation) being a safe choice, supporting our previous conclusion of the choice of $h(x)=x^{\max\{2-a,0\}}$ for general $a$-$b$ models in \citet{yu19,yu21}. %
It can also be observed that while using multiple removed coordinates $\mathcal{J}$ is beneficial for the high-dimensional case, the improvement in the low-dimensional case may not justify the added computational burden.

In Figure~\ref{plot_Ad_error} we plot the estimation error in spectral
and Frobenius norms, i.e.~$\mnorm{\hat{\mathbf{K}}-\mathbf{K}_0}_2$
and $\mnorm{\hat{\mathbf{K}}-\mathbf{K}_0}_F$, against the quantile
probability $\pi$. The estimate is chosen by cross validation from the
estimates with $|\mathcal{J}|=5$. The $y$-ticks on the right-hand side
are the errors, and those on the left are the errors divided by the
error for $h(x)=1$, measuring the relative performance of each method
compared with \citet{hyv05}. In contrast to Figure~\ref{plot_Ad}, smaller values on the $y$-axis indicate better performance. As in Figure~\ref{plot_Ad}, $h(x)=x^2$ performs the best when considering the Frobenius norm. When the error is measured in the spectral norm, $h(x)=x^2$ has the largest error for $n=80$ but shows better improvements over other estimators as $n$ increases.

\subsubsection{Effect of Diagonal Multipliers on ROCs}
In our experiments, we set an upper bound on the diagonal multiplier based on our theoretical analysis in (\ref{eq_bounded_nonlog_delta}). To investigate whether the AUCs could be significantly improved with very large diagonal multipliers, in Figure~\ref{plot_dm} we present the average ROC curves over 50 trials for the solution paths with $|\mathcal{J}|=5$ and varying diagonal multipliers. (The $y$ axes are truncated from below for better visualization.) The upper bound diagonal multipliers (\ref{eq_bounded_nonlog_delta}), which we used throughout Section~\ref{AUCs and Estimation Errors}, are highlighted in bold and italics in the legend on the right, namely $1.351$ for $n=80$ and $1.099$ for $n=1000$. The legends as well as the colors are sorted by the AUCs in decreasing order.

The results show that the AUCs reach a peak and decrease after some diagonal multiplier much larger than the theoretical upper bound. It is thus tempting to choose a very large diagonal multiplier to achieve high AUC. However, in real applications, instead of focusing on the AUC, one must choose one estimate from the solution path by cross validation and examine the performance of that estimate. For each diagonal multiplier, the square on the corresponding curve with the same color represents the TPR and FPR of the estimate picked by cross validation, averaged over 50 trials. It can be seen that the estimates for the upper bound multiplier chosen by cross validation produce the most reasonable TPRs and FPRs, with the corresponding squares much closer to the upper-left corner than those for larger diagonal multipliers.

\begin{figure}[!htp]
\centering
\vspace{-0.15in}
\subfloat[$n=80$, ROCs]
{\includegraphics[scale=0.29]{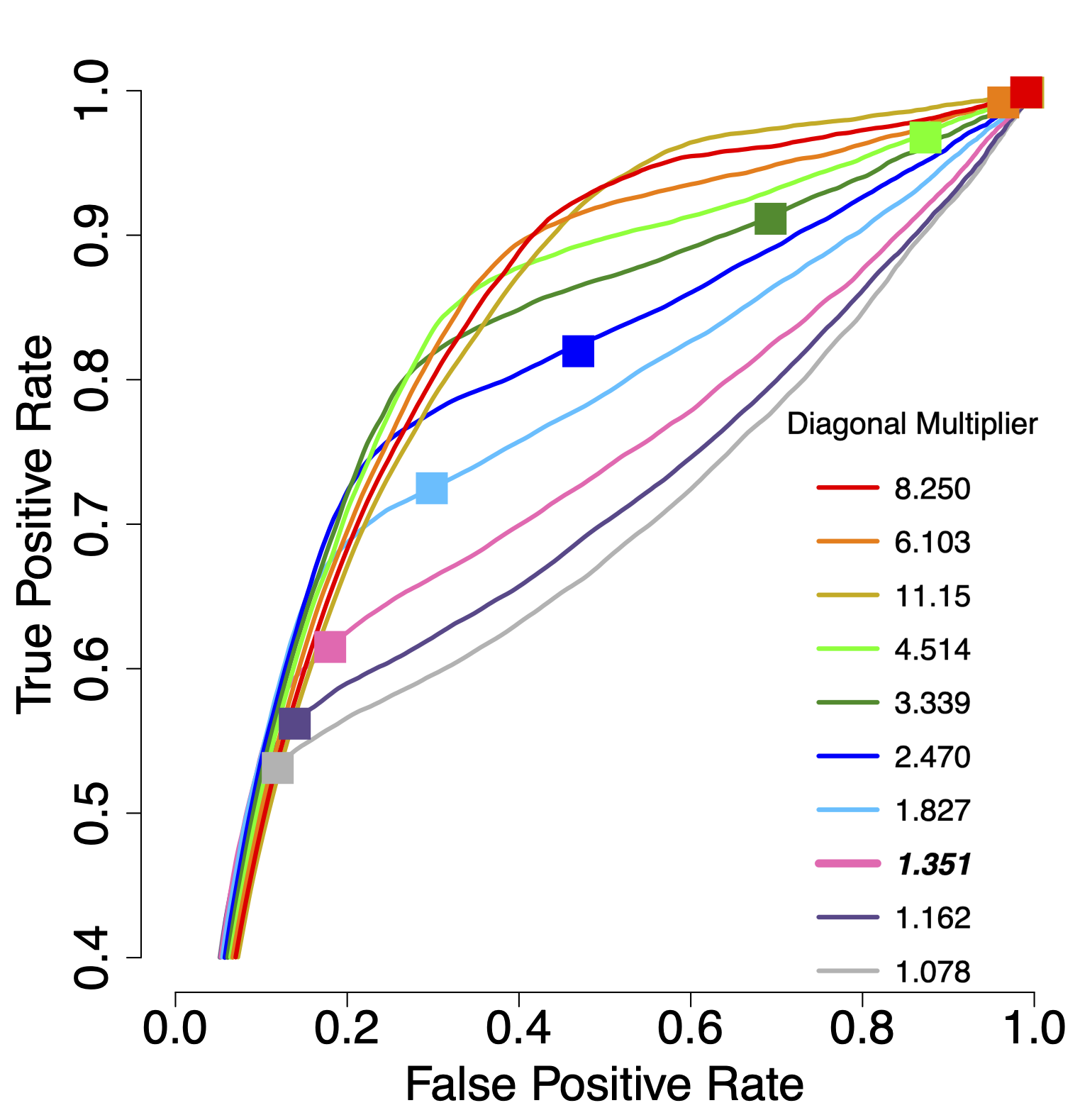}\hspace{0.15in}}
\subfloat[$n=1000$, ROCs]
{\includegraphics[scale=0.29]{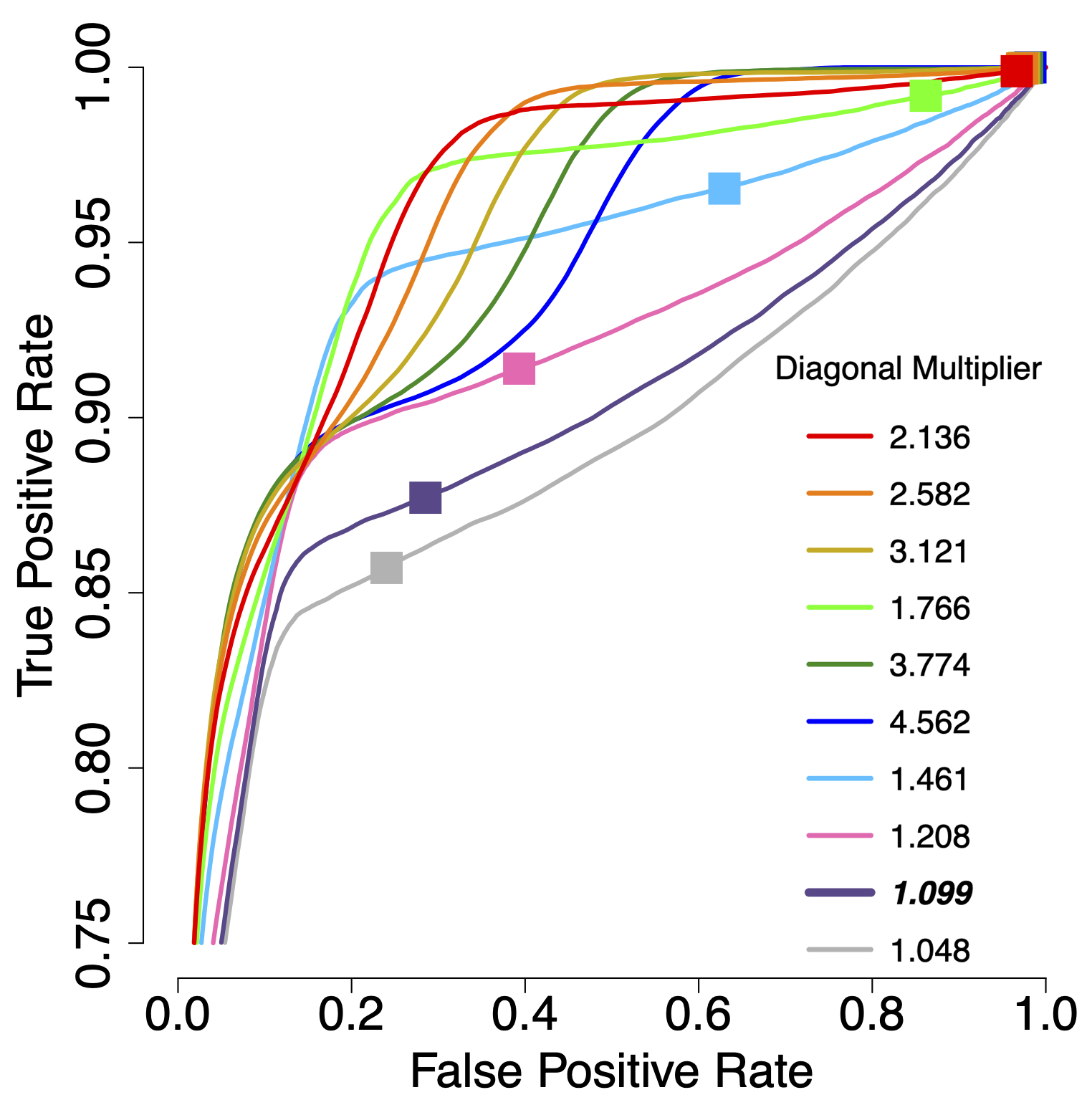}\hspace{-0.02in}}
\caption{ROCs averaged over 50 trials for edge recovery, with varying diagonal multipliers; $|\mathcal{J}|=5$. Squares correspond to the average of 50 TPRs and FPRs for estimates picked by cross validation. Note that the $y$ axes are truncated from below to better separate the curves for visualization.}\label{plot_dm}
\end{figure}

\section{Microbiome Data Analysis}
\label{sec:micr-data-analys}

To demonstrate our method, we analyze the human gut microbiome dataset originally studied in \citet{yat12} and also analyzed in \citet{wan19}. The dataset contains counts for $m=149$ taxa for $n=100$ healthy children and adults from Venezuela, Malawi, and US metropolitan areas \citep{wan19}. Following \citet{wan19}, we separate the 100 samples into two groups based on their age, with $n_1 = 67$ individuals of age $<3$ years as the first group, and $n_2 = 33$ of age $\geq 3$ as the second. 
%Assuming the true underlying distribution of the data is an $a$-$b$ model (Equation~\ref{eq_interaction_density2}) on $\Delta_{99}$, our goal is to use two different $a/b$ settings to investigate the difference in the edges in the $\mathbf{K}$ matrices for both groups, which we denote as $\mathbf{K}_1$ and $\mathbf{K}_2$.
To illustrate our models, we conduct a differential analysis of microbial interaction networks for these two groups \citep{sho21} by comparing their corresponding $\mathbf{K}$ matrices, which we denote as $\mathbf{K}_1$ and $\mathbf{K}_2$. 

We assume the true underlying distribution of the data is an $a$-$b$ model (Equation~\ref{eq_interaction_density2}) on $\Delta_{99}$ and use two settings: (i) $a=b=0$ assuming $\mathbf{K}_1\mathbf{1}_{100}=\mathbf{K}_2\mathbf{1}_{100}=\mathbf{0}_{100}$, namely the $A^{99}$ model in \citet{ait85} discussed in Section~\ref{Estimation for $A^{m-1}$ Models}; (ii) $a=b=1/2$, the exponential square-root model \citep{ino16} restricted to the simplex, whose estimation was covered in Section~\ref{Estimation-Simplex}. For each setting, we use permutation tests with $B=500$ trials, where in each trial we estimate $\mathbf{K}_1$ and $\mathbf{K}_2$ using two randomly shuffled datasets of sizes $n_1$ and $n_2$. We use cross validation to determine the tuning parameter $\lambda$ and set the diagonal multipliers to the upper bound in (\ref{eq_bounded_nonlog_delta}). Moreover, based on the results in Section~\ref{AUCs and Estimation Errors}, since the data is high-dimensional, it is beneficial to use $|\mathcal{J}|=5$, where we choose $\mathcal{J}=\{i*\lfloor m/5\rfloor|i=1,\ldots,5\}=\{29,58,87,116,145\}$. Similarly, given the simulation results in Section~\ref{AUCs and Estimation Errors}, the $h$ functions are also chosen as $h(x)=x^{2-a}$. 

Denote the estimates for each shuffled dataset as $\widehat{\mathbf{K}}_{1,(b)}$ and $\widehat{\mathbf{K}}_{2,(b)}$, $b=1,\ldots,B$, and those for the original dataset as $\widehat{\mathbf{K}}_{1}$ and $\widehat{\mathbf{K}}_{2}$.
We are interested in two tests:
\begin{enumerate}[(1)]
\item Global test: $\mathbf{K}_1\neq\mathbf{K}_2$, with $p$-value equal to 
%\[\frac{1}{B}\sum_{i=1}^B\mathds{1}\left(\left||S(\hat{\mathbf{K}}_1)|-|S(\hat{\mathbf{K}}_2)|\right|\leq\left||S(\hat{\mathbf{K}}_{1,(b)})|-|S(\hat{\mathbf{K}}_{2,(b)})|\right|\right);\]
\[\frac{1}{B}\sum_{i=1}^B\mathds{1}\left(\left|S(\hat{\mathbf{K}}_1)\triangle S(\hat{\mathbf{K}}_2)\right|\leq\left|S(\hat{\mathbf{K}}_{1,(b)})\triangle S(\hat{\mathbf{K}}_{2,(b)})\right|\right),\]
where  $|S(\mathbf{K}_1)\triangle S(\mathbf{K}_2)|$ denotes the total number of edge differences in the interaction matrices $\mathbf{K}_1$ and $\mathbf{K}_2$.
\item Local test: $\kappa_{1,jk}\neq\kappa_{2,jk}$ for $j,k=1,\ldots,149$, $j\neq k$, with $p$-values 
\[\frac{1}{B}\sum_{i=1}^B\mathds{1}\left(\left|\hat{\kappa}_{1,jk}-\hat{\kappa}_{2,jk}\right|\leq\left|\hat{\kappa}_{1,(b),jk}-\hat{\kappa}_{2,(b),jk}\right|\right)\]
adjusted for multiple testing using \citet{ben01}.
\end{enumerate}

For (1), the $p$-value we get assuming the $A^{99}$ models is 0.188, while the $p$-value assuming exponential square-root models is 0.31. Thus, under any conventional significance level and assuming either model, we do not reject the null hypothesis that there is no difference in the underlying interaction matrix. For (2), the resulting differential graphs for $\alpha=0.05$ is shown in Figure~\ref{data_diff_graphs}. 
%In Figure~\ref{data_deg_hist} we compare the histograms of node degrees of the estimated graphs for both groups on the original dataset (for the exponential square-root model the estimated graph for group 1 is empty). 
As expected, although the signal may not be strong enough for the global test to detect difference between the edge nonzero patterns of the two interaction matrices, pairwise comparisons of the values of the interaction matrices can still be informative and useful for identifying differences in conditional dependences in the two groups.

\begin{figure}[t]
\centering
\subfloat[$A^{99}$ model, layout optimized][$A^{99}$ model,\\layout optimized]{
\includegraphics[width=0.24\textwidth]{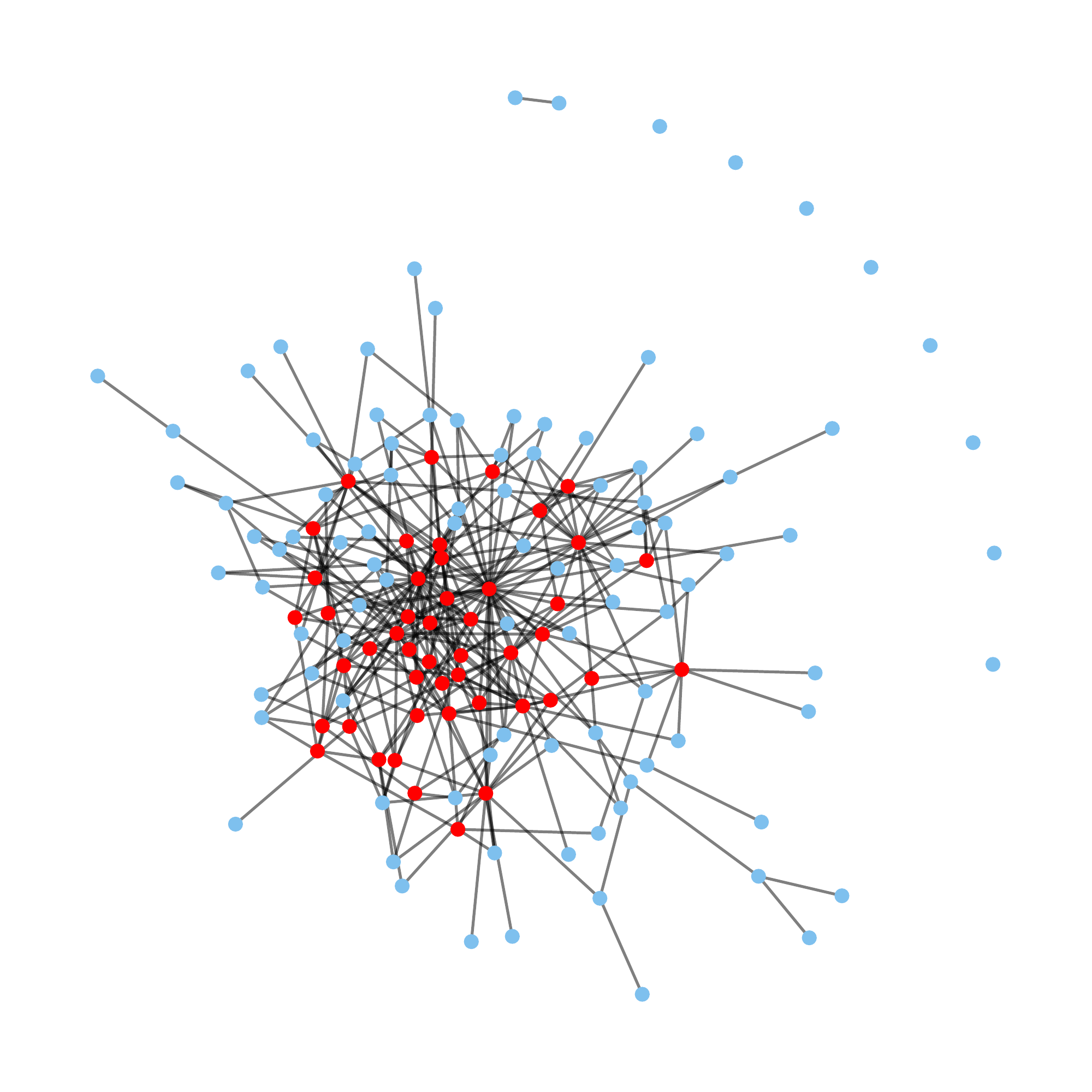}}%\hspace{1in}
\subfloat[$A^{99}$ model, no isolated nodes, layout optimized][$A^{99}$ model,\\no isolated nodes,\\layout optimized]{\includegraphics[width=0.24\textwidth]{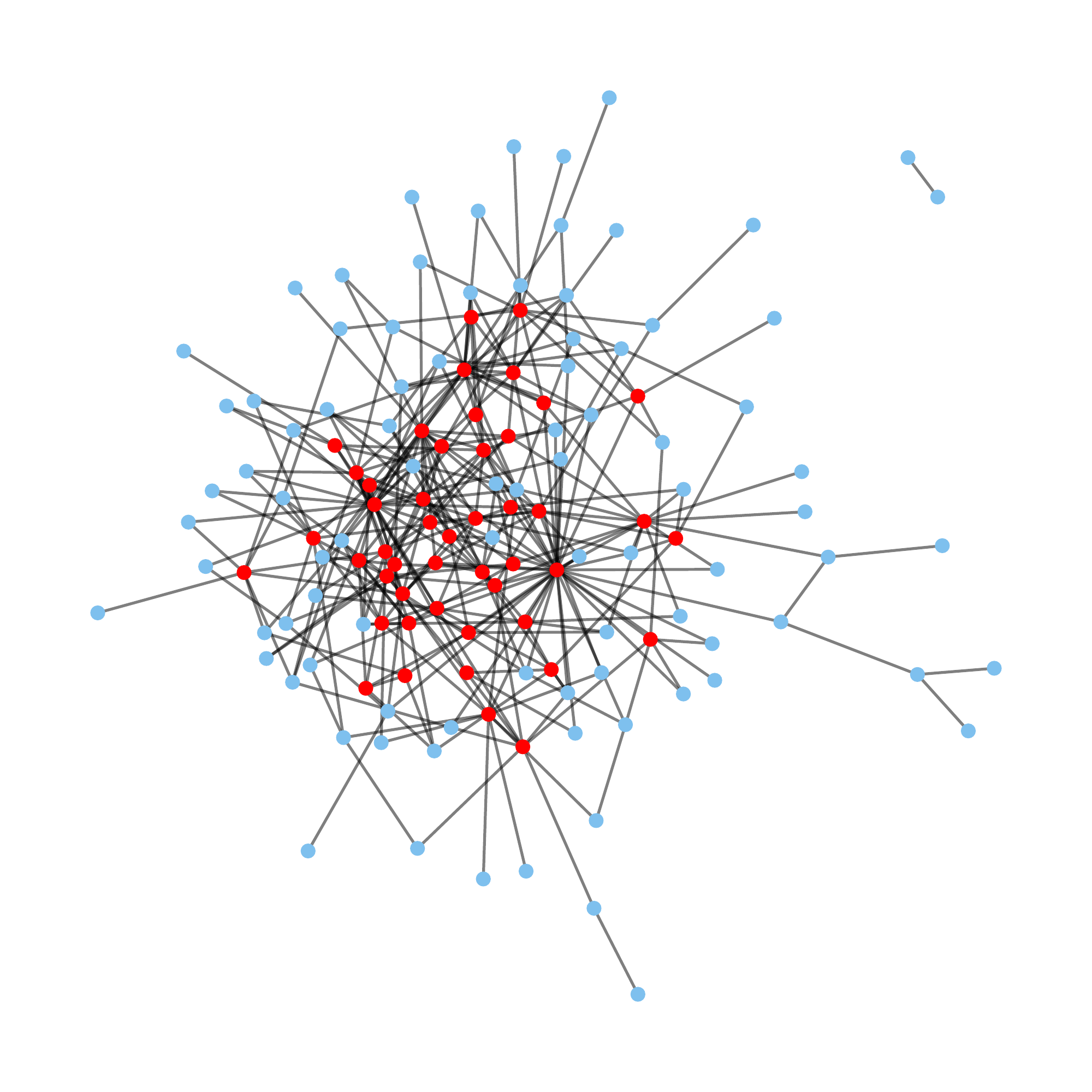}}
\subfloat[Exp model, same layout as (a)][Exp model,\\same layout as (a)]{\includegraphics[width=0.24\textwidth]{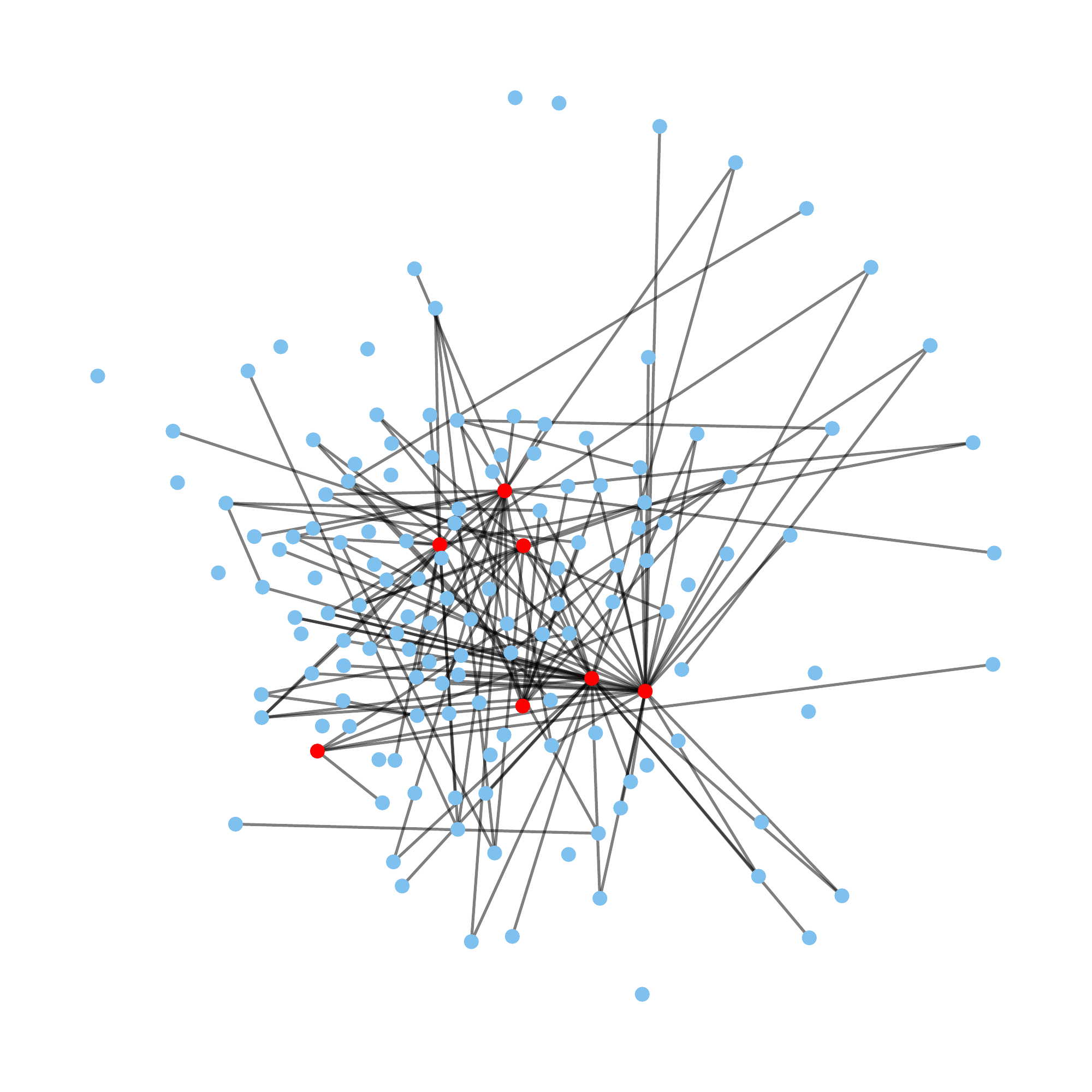}}%\hspace{1in}
\subfloat[Exp model, no isolated nodes, layout optimized][Exp model,\\no isolated nodes,\\layout optimized]{\includegraphics[width=0.24\textwidth]{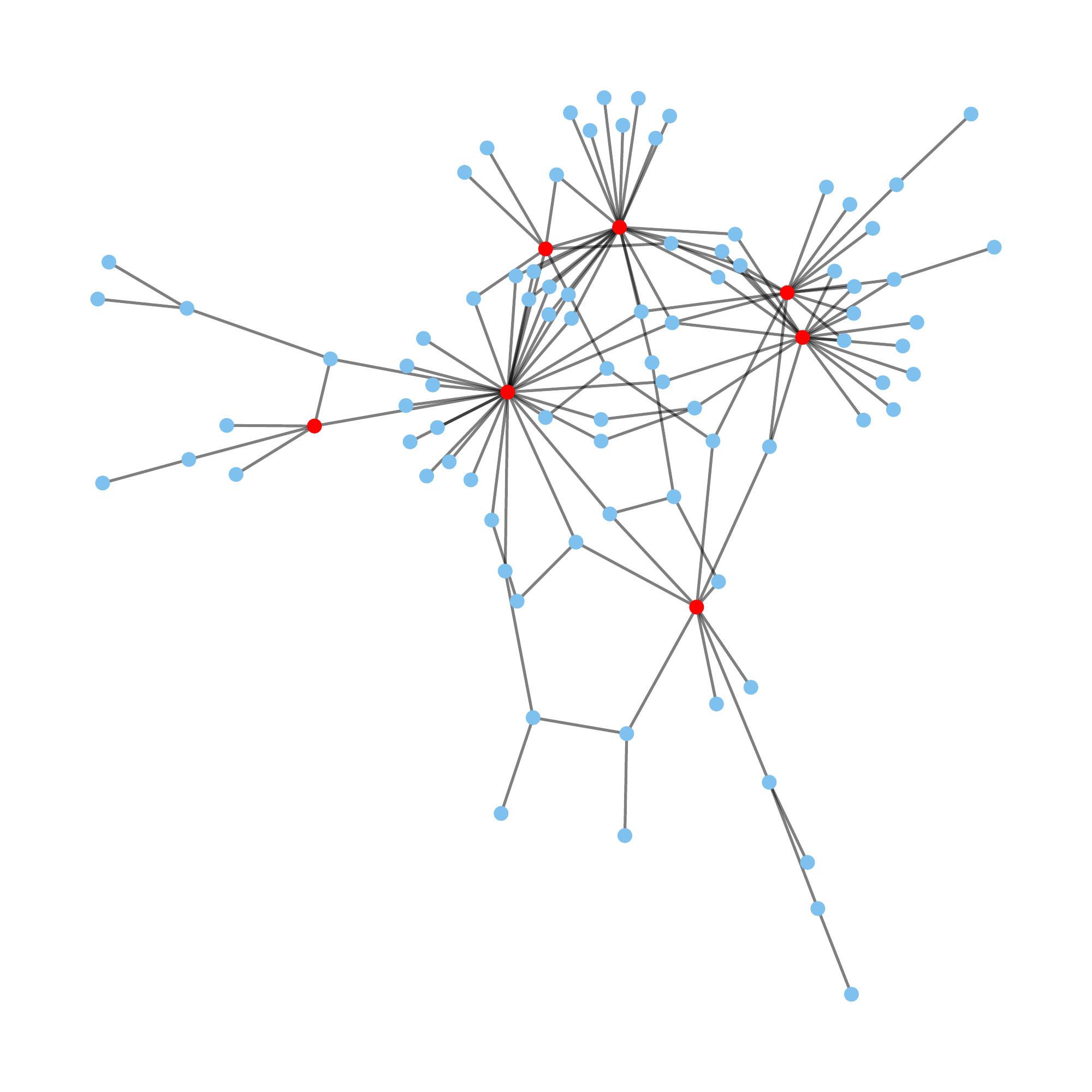}}
\caption{Differential graphs estimated by regularized generalized score matching estimator with permutation tests assuming the $A^{m-1}$ model (a, b) and the exponential square-root model (c, d). %\\%Nodes without differential edges in either graph are not shown. 
In (a) and (c), the layout is optimized for the graph (a). %\\
In (b) and (d), nodes with no differential edges are removed, with layout optimized for each plot. %\\
In all graphs, red points indicate nodes with degree at least 5 (``hub nodes'').}\label{data_diff_graphs}
\end{figure}
%\begin{figure}[!htp]
%\centering
%\subfloat[$A^{99}$, group 1]{\includegraphics[width=0.24\textwidth]{Plots/Data/hist_Ad_G1.pdf}}
%\subfloat[$A^{99}$, group 2]{\includegraphics[width=0.24\textwidth]{Plots/Data/hist_Ad_G2.pdf}}
%\subfloat[Exp, group 1]{\includegraphics[width=0.24\textwidth]{Plots/Data/hist_exp_G1.pdf}}
%\subfloat[Exp, group 2]{\includegraphics[width=0.24\textwidth]{Plots/Data/hist_exp_G2.pdf}}
%\caption{Node degree distributions of the estimated graphs on the original dataset, assuming the $A^{99}$ model (a, b) and the exponential square-root model (c, d).}\label{data_deg_hist}
%\end{figure}

%\newpage

%\subsection{Real Data Analysis}

\section{Discussion}
\label{sec:discussion}

Building on the ideas of \citet{hyv05,hyv07}, the method of
\citet{yu21} estimates densities supported on general domains using a
generalized score matching loss.
%; this is in extension of prior work by
%\citet{hyv05}, \citet{hyv07} and \citet{yu19}. As in these works,
The resulting estimator has a closed form when applied in
low-dimensional settings without any regularization and avoids calculating a possibly intractable normalizing constant.

The domains considered by \citet{yu21} are required to have positive
Lebesgue measure in order to guarantee consistent
estimation.
%
%However, the method in \citet{yu21} only works for domains with
%positive Lebesgue measure in $\mathbb{R}^m$, since the population
%version of the loss is minimized when the proposed density is equal
%to the truth almost surely.
In this paper, we demonstrate how to extend their method to the case
of compositional data on a probability simplex of general dimension.  Specifically, we show how profiling out the last component of $\boldsymbol{x}$ yields an effective methodology for a flexible class of interaction models for compositional data.

%As in \citet{yu21}, w
We focus on $a$-$b$ pairwise interaction models with density proportional to
$\exp\{-{\boldsymbol{x}^a}^{\top}\mathbf{K}\boldsymbol{x}^a/(2a)+\boldsymbol{\eta}^{\top}\boldsymbol{x}^b/b\}$,
where for $a=0$ we let
${\boldsymbol{x}^a}^{\top}\mathbf{K}\boldsymbol{x}^a/(2a)\equiv{\log\boldsymbol{x}}^{\top}\mathbf{K}\log
\boldsymbol{x}/2$ and for $b=0$,
$\boldsymbol{\eta}^{\top}\boldsymbol{x}^b/b\equiv\boldsymbol{\eta}^{\top}\log\boldsymbol{x}$.
For this class, our results detail the construction of estimators for
simplex domains, and provide additional details on the $A^{m-1}$ models
\citep{ait85} as an important example.

In our theoretical treatment, we show that for general $a$-$b$ models on $(m-1)$-dimensional simplex domains and with $a>0$, the sparsity pattern of the interaction matrix $\mathbf{K}$ may be recovered successfully when the sample size is of order $n=\Omega(\log m)$. This directly parallels similar results obtained in prior work for unconstrained domains. In the case of $a=0$, we require an additional multiplicative factor that may weakly depend on $m$.

In order to account for boundary effects, our method introduces a set of weights in the score matching loss. Through simulation studies, we confirm that weights derived from the choice of a function $\boldsymbol{h}(\boldsymbol{x})=(x_1^c,\ldots,x_m^c)$ with $c=\max\{2-a,0\}$ perform the best in most settings in terms of edge recovery, generalizing the conclusion in \citet{yu21}.

Two problems naturally emerge as topics for future work.  On the one hand,
it would be interesting to extend our theoretical results on $a$-$b$
models with $a=0$ in order to get a full understanding of the sample
complexity of our estimators; see the discussion after
Corollary~\ref{theorem_simplex_log_ggm}.  On the other hand, it would
be interesting to develop a more systematic way to deal with
Lebesgue-null sets beyond simplices.

% would be for a future research
% topic. Investigating the real sample complexity for simplex domains
% with $a=0$ would also be theoretically interesting.

%\newpage

\appendix
\renewcommand{\thesection}{\Alph{section}}
\section{Proofs}\label{Proofs}

\begin{proof}[Proof of Theorem \ref{thm_A_simp}]
Fix $j=1,\ldots,m-1$ and $\boldsymbol{x}_{-j,-m}\in\mathcal{S}_{-j,\Deltam_{-m}}$, i.e.~$\boldsymbol{x}_{-j,-m}\in\mathbb{R}_+^{m-2}$ such that $\mathbf{1}_{m-2}^{\top}\boldsymbol{x}_{-j,-m}<1$. In our discussion, for ease of notation, given $x_j$ and $\boldsymbol{x}_{-j,-m}$, we may still write $x_m\equiv1-\mathbf{1}_{m-1}^{\top}\boldsymbol{x}_{-m}$, a function in $x_j$ and $\boldsymbol{x}_{-j,-m}$, and for simplicity we may drop its dependence on $x_j$ and $\boldsymbol{x}_{-j,-m}$. Note that $\mathcal{C}_{j}(\boldsymbol{x}_{-j,-m})=\left(0,1-\mathbf{1}_{m-2}^{\top}\boldsymbol{x}_{-j,-m}\right)$.

(I) \emph{Case $a>0$ and $b\geq 0$:} For  (A.1), we need $p_0(\boldsymbol{x}_{-m})\partial_j\log p(\boldsymbol{x}_{-m})(h_j\circ\varphi_j)(\boldsymbol{x}_{-m})\to 0$ as $x_j\searrow 0^+$ and $x_j\nearrow 1-\mathbf{1}_{m-2}^{\top}\boldsymbol{x}_{-j,-m}$. As $x_j$ goes to any finite constant, by Equation \ref{eq_density_simplex} $p_0(\boldsymbol{x}_{-m})$ converges to a non-zero constant when $b>0$, or a finite constant times the limit of $x_j^{\eta_j}x_m(x_j)^{\eta_m}$ when $b=0$. Note that 
\begin{multline*}
\partial_j\log p(\boldsymbol{x}_{-m})= -\left(\boldsymbol{\kappa}_{-m,j}^{\top}\boldsymbol{x}_{-m}^a\right)x_j^{a-1}+\left(\boldsymbol{\kappa}_{-m,m}^{\top}\boldsymbol{x}_{-m}^a\right)x_m^{a-1}(x_j)\\
-x_m^a(x_j)\kappa_{jm}x_j^{a-1}
+x_m^{2a-1}(x_j)\kappa_{mm}+\eta_j x_j^{b-1}-\eta_m x_m^{b-1}(x_j).
\end{multline*}
\begin{enumerate}[i)]
\item If $b>0$, by arguments above we only consider $(h_j\circ\varphi_j)(\boldsymbol{x}_{-m})\partial_j\log p(\boldsymbol{x}_{-m})$.
\begin{enumerate}[a)]
\item As $x_j\searrow 0^+$, $x_m\nearrow 1-\mathbf{1}_{m-2}^{\top}\boldsymbol{x}_{-j,-m}>0$, so $\partial_j\log  p(\boldsymbol{x}_{-m})=\mathcal{O}\left(x_j^{a-1}\right)+\mathcal{O}\left(x_j^{b-1}\right)+\mathcal{O}(1)$. Thus we need $\alpha_j>\max\{0,1-a,1-b\}$ so that $(h_j\circ\varphi_j)(\boldsymbol{x}_{-m})\partial_j\log p(\boldsymbol{x}_{-m})\to 0$.
\item The case where $x_j\nearrow 1-\mathbf{1}_{m-2}^{\top}\boldsymbol{x}_{-j,-m}$ and $x_m\searrow 0^+$ is an analog of a) by noting that $\varphi_j$ is symmetric in $x_j$ about the midpoint of its domain $(1-\mathbf{1}_{m-2}^{\top}\boldsymbol{x}_{-j,-m})/2$.
\end{enumerate}
\item If $b=0$, we need $x_j^{\eta_j}x_m^{\eta_m}(x_j)(h_j\circ\varphi_j)(\boldsymbol{x}_{-m})\partial_j\log p(\boldsymbol{x}_{-m})\to 0$. Note that this quantity has the same form as in a) just with $\eta_j$ or $\eta_m$ added to the $a$ and $b$ ($=0$) in the exponents, we thus require $\alpha_j>\max\{0,1-a-\eta_j,1-a-\eta_m,1-\eta_j,1-\eta_m\}=\max\{0,1-\eta_j,1-\eta_m\}$.
\end{enumerate}
In conclusion, (A.1) requires $\alpha_j\geq \max\{0,1-a,1-b\}$ for $b>0$ or $\alpha_j>\max\left\{0,1-\eta_j\right\}$. For (A.2), we only prove the first integrability condition, since the second integrability condition is similar. For the first, we need to show that
\[\int_{\boldsymbol{x}_{-m}\succ\boldsymbol{0},\,\mathbf{1}_{m-1}^{\top}\boldsymbol{x}_{-m}<1}p_0(\boldsymbol{x}_{-m})(h_j\circ\varphi_j)(\boldsymbol{x}_{-m})\left(\partial_j\log p(\boldsymbol{x}_{-m})\right)^2\d\boldsymbol{x}_{-m}<+\infty.\]
Using the fact that $\mathbf{0}\prec\boldsymbol{x}_{-m}\prec\mathbf{1}$ and $0<x_j^a<1$, $0<x_m^a<1$ with the triangle inequality multiple times, we have
\begin{align*}
\left|\partial_j\log p(\boldsymbol{x}_{-m})\right|&\leq\sum_{i=1}^{m-1}\left(|\kappa_{ij}|x_j^{a-1}+|\kappa_{im}|x_m^{a-1}\right)+|\kappa_{jm}|x_j^{a-1}+x_m^{a-1}|\kappa_{mm}|+|\eta_j| x_j^{b-1}+|\eta_m|x_m^{b-1}\\
&\leq\mnorm{\mathbf{K}}_1x_j^{a-1}+\mnorm{\mathbf{K}}_1x_m^{a-1}+|\eta_j|x_j^{b-1}+|\eta_m|x_m^{b-1},
\end{align*}
where $\mnorm{\mathbf{K}}_1\equiv\max_{j=1,\ldots,m}\sum_{i=1}^m|\kappa_{ij}|$. We again consider the following two cases.

\begin{enumerate}[i)]
\item If $b>0$, $p_0(\boldsymbol{x}_{-m})$ is bounded by an absolute constant, which we therefore ignore. We first fix $\boldsymbol{x}_{-j,-m}$ and denote $y_j(\boldsymbol{x}_{-j,-m})\equiv 1-\mathbf{1}_{m-2}^{\top}\boldsymbol{x}_{-j,-m}=x_j+x_m$. Then, writing $h_{\varphi,j,\boldsymbol{x}}\equiv (h_j\circ\varphi_j)(\boldsymbol{x}_{-m})$,
\begin{align*}
&\,\int_0^{y_j}h_{\varphi,j,\boldsymbol{x}}(\partial_j\log p(\boldsymbol{x}_{-m}))^2\d x_j\\
\leq&\,\int_0^{y_j}h_{\varphi,j,\boldsymbol{x}}\left[\mnorm{\mathbf{K}}_1 \left(x_j^{a-1}+x_m^{a-1}(x_j)\right)+|\eta_j|x_j^{b-1}+|\eta_m|x_m^{b-1}(x_j)\right]^2\d x_j\\
\leq&\,\int_0^{y_j/2}h_{\varphi,j,\boldsymbol{x}}\left(\mnorm{\mathbf{K}}_1 x_j^{a-1}+|\eta_j|x_j^{b-1}+c_{1,m}(\boldsymbol{x}_{-j,-m})\right)^2\d x_j\\
&\quad\quad+\int_{y_j/2}^{y_j}h_{\varphi,j,\boldsymbol{x}}\left(\mnorm{\mathbf{K}}_1x_m^{a-1}(x_j)+|\eta_m|x_m^{b-1}(x_j)+c_{1,j}(\boldsymbol{x}_{-j,-m})\right)^2\d x_j\\
=&\,\int_0^{y_j/2}h_j(x_j)\left(\mnorm{\mathbf{K}}_1 x_j^{a-1}+|\eta_j|x_j^{b-1}+c_{1,m}(\boldsymbol{x}_{-j,-m})\right)^2\d x_j\\
&\quad\quad\quad+\int_0^{y_j/2} h_j(x_j)\left(\mnorm{\mathbf{K}}_1 x_j^{a-1}+|\eta_m|x_j^{b-1}+c_{1,j}(\boldsymbol{x}_{-j,-m})\right)^2\d x_j
\end{align*}
where in the last step we used change of variable $x_j\leftarrow x_m(x_j)=y_j-x_j$ for the second term, and where
\begin{align*}
0<c_{1,j}\equiv\max_{y_j/2\leq x_j\leq y_j}\left(\mnorm{\mathbf{K}}_1x_j^{a-1}+|\eta_j|x_j^{b-1}\right)=\mathcal{O}\left(y_j^{a-1}\right)+\mathcal{O}\left(y_j^{b-1}\right)+\mathcal{O}(1)<+\infty,
\end{align*}
with $\mathcal{O}$ depending on $\mathbf{K}$ and $\boldsymbol{\eta}$. We thus have (dropping the dependency $y_j\equiv y_j(\boldsymbol{x}_{-j,-m})$ to save space)
\begin{align*}
&\,\int_{\substack{\boldsymbol{x}_{-m}\succ\boldsymbol{0},\\\mathbf{1}_{m-1}^{\top}\boldsymbol{x}_{-m}<1}}h_{\varphi,j,\boldsymbol{x}}\left(\partial_j\log p(\boldsymbol{x}_{-m})\right)^2\d\boldsymbol{x}_{-m}\\
=&\,\int_{\substack{\boldsymbol{x}_{-j,-m}\succ\mathbf{0},\\y_j>0}}\int_0^{y_j} h_{\varphi,j,\boldsymbol{x}}(\partial_j\log p(\boldsymbol{x}_{-m}))^2 \d x_j\d\boldsymbol{x}_{-j,-m}\\
\leq&\,\int_{\substack{\boldsymbol{x}_{-j,-m}\succ\mathbf{0},\\y_j>0}}\int_0^{y_j/2}h_j(x_j)\times\\
&\quad \left(\mathcal{O}\left(x_j^{a-1}\right)+\mathcal{O}\left(x_j^{b-1}\right)+\mathcal{O}\left(y_j^{a-1}\right)+\mathcal{O}\left(y_j^{b-1}\right)+\mathcal{O}(1)\right)^2\d x_j\d\boldsymbol{x}_{-j,-m}\\
\leq&\,\int_{\substack{\boldsymbol{x}_{-j,-m}\succ\mathbf{0},\\y_j>0}}\int_0^{y_j/2}h_j(x_j)\left(\mathcal{O}\left(x_j^{a-1}\right)+\mathcal{O}\left(x_j^{b-1}\right)+\mathcal{O}(1)\right)^2\d x_j\d\boldsymbol{x}_{-j,-m}\\
&+\int_{\substack{\boldsymbol{x}_{-j,-m}\succ\mathbf{0},\\y_j>0}}\int_0^{y_j/2}h_j(x_j)\left(\mathcal{O}\left(y_j^{a-1}\right)+\mathcal{O}\left(y_j^{b-1}\right)+\mathcal{O}(1)\right)^2\d x_j\d\boldsymbol{x}_{-j,-m}\\
\leq&\,\int_{\mathbf{1}\succ\boldsymbol{x}_{-j,-m}\succ\mathbf{0}}\int_0^{1}x_j^{\alpha_j}\left(\mathcal{O}\left(x_j^{a-1}\right)+\mathcal{O}\left(x_j^{b-1}\right)+\mathcal{O}(1)\right)^2\d x_j\d\boldsymbol{x}_{-j,-m}\\
&\quad\quad+\int_{\substack{\boldsymbol{x}_{-j,-m}\succ\mathbf{0},\\y_j>0}} \frac{\left(y_j/2\right)^{\alpha_j+1}}{\alpha_j+1} \left(\mathcal{O}\left(y_j^{a-1}\right)+\mathcal{O}\left(y_j^{b-1}\right)+\mathcal{O}(1)\right)^2\d\boldsymbol{x}_{-j,-m}\\
&\leq\int_0^1\mathcal{O}\left(x_j^{2a-2+\alpha_j}\right)+\mathcal{O}\left(x_j^{2b-2+\alpha_j}\right)+\mathcal{O}\left(x_j^{\alpha_j}\right)\d x_j\\
&\quad+\sum_{\substack{p\in\{2a-1+\alpha_j,\\2b-1+\alpha_j,\,\alpha_j+1\}}}\int_{\substack{\boldsymbol{x}_{-j,-m}\succ\mathbf{0},\\\mathbf{1}_{m-2}^{\top}\boldsymbol{x}_{-j,-m}<1}} \mathcal{O}\left(\left(1-\mathbf{1}_{m-2}^{\top}\boldsymbol{x}_{-j,-m}\right)^p\right)\d \boldsymbol{x}_{-j,-m}\\
&=\int_0^1o\left(x_j^{a-1}\right)+o\left(x_j^{b-1}\right)+o\left(x^0\right)\d x_j+\sum_{p\in\{a,b,1\}}\mathcal{O}\left(\Gamma(p+1)/\Gamma(p+m-1)\right)<+\infty
\end{align*}
for $\alpha_j>\max\{0,1-a,1-b\}$, where the second term of the last quantity follows from the normalizing constant of the Dirichlet distribution with parameters $(\mathbf{1}_{m-2},p+1)$.
\item If $b=0$, then $p_0(\boldsymbol{x}_{-m})$ is bounded by
  $C_2\prod_{j=1}^m x_j^{\eta_{0j}}$, where $C_2$ is the product of the
  inverse
  normalizing constant of $p_0(\boldsymbol{x}_{-m})$ and the supremum
 $\sup_{\boldsymbol{x}\succ\boldsymbol{0},\mathbf{1}^{\top}\boldsymbol{x}=1}\exp\left(-{\boldsymbol{x}^a}^{\top}\mathbf{K}_0\boldsymbol{x}^a/(2a)\right)$, a positive and finite constant. Then by the same reasoning as in i), with $y_j(\boldsymbol{x}_{-j,-m})\equiv 1-\mathbf{1}_{m-2}^{\top}\boldsymbol{x}_{-j.-m}$ and noting that $\boldsymbol{\eta}\succ -\mathbf{1}_m$,
\begin{align*}
&\,\int_{\substack{\boldsymbol{x}_{-m}\succ\boldsymbol{0},\\ \mathbf{1}_{m-1}^{\top}\boldsymbol{x}_{-m}<1}}p_0(\boldsymbol{x}_{-m})(h_j\circ\varphi_j)(\boldsymbol{x}_{-m})\left(\partial_j\log p(\boldsymbol{x}_{-m})\right)^2\d\boldsymbol{x}_{-m}\\
\leq&\,\int_{\substack{\boldsymbol{x}_{-j,-m}\succ\mathbf{0},\\y_j>0}}\int_0^{y_j/2}C_2\prod_{k=1}^m x_k^{\eta_{0k}}h_j(x_j)\left(\mathcal{O}\left(x_j^{a-1}\right)+\mathcal{O}\left(x_j^{-1}\right)+\mathcal{O}(1)\right)^2\d x_j\d\boldsymbol{x}_{-j,-m}\\
&+\int_{\substack{\boldsymbol{x}_{-j,-m}\succ\mathbf{0},\\y_j>0}}\int_0^{y_j/2}C_2\prod_{k=1}^m x_k^{\eta_{0k}}h_j(x_j)\left(\mathcal{O}\left(y_j^{a-1}\right)+\mathcal{O}\left(y_j^{-1}\right)+\mathcal{O}(1)\right)^2\d x_j\d\boldsymbol{x}_{-j,-m}\\
\leq&\,C_2\prod_{k\neq j,m}\int_0^1x_k^{\eta_{0k}}\d x_k\int_0^{1}x_j^{\eta_{0k}+\alpha_j}\left(\mathcal{O}\left(x_j^{-1}\right)+\mathcal{O}(1)\right)^2\d x_j\\
&\quad\quad\quad+C_2\int_{\boldsymbol{x}_{-j,-m}\succ\mathbf{0},\,y_j>0}\frac{(y_j/2)^{\alpha_j+1+\eta_{0j}}}{\alpha_j+1+\eta_{0j}}\times\prod_{k\neq j,m}x_k^{\eta_{0k}}\left(\mathcal{O}\left(y_j^{-1}\right)+\mathcal{O}(1)\right)^2\d\boldsymbol{x}_{-j,-m}\\
\leq&\,C_2\prod_{k\neq j,m}\frac{1}{\eta_{0k}+1}\int_0^1\mathcal{O}\left(x_j^{-1}\right)+\mathcal{O}\left(x_j\right)\d x_j+\sum_{p\in\{0,2\}}C_2\int_{\substack{\boldsymbol{x}_{-j,-m}\succ\boldsymbol{0},\\y_j>0}}\prod_{k\neq j,m}x_k^{\eta_{0k}}\mathcal{O}\left(y_j^p\right)\d\boldsymbol{x}_{-j,-m}\\
<&\,+\infty
\end{align*}
since the integral in the second term is the inverse normalizing constant of the Dirichlet distribution with parameters $(\boldsymbol{\eta}_{0,-j,-m}+\mathbf{1}_{m-2},p+1)$,  i.e.~$\frac{\Gamma(p+1)\prod_{k\neq j,m}\Gamma(\eta_{0k}+1)}{\Gamma\left(\mathbf{1}_{m-2}^{\top}\boldsymbol{\eta}_{0,-j,-m}+p+m-1\right)}<+\infty$.
\end{enumerate}
This ends the proof for the first integrability condition for (A.1) for $a>0$. For the second half, the integrand we consider is 
\[p_0(\boldsymbol{x}_{-m})\left|\partial_j\left(\partial_j\log p(\boldsymbol{x}_{-m})(h_j\circ\varphi)(\boldsymbol{x}_{-m})\right)\right|.\]
The arguments are similar to those for the first condition, where we first  bound $\partial_{jj}\log p(\boldsymbol{x})$ using sums of products of powers of $\boldsymbol{x}$. Then for each fixed $\boldsymbol{x}_{-j,-m}$ we split the domain of  $x_j$ into two halves and deal with the potential singularity at  $x_j\searrow 0^+$, where one can show that the requirement on  $\alpha_j$ is just enough for the integrand to be $o\left(x_j^{-1}\right)$ and thus the integral is finite. The detailed proof is tedious and is omitted.

(II) \emph{Case $a=0$ and $b\geq 0$:} First consider $b=0$. We again write $y_j\equiv 1-\mathbf{1}_{m-2}^{\top}\boldsymbol{x}_{-j,-m}$. Fixing $\boldsymbol{x}_{-j,-m}\in\mathcal{S}_{-j,(\Deltam0_{-m}}$,
\begin{align*}
&\,p_{0}(\boldsymbol{x}_{-m})\left|\partial_j\log p(\boldsymbol{x}_{-m})\right|\nonumber\\
\propto&\,\exp\left[-\frac{1}{2}{\log(\boldsymbol{x}_{-m})}^{\top}\mathbf{K}_{0,-m,-m}\log(\boldsymbol{x}_{-m})-{\log(\boldsymbol{x}_{-m})}^{\top}\boldsymbol{\kappa}_{0,-m,m}\log x_m\right.\nonumber\\
&\quad\quad\quad\quad\quad\quad\quad\quad\left.-\frac{1}{2}\kappa_{0,m,m}\left(\log x_m\right)^{2}+\boldsymbol{\eta}_{0,-m}^{\top}\log(\boldsymbol{x}_{-m})+\eta_{0,m}\log x_m\right]\times\nonumber\\
&\quad\quad\left| -\boldsymbol{\kappa}_{-m,j}^{\top}\log(\boldsymbol{x}_{-m})/x_j+\boldsymbol{\kappa}_{-m,m}^{\top}\log(\boldsymbol{x}_{-m})/x_m-\kappa_{jm}\log x_m/x_j\right.\nonumber\\
&\quad\quad\quad\quad\quad\quad\quad\quad\quad\quad\quad\quad\quad\quad\quad\quad\quad\left.+\kappa_{mm}\log x_m/x_m+\eta_j/x_j-\eta_m/x_m\right|.\nonumber\\
\leq&\,\prod_{k\neq j,m}\exp\left[-\frac{N_{\mathbf{K}_{0,-m,-m}}}{2}(\log x_k)^2+\eta_{0,k}\log x_k\right]\exp\left[-\frac{N_{\mathbf{K}_{0,-m,-m}}}{2}(\log x_j)^2+\right.\nonumber\\
&\quad\left.\eta_{0,j}\log x_j-{\log(\boldsymbol{x}_{-m})}^{\top}\boldsymbol{\kappa}_{0,-m,m}\log x_m-\frac{\kappa_{0,m,m}\left(\log x_m\right)^{2}}{2}+\eta_{0,m}\log x_m\right]\nonumber\\
&\quad\quad\quad\times\left| -\boldsymbol{\kappa}_{-m,j}^{\top}\log(\boldsymbol{x}_{-m})/x_j+\boldsymbol{\kappa}_{-m,m}^{\top}\log(\boldsymbol{x}_{-m})/x_m-\kappa_{jm}\log x_m/x_j\right.\nonumber\\
&\quad\quad\quad\quad\quad\quad\quad\quad\quad\quad\quad\quad\quad\quad\quad\quad\quad\left.+\kappa_{mm}\log x_m/x_m+\eta_j/x_j-\eta_m/x_m\right|.%\label{eq_density_bound_Ad}
\end{align*}
which is $\mathcal{O}\left(\exp\left(\mathcal{O}\left((\log x_j)^2\right)+\mathcal{O}\left(\log x_j\right)+\mathcal{O}\left(\log\log x_j\right)\right)\right)$ as $x_j\searrow 0^+$. Since the coefficient on the leading term is negative the entire term goes to 0. By symmetry the quantity goes to zero also when $x_j\nearrow y_j^-$. Thus, (A.1) holds for any $\alpha_j\geq 0$.

Similarly, $p_0(\boldsymbol{x}_{-m})(h_j\circ\varphi_j)(\boldsymbol{x}_{-m})(\partial_j\log p(\boldsymbol{x}_{-m}))^2$ and\\$p_0(\boldsymbol{x}_{-m})\left|\partial\left(\partial_j\log p(\boldsymbol{x}_{-m})(h_j\circ\varphi)(\boldsymbol{x}_{-m})\right)\right|$ are $\prod_{j=1}^m\mathcal{O}\left(\exp\left(-(\log x_j)^2\right)\right)$ times a polynomial, and are thus bounded and go to 0 at the boundaries of $\Deltam_{-m}$. Thus extending the integrands to 0 at the boundaries, they are continuous and bounded in the compact $\overline{\Deltam_{-m}}$, so integrals\\
$\int_{\Deltam_{-m}}p_0(\boldsymbol{x}_{-m})(h_j\circ\varphi_j)(\boldsymbol{x}_{-m})(\partial_j\log p(\boldsymbol{x}_{-m}))^2\d\boldsymbol{x}_{-m}$ and\\
$\int_{\Deltam_{-m}}p_0(\boldsymbol{x}_{-m})\left|\partial\left(\partial_j\log p(\boldsymbol{x}_{-m})(h_j\circ\varphi)(\boldsymbol{x}_{-m})\right)\right|\d\boldsymbol{x}_{-m}$ are finite, thus proving (A.2).

For $a=0$ and $b>0$, $\boldsymbol{\eta}\preceq\boldsymbol{0}$, and the proof is similar and is omitted. In particular, $p_0(\boldsymbol{x}_{-m})$ is bounded by that with $a=0$, $b=0$, $\boldsymbol{\eta}\equiv\mathbf{0}_m$, and thus its product with any polynomial is bounded and goes to 0 at the boundary of $\overline{\Deltam_{-m}}$.
\end{proof}

\begin{proof}[Proof of Theorem \ref{thm_simplex_identifiability}]
For notational simplicity, denote $\overset{\sim}{\mathbf{K}}\equiv\mathbf{K}_1-\mathbf{K}_2$ with columns $\overset{\sim}{\boldsymbol{\kappa}}_{,1}\ldots,\overset{\sim}{\boldsymbol{\kappa}}_{,m}$, and denote $\overset{\sim}{\boldsymbol{\eta}}\equiv\boldsymbol{\eta}_1-\boldsymbol{\eta}_2$. Assume that either $\overset{\sim}{\mathbf{K}}\neq\mathbf{0}_{m\times m}$ or $\overset{\sim}{\boldsymbol{\eta}}\neq\boldsymbol{0}_m$, otherwise there is nothing to prove. By Equation \ref{eq_grad_logp_simplex}, writing $x_m\equiv 1-\mathbf{1}_{m-1}^{\top}\boldsymbol{x}_{-m}$ and $\boldsymbol{x}=(\boldsymbol{x}_{-m};x_m)$ and taking the gradient of the log of both sides of the equation with respect to $x_j$, $j=1,\ldots,m-1$, we have
\begin{equation}\label{eq_proof_simplex_identifiability}
\left(x_j^{a-1}\overset{\sim}{\boldsymbol{\kappa}}_{,j}-x_m^{a-1}\overset{\sim}{\boldsymbol{\kappa}}_{,m}\right)^{\top}\boldsymbol{x}^a=\overset{\sim}{\eta}_j x_j^{b-1}-\overset{\sim}{\eta}_m x_m^{b-1}
\end{equation}
for all $\boldsymbol{x}_{-m}\in\Deltam_{-m}\equiv\left\{\boldsymbol{x}_{-m}\in\mathbb{R}_+^{m-1}\right|\left.\boldsymbol{x}_{-m}\succ\boldsymbol{0},\mathbf{1}_{m-1}^{\top}\boldsymbol{x}_{-m}<1\right\}$. In the following when $a=0$ by $x^a$ we mean $\log(x)$, and by $x^{a-1}$ we mean $1/x$ and we do not treat this case differently as the same expressions still hold.
\begin{enumerate}[(I)]
\item Suppose $\left(x_j^{a-1}\overset{\sim}{\boldsymbol{\kappa}}_{,j}-x_m^{a-1}\overset{\sim}{\boldsymbol{\kappa}}_{,m}\right)_{-j,-m}=\mathbf{0}_{m-2}$ for all $\boldsymbol{x}_{-m}\in\Deltam_{-m}$ and $x_m=1-\mathbf{1}_{-m}^{\top}\boldsymbol{x}_{-m}$ or $\tilde{\eta}_j$.
\begin{enumerate}[(i)]
\item Suppose $a=1$, then Equality \ref{eq_proof_simplex_identifiability} becomes\\ $\left(\overset{\sim}{\kappa}_{j,j}-\overset{\sim}{\kappa}_{j,m}\right)x_j+\left(\overset{\sim}{\kappa}_{m,j}-\overset{\sim}{\kappa}_{m,m}\right)x_m=\overset{\sim}{\eta}_jx_j^{b-1}-\overset{\sim}{\eta}_mx_m^{b-1}$, and we must have $b=2$ or $b=1$ or $\overset{\sim}{\eta}_j=\overset{\sim}{\eta}_m=0$, i.e.~$\boldsymbol{\eta}_1=\boldsymbol{\eta}_2$.
\item Suppose $a\neq 1$. The assumption implies $\left(\overset{\sim}{\boldsymbol{\kappa}}_{j}\right)_{-j,-m}=\left(\overset{\sim}{\boldsymbol{\kappa}}_{m}\right)_{-j,-m}=\mathbf{0}_{m-2}$. Then Equality \ref{eq_proof_simplex_identifiability} becomes $\left(x_j^{a-1}\overset{\sim}{\kappa}_{j,j}-x_m^{a-1}\overset{\sim}{\kappa}_{j,m}\right)x_j^a+$\\$\left(x_j^{a-1}\overset{\sim}{\kappa}_{m,j}-x_m^{a-1}\overset{\sim}{\kappa}_{m,m}\right)x_m^a=\overset{\sim}{\eta}_j x_j^{b-1}-\overset{\sim}{\eta}_m x_m^{b-1}$. Since $x_j>0$ and $x_m>0$ are arbitrary (as $\mathbf{1}_{-m}^{\top}\boldsymbol{x}_{-m}$ can vary) as long as $x_j+x_m<1$, the cross terms must not exist, and so $\overset{\sim}{\kappa}_{j,m}=\overset{\sim}{\kappa}_{m,j}=0$. It thus follows that $\tilde{\boldsymbol{\kappa}}_{-j,j}=\tilde{\boldsymbol{\kappa}}_{-m,m}=\mathbf{0}_{m-1}$ and hence $\overset{\sim}{\mathbf{K}}$ is diagonal, and the original equality becomes $-\frac{1}{2}\mathrm{diag}(\overset{\sim}{\mathbf{K}})^{\top}\left(\boldsymbol{x}^{a}\right)^2+\overset{\sim}{\boldsymbol{\eta}}^{\top}\boldsymbol{x}^b=0$, in which by $\boldsymbol{x}^{0}$ we mean $\log(\boldsymbol{x})$. Thus we must have $2a=b\neq 0$ and $\mathbf{K}_1-\mathbf{K}_2=2\boldsymbol{\eta}_1-2\boldsymbol{\eta}_2$.
\end{enumerate}
\item Now fix $x_j$ and $x_m$ such that $\left(x_j^{a-1}\overset{\sim}{\boldsymbol{\kappa}}_{,j}-x_m^{a-1}\overset{\sim}{\boldsymbol{\kappa}}_{,m}\right)_{-j,-m}\neq \mathbf{0}_{m-2}$. Note that $\mathbf{1}_{m-2}^{\top}\boldsymbol{x}_{-j,-m}=1-x_j-x_m$ is also fixed. Now the right-hand side and the first vector on the left-hand side of Equality \ref{eq_proof_simplex_identifiability} are both constant, while $\boldsymbol{x}_{-j,-m}$ is allowed to vary freely as long as their sum is fixed. A necessary condition of Equality \ref{eq_proof_simplex_identifiability} is thus
\begin{equation}\label{eq_proof_simplex_identifiability2}
\left(x_j^{a-1}\overset{\sim}{\boldsymbol{\kappa}}_{,j}-x_m^{a-1}\overset{\sim}{\boldsymbol{\kappa}}_{,m}\right)_{-j,-m}^{\top}\boldsymbol{x}_{-j,-m}^a=\text{const depending on }x_j\text{ and }x_m\text{ only}
\end{equation}
for all $\boldsymbol{x}_{-j,-m}^a\in\mathcal{U}_{x_j,x_m}\equiv\{\boldsymbol{y}^a:\boldsymbol{y}\succ \mathbf{0}_{m-2},\mathbf{1}_{m-2}\boldsymbol{y}=\mathbf{1}_{m-2}^{\top}\boldsymbol{x}_{-j,-m}=1-x_j-x_m\}$.

Suppose by contradiction that $a\neq 1$. Then $\mathcal{U}_{x_j,x_m}$ is not entirely on a hyperplane, and by assumption $\left(x_j^{a-1}\overset{\sim}{\boldsymbol{\kappa}}_{,j}-x_m^{a-1}\overset{\sim}{\boldsymbol{\kappa}}_{,m}\right)_{-j,-m}$ is not a zero vector, so the equality cannot hold. We thus have $a=1$, so that $\mathcal{U}_{x_j,x_m}$ lies on the hyperplane $\mathcal{H}_{x_j,x_m}\equiv\{\boldsymbol{y}:\mathbf{1}_{m-2}^{\top}\boldsymbol{y}=1-x_j-x_m\}$. Since $\mathcal{H}_{x_j,x_m}\equiv\{c\boldsymbol{x}_{-j,-m}:c\in\mathbb{R},\boldsymbol{x}\in\mathcal{U}_{x_j,x_m}\}$, Equality \ref{eq_proof_simplex_identifiability2} must hold for all $\boldsymbol{x}_{-j,-m}$ in the hyperplane $\mathcal{H}_{x_j,x_m}$, and by the assumption that $\left(\overset{\sim}{\boldsymbol{\kappa}}_{,j}-\overset{\sim}{\boldsymbol{\kappa}}_{,m}\right)_{-j,-m}$ is nonzero it must be a constant multiple of $\mathbf{1}_m$, and the right-hand side of Equality \ref{eq_proof_simplex_identifiability2} is hence $c_0(1-x_j-x_m)$ for some absolute constant $c_0\neq 0$ assuming $\left(\overset{\sim}{\boldsymbol{\kappa}}_{,j}-\overset{\sim}{\boldsymbol{\kappa}}_{,m}\right)_{-j,-m}=c_0\mathbf{1}_m$.  Plugging this back in Equality \ref{eq_proof_simplex_identifiability} we get
\[c_0(1-x_j-x_m)+\left(\overset{\sim}{{\kappa}}_{j,j}-\overset{\sim}{\kappa}_{j,m}\right)x_j+\left(\overset{\sim}{{\kappa}}_{m,j}-\overset{\sim}{\kappa}_{m,m}\right)x_m=\overset{\sim}{\eta}_j x_j^{b-1}-\overset{\sim}{\eta}_m x_m^{b-1},\]
and hence as  in (I) (i) we have $b=2$ or $b=1$ or $\boldsymbol{\eta}_1=\boldsymbol{\eta}_2$.
\end{enumerate}

\end{proof}

\begin{proof}[Proof of Theorem \ref{thm_assumption_Ad}]
Write $\iota_{+m}(\boldsymbol{x}_{-m})=\left(\boldsymbol{x}_{-m},1-\mathbf{1}_{m-1}^{\top}\boldsymbol{x}_{-m}\right)$ for any $\boldsymbol{x}_{-m}\in\Deltam_{-m}$. We first prove the finiteness of the normalizing constant. If $\mathbf{K}$ is positive definite, the inverse normalizing constant is
\begin{align*}
&\,\int_{\Deltam_{-m}}\exp\left(-\frac{1}{2}\log\iota_{+m}(\boldsymbol{x}_{-m})^{\top}\mathbf{K}\log\iota_{+m}(\boldsymbol{x}_{-m})+\boldsymbol{\eta}^{\top}\log\iota_{+m}(\boldsymbol{x}_{-m})\right)\d \boldsymbol{x}_{-m}\\
\leq&\,\int_{\Deltam_{-m}}\exp\left(\sum_{j=1}^{m}\left(-\lambda_{\min}(\mathbf{K})\left(\log \iota_{+m}\left(\boldsymbol{x}_{-m}\right)\right)_j^2+\eta_j\left(\log \iota_{+m}\left(\boldsymbol{x}_{-m}\right)\right)_j\right)\right)
\d\boldsymbol{x}_{-m}\\
\leq&\,\int_{\Deltam_{-m}}\exp\left(\sum_{j=1}^{m}\frac{\eta_j^2}{4\lambda_{\min}(\mathbf{K})}\right)
\d\boldsymbol{x}_{-m}<+\infty,
\end{align*}
proving (I). Now assume $\mathbf{K}$ is no longer positive definite. If $\mathbf{K}\mathbf{1}_m=\mathbf{0}_m$, for any $\boldsymbol{a}\in\mathbb{R}^m$, 
\begin{align*}
\boldsymbol{a}^{\top}\mathbf{K}\boldsymbol{a}&=\left(\boldsymbol{a}_{-j},a_j\right)^{\top}\begin{bmatrix}\mathbf{K}_{-j,-j} & \boldsymbol{\kappa}_{-j,j} \\ \boldsymbol{\kappa}_{-j,j}^{\top} & \kappa_{jj}\end{bmatrix} \left(\boldsymbol{a}_{-j},a_j\right)\\
&=\left(\boldsymbol{a}_{-j},a_j\right)^{\top}\begin{bmatrix}\mathbf{K}_{-j,-j} & -\mathbf{K}_{-j,-j}\mathbf{1}_{m-1} \\ -\mathbf{1}_{m-1}^{\top}\mathbf{K}_{-j,-j} & \mathbf{1}_{m-1}^{\top}\mathbf{K}_{-j,-j}\mathbf{1}_{m-1}\end{bmatrix} \left(\boldsymbol{a}_{-j},a_j\right)\\
&=\left(\boldsymbol{a}_{-j}-a_j\mathbf{1}_m\right)^{\top}\mathbf{K}_{-j,-j}\left(\boldsymbol{a}_{-j}-a_j\mathbf{1}_{m-1}\right),
\end{align*}
which is zero if and only if $\boldsymbol{a}_{-j}=a_j\mathbf{1}_{m-1}$, i.e.~$a_1=\cdots=a_m$, and is positive otherwise. Thus, if $\mathbf{K}\mathbf{1}_m=\mathbf{0}_m$, the condition that $\mathbf{K}_{-j,-j}$ is positive definite for some $j=1,\ldots,m$ is equivalent to that $\mathbf{K}_{-j,-j}$ is positive definite for all $j=1,\ldots,m$, and implies that $\mathbf{K}$ is positive semi-definite.

If $\mathbf{K}$ is positive semi-definite and $\boldsymbol{\eta}\succ-\mathbf{1}_m$, the inverse normalizing constant is
\begin{align*}
&\,\int_{\Deltam_{-m}}\exp\left(-\frac{1}{2}\log\iota_{+m}(\boldsymbol{x}_{-m})^{\top}\mathbf{K}\log\iota_{+m}(\boldsymbol{x}_{-m})+\boldsymbol{\eta}^{\top}\log\iota_{+m}(\boldsymbol{x}_{-m})\right)\d \boldsymbol{x}_{-m}\\
\leq&\,\int_{\Deltam_{-m}}\exp\left(\boldsymbol{\eta}^{\top}\log\iota_{+m}(\boldsymbol{x}_{-m})\right)\d\boldsymbol{x}_{-m}=\frac{\prod_{j=1}^m\Gamma(\eta_j+1)}{\Gamma\left(\mathbf{1}_{m}^{\top}\boldsymbol{\eta}+m\right)}<+\infty
\end{align*}
since the last quantity is the inverse normalizing constant of the Dirichlet distribution with parameters $(\boldsymbol{\eta}+\mathbf{1}_m)$, proving (III).

On the other hand, suppose $\mathbf{K}\mathbf{1}_m=\mathbf{0}_m$ and
$\mathbf{K}_{-j,-j}$ is positive definite for some/all
$j=1,\ldots,m$. Again letting $\boldsymbol{x}_{-m}$ be the free
variables and letting
$x_m=1-\mathbf{1}_{m-1}^{\top}\boldsymbol{x}_{-m}$, define the
\emph{additive log-ratio transformation} applied to $\boldsymbol{x}$:
$\boldsymbol{y}_{-m}\equiv
% \mathrm{alt}(\boldsymbol{x})\equiv
\log \boldsymbol{x}_{-m}-(\log x_{m})\mathbf{1}_{m-1}$, a random vector supported on $\mathbb{R}^{m-1}$. Append an extra $y_m=0$ for ease of notation. The transformation is thus bijective and the inverse transformation, the \emph{additive logistic transformation} $\boldsymbol{x}=\exp(\boldsymbol{y})/\mathbf{1}_m^{\top}\exp(\boldsymbol{y})$. Since
\[\partial x_k/\partial  y_j=-x_kx_j,\quad\partial x_j/\partial y_j=x_j(1-x_j)\]
for $j\neq k$, $j,k= l,\ldots,m-1$, we have
\begin{align*}
\left|\frac{\partial\boldsymbol{x}_{-m}}{\partial\boldsymbol{y}_{-m}}\right|&=\left|\begin{matrix}
x_1(1-x_1) & -x_1x_2 & \cdots & -x_1x_{m-1} \\
-x_1x_2 & x_2(1-x_2) & \cdots & -x_2x_{m-1} \\
\vdots & \vdots & \ddots & \vdots \\
-x_1x_{m-1} & -x_2x_{m-1} & \cdots & x_{m-1}(1-x_{m-1})
\end{matrix}\right|\\
&=\prod_{j=1}^m x_j=\exp\left(\mathbf{1}_m^{\top}\log\boldsymbol{x}\right).
\end{align*}
Then by $\mathbf{1}_m^{\top}\mathbf{K}=\mathbf{K}\mathbf{1}_m=\mathbf{0}_m$, $\boldsymbol{y}_{-m}$ has density proportional to
\begin{align*}
%&\,p_{\boldsymbol{Y}_{-m}}(\boldsymbol{y}_{-m})\\
&\,p(\boldsymbol{x}_{-m})\left|\partial\boldsymbol{x}_{-m}/\partial\boldsymbol{y}_{-m}\right| \\
\propto&\,\exp\left(-\frac{1}{2}\left(\log \boldsymbol{x}-(\log x_{m})\mathbf{1}_m\right)^{\top}\mathbf{K}\left(\log\boldsymbol{x}-(\log x_{m})\mathbf{1}_m\right)\right.\\
&\quad\quad\quad \left.\phantom{\frac{1}{2}} +(\boldsymbol{\eta}+\mathbf{1}_m)^{\top}\left(\log\boldsymbol{x}-(\log x_m)\mathbf{1}_m\right)+(\log x_m)\mathbf{1}_m^{\top}(\boldsymbol{\eta}+\mathbf{1}_m) \right)\\
=&\,\exp\left(-\frac{1}{2}\boldsymbol{y}^{\top}\mathbf{K}\boldsymbol{y}+(\boldsymbol{\eta}+\mathbf{1}_m)^{\top}\boldsymbol{y}+\mathbf{1}_m^{\top}(\boldsymbol{\eta}+\mathbf{1}_m)\log x_m \right)\\
=&\,\exp\left(-\frac{1}{2}\boldsymbol{y}_{-m}^{\top}\mathbf{K}_{-m,-m}\boldsymbol{y}_{-m}+(\boldsymbol{\eta}_{-m}+\mathbf{1}_{m-1})^{\top}\boldsymbol{y}_{-m}\right.\\
&\quad\quad\quad\left.-\left(\mathbf{1}_m^{\top}\boldsymbol{\eta}+m\right)\log\left(1+\mathbf{1}_{m-1}^{\top}\exp(\boldsymbol{y}_{-m})\right) \right).
\end{align*}
Note that $\log x_m=-\log\left(1+\mathbf{1}_{m-1}^{\top}\exp(\boldsymbol{y}_{-m})\right)<0$, so for $\mathbf{1}_m^{\top}\boldsymbol{\eta}+m\geq 0$ the last display is always upper-bounded by a constant times a normal density with a positive definite inverse covariance matrix $\mathbf{K}_{-m,-m}$, and thus the normalizing constant is finite, thus proving (II).

As for (A.1), fix $j=1,\ldots,m-1$ and any $\ell\in\{1,\ldots,m-1\}\backslash\{j\}$, and write $\boldsymbol{z}\equiv\log\boldsymbol{x}-(\log x_{\ell})\mathbf{1}_m$. Fix any $\boldsymbol{x}_{-j,-m}\in\mathbb{R}^{m-2}$ with $\boldsymbol{x}_{-j,-m}\succ\mathbf{0}$, $\mathbf{1}_{m-2}^{\top}\boldsymbol{x}_{-j,-m}<1$. Then if (I) $\mathbf{K}_0$ is positive definite or (II) $\mathbf{K}_0\mathbf{1}_m=\mathbf{0}_m$ and $\mathbf{K}_{0,-\ell,-\ell}$ is positive definite, by the proof above, $p_{0}(\boldsymbol{x}_{-m})x_i^{t}$ is upper bounded by a finite constant depending on $\mathbf{K}_0$ and $\boldsymbol{\eta}_0$ for any $t\in\mathbb{R}$ and $i=j,m$, since it is a constant times the density with parameters $\mathbf{K}_0$ and $\boldsymbol{\eta}_0+t\boldsymbol{e}_i$, and since we did not impose any restriction on the $\boldsymbol{\eta}$ parameter. On the other hand, for (III) $\mathbf{K}_0\mathbf{1}_m=\mathbf{0}_m$, $\mathbf{K}_{0,-\ell,-\ell}$ is positive semi-definite and $\boldsymbol{\eta}_0\succ -\mathbf{1}_m$,
\begin{align*}
p_{0}(\boldsymbol{x}_{-m})&\propto\exp\left(-\frac{1}{2}\log\boldsymbol{x}^{\top}\mathbf{K}_0\log\boldsymbol{x}+\boldsymbol{\eta}_0^{\top}\log\boldsymbol{x}\right)\nonumber\\
&=\exp\left(-\frac{1}{2}\boldsymbol{z}_{-\ell}^{\top}\mathbf{K}_{0,-\ell,-\ell}\boldsymbol{z}_{-\ell}+\boldsymbol{\eta}_{0,-\ell}^{\top}\boldsymbol{z}_{-\ell}+\left(\mathbf{1}_m^{\top}\boldsymbol{\eta}_0\right)\log x_{\ell}\right)\nonumber\\
&\leq\exp\left(\boldsymbol{\eta}_{0,-\ell}^{\top}\boldsymbol{z}_{-\ell}+\left(\mathbf{1}_m^{\top}\boldsymbol{\eta}_0\right)\log x_{\ell}\right)\nonumber\\
&\propto\exp\left(\eta_{0,j}z_j+\eta_{0,m}z_m\right).%\label{eq_proof_Ad_integration_by_parts}
\end{align*}
On the other hand, 
\begin{align*}
&\,\left|\partial_j\log p(\boldsymbol{x}_{-m})\right|\min\{x_j,x_m\}^{\alpha_j}\\
=&\,\left|-\boldsymbol{\kappa}_{,j}^{\top}\log \boldsymbol{x}/x_j+\boldsymbol{\kappa}_{,m}^{\top}\log \boldsymbol{x}/x_m+\eta_j/x_j-\eta_m/x_m\right|\min\{x_j,x_m\}^{\alpha_j}\\
=&\,\left(\left|\boldsymbol{\kappa}_{,j}^{\top}\log \boldsymbol{x}/x_j\right|+\left|\boldsymbol{\kappa}_{,m}^{\top}\log \boldsymbol{x}/x_m\right|+\left|\eta_j/x_j\right|+\left|\eta_m/x_m\right|\right)\min\{x_j,x_m\}^{\alpha_j}\\
\leq&\,\left(\left|\boldsymbol{\kappa}_{,j}^{\top}\log \boldsymbol{x}\right|+\left|\boldsymbol{\kappa}_{,m}^{\top}\log \boldsymbol{x}\right|+\left|\eta_j\right|+\left|\eta_m\right|\right)\min\{x_j,x_m\}^{\alpha_j-1}.
\end{align*}
Thus, as $x_j\searrow 0^+$ or $x_m\searrow 0^+$ (i.e.~$x_j\nearrow 1-\mathbf{1}_{m-2}^{\top}\boldsymbol{x}_{-j,-m}$), by multiplying the two bounds we have $p_{0}(\boldsymbol{x}_{-m})\left|\partial_j\log p(\boldsymbol{x}_{-m})\right|(h_j\circ\varphi_j)(\boldsymbol{x})\searrow 0^+$ for any $\alpha_j$ for (I) and (II) (by letting $t$ to e.g.~$\alpha_j-2$ in the discussion above), or for (III) by a constant times 
\begin{align*}
\left(\left|\boldsymbol{\kappa}_{,j}^{\top}\log \boldsymbol{x}\right|+\left|\boldsymbol{\kappa}_{,m}^{\top}\log \boldsymbol{x}\right|+\left|\eta_j\right|+\left|\eta_m\right|\right)\min\{x_j,x_m\}^{\alpha_j-1}x_j^{\eta_{0,j}}x_m^{\eta_{0,m}}\searrow 0^+
\end{align*}
if $\alpha_j>\max\{1-\eta_{0,j},1-\eta_{0,m}\}$.

As for (A.2), the results follow by a similar discussion for the Gamma model ($a$-$b$ model with $b=0$) on the standard simplex in Section \ref{$a$-$b$ Models on Standard Simplices}.
\end{proof}

\begin{proof}[Proof of Theorem \ref{theorem_simplex_nonlog_ggm}]
It suffices to bound $\boldsymbol{\Gamma}$ and $\boldsymbol{g}$ using their forms in Section \ref{Estimation-Simplex} and apply Theorem 1 in \citet{lin16}. We first bound $(h_j\circ\varphi_j)(\boldsymbol{x})x_j^{p_j}x_{m}^{p_m}$ with $h_j(x)=x^{\alpha_j}$, $\alpha_j\geq \max\{0,-p_j,-p_m,-p_j-p_m\}$, $p_j,p_m\in\mathbb{R}$, and $0<x_j+x_m<1$. By the definition of $\boldsymbol{\varphi}$ on simplices, $\varphi_j(\boldsymbol{x})=\min\{C_j,x_j,x_m\}$, so $(h_j\circ\varphi_j)(\boldsymbol{x})x_j^{p_j}x_{m}^{p_m}=\min\{C_j,x_j,x_m\}^{\alpha_j}x_j^{p_j}x_m^{p_m}$ and is tightly lower bounded by $0$. Noting that $\min\{x_j,x_m\}<1/2$, we consider the following cases.
\begin{enumerate}[(I)]
\item If $C_j<\min\{x_j,x_m\}$, then $C_j<1/2$ and the quantity is $C_j^{\alpha_j}x_j^{p_j}x_m^{p_m}\leq C_j^{\alpha_j+(p_j)_{-}+(p_m)_{-}}<2^{-\alpha_j-(p_j)_--(p_m)-}$ where $(y)_-=y$ if $y<0$ and $0$ otherwise. 
\item Otherwise suppose $x_j\leq x_m$ and $x_j\leq C_j$, then the quantity is equal to $x_j^{\alpha_j+p_j}x_m^{p_m}$, which is upper bounded by $x_j^{\alpha_j+p_j+p_m}<2^{-\alpha_j-p_j-p_m}<1$ if $p_m\leq 0$; if $p_m> 0$ it is upper bounded by $((\alpha_j+p_j)/(\alpha_j+p_j+p_m))^{\alpha_j+p_j}(p_m/(\alpha_j+p_j+p_m))^{p_m}$ if $(\alpha_j+p_j)/(\alpha_j+p_j+p_m)\leq 1/2$ or by $2^{-\alpha_j-p_j-p_m}$ otherwise. Note that the statement for $p_m>0$ covers the one for $p_m\leq 0$. The conclusion for $x_m\leq x_j$ and $x_m\leq C_j$ follows by symmetry, and note that at most one of $(\alpha_j+p_j)/(\alpha_j+p_j+p_m)\leq 1/2$ and $(\alpha_j+p_m)/(\alpha_j+p_j+p_m)\leq 1/2$ can hold.
\end{enumerate}
In conclusion, defining 
\[\zeta_2(\alpha_j,p_j,p_m)=\begin{cases}
\left(\frac{\alpha_j+p_j}{\alpha_j+p_j+p_m}\right)^{\alpha_j+p_j}\left(\frac{p_m}{\alpha_j+p_j+p_m}\right)^{p_m}, & \text{if }p_m\geq \alpha_j+p_j,\\
 \left(\frac{\alpha_j+p_m}{\alpha_j+p_j+p_m}\right)^{\alpha_j+p_m}\left(\frac{p_j}{\alpha_j+p_j+p_m}\right)^{p_j}, & \text{if }p_j\geq \alpha_j+p_m,\\ 
 2^{-\alpha_j-p_j-p_m}, & \text{otherwise},\end{cases}
 \]
 we have $(h_j\circ\varphi_j)(\boldsymbol{x})x_j^{p_j}x_{m}^{p_m}\leq \zeta_2(\alpha_j,p_j,p_m)<1$.
Similarly, $\partial_j(h_j\circ\varphi_j)(\boldsymbol{x})x_j^{p_j}x_{m}^{p_m}\leq \alpha_j\zeta_2(\alpha_j-1,p_j,p_m)<\alpha_j$, if $\alpha_j-1\geq\max\{0,-p_j,-p_m,-p_j-p_m\}$.

Then for all $j,k,\ell$, as long as $\alpha_j\geq \max\{1,2-a,2-b\}$, we have $0\leq \gamma_{j,k,\ell}<1$, and similarly $0\leq g_{j,k}<\max_{j=1,\ldots,m}\alpha_j+\max\{|a-1|+2a,|b-1|\}$. The rest follows from the same proof as \refyu{Theorem 5.3 of \citet{yu21}.}
 
Note that using the form of $\boldsymbol{\Gamma}$ in Section \ref{Estimation-Simplex}, a tighter bound for $\gamma_{j,k,\ell}$ is 
\begin{align*}
\max_{j,k=1,\ldots,m}\max\{&\zeta_2(\alpha_j,2a-2,0),\zeta_2(\alpha_j,4a-2,0),\zeta_2(\alpha_j,2a-2,2a),\\
&\zeta_2(\alpha_j,2b-2,0),\zeta_2(\alpha_j,0,2a-2),\zeta_2(\alpha_j,2a,2a-2),\\
&\zeta_2(\alpha_j,0,4a-2),\zeta_2(\alpha_j,0,2b-2),\zeta_2(\alpha_j,a-1,a-1),\\
&\zeta_2(\alpha_j,2a-1,2a-1),\zeta_2(\alpha_j,a-1,3a-1),\zeta_2(\alpha_j,3a-1,a-1),\\
&\zeta_2(\alpha_j,b-1,b-1)\},
\end{align*}
and the one for $g_{j,k}$ can be similarly written in terms of $\zeta_2(\alpha_j,\cdot,\cdot)$ and $\alpha_j\zeta_2(\alpha_j-1,\cdot,\cdot)$.
\end{proof}

\begin{proof}[Proof of Lemma \ref{theorem_subexp_simplex}]
For $a>0$ or $b>0$, the proof of \refyu{Lemma 5.1 of \citet{yu21}}
works even for the simplex domain. We thus only consider the case
where $a=b=0$, for which 
% Similar to the proof there, we prove by showing
we show that the moment-generating function of $\log X_j$ is finite and invoking Theorem 2.13 in \citet{wai19}. According to Theorem \ref{thm_assumption_Ad}, assume
\begin{enumerate}[(I)]
\item $\mathbf{K}_0$ is positive definite, or
\item $\mathbf{K}_0\mathbf{1}_m=\mathbf{0}$, $\mathbf{K}_{0,-k,-k}$ is positive definite for some $k=1,\ldots,m$, and $\mathbf{1}_m^{\top}\boldsymbol{\eta}+m\geq 0$, or
\item $\mathbf{K}_0\mathbf{1}_m=\mathbf{0}$, $\mathbf{K}_0$ is positive semi-definite, and $\boldsymbol{\eta}\succ-\mathbf{1}_m$.
\end{enumerate}
For any suitable $t$, $\mathbb{E}_0\exp(t\log X_j)$ is the inverse normalizing constant for the model with parameters $\mathbf{K}_0$ and $\boldsymbol{\eta}_0+t\boldsymbol{e}_j$, and is thus finite for (I) with $t\in\mathbb{R}$ and (III) with $t\in(-1-\eta_{0,j},+\infty)\ni 0$. % or (2) assuming $\mathbf{1}_m^{\top}\boldsymbol{\eta}+m>0$ with $t\in(-m-\mathbf{1}_m^{\top}\boldsymbol{\eta},+\infty)\ni 0$. 
For (II), recall that in the proof of Theorem \ref{thm_assumption_Ad} we have shown that for any $k=1,\ldots,m$, the density of $\log \boldsymbol{X}_{-k}-\left(\log X_k\right)\mathbf{1}_{m-1}$ is bounded by a constant times a Gaussian density, and thus $\mathbb{E}_0\left[X_j^t/X_k^t\right]=\mathbb{E}_0\exp\left(t(\log X_j-\log X_{k})\right)<+\infty$ for any $k=1,\ldots,m$ and $t\in\mathbb{R}$. So for any $t<0$,
\begin{align*}
\mathbb{E}_0 X_j^t&\leq\mathbb{E}_0\left[X_j^t|X_j\geq 1/m\right]\mathbb{P}\left(X_j\geq 1/m\right)+\sum_{k\neq j}\mathbb{E}_0\left[\left.X_j^t\right|X_k\geq 1/m\right]\mathbb{P}\left(X_k\geq 1/m\right)\\
&\leq m^{-t}\mathbb{P}\left(X_j\geq 1/m\right)+\sum_{k\neq j}m^{-t}\mathbb{E}_0\left[\left.X_j^t/X_k^t\right|X_k\geq 1/m\right]\mathbb{P}\left(X_k\geq 1/m\right)\\
&\leq m^{-t}+\sum_{k\neq j}m^{-t}\mathbb{E}_0\left[X_j^t/X_k^t\right]<+\infty.
\end{align*}
On the other hand, $\mathbb{E}_0X_j^t\leq 1$ for $t\geq 0$. Thus, $\mathbb{E}_0\exp(t\log X_j)<+\infty$ for any $t\in\mathbb{R}$ for (II). Hence, for all of (I)--(III) we have $\mathbb{E}_0\exp(t\log X_j)<+\infty$ for $t$ in a neighborhood around $0$.
\end{proof}

\begin{proof}[Proof of Corollary \ref{theorem_simplex_log_ggm}]
Let $\|\log X_j\|_{\psi_1}\equiv\sup_{q\geq 1}(\mathbb{E}_0|\log X_j|^q)^{1/q}/q$ be the sub-exponential norm of $\log X_j$, then by Lemma 21.6) of \citet{yu19} or Corollary 5.17 of \citet{ver12},
\[\mathbb{P}\left(-\log X_j+\mathbb{E}_0\log X_j\geq\epsilon_3\right)\leq\exp\left(-\min\left(\frac{\epsilon_3^2}{8e^2\|\log X_j\|_{\psi_1}^2},\frac{\epsilon_3}{4e\|\log X_j\|_{\psi_1}}\right)\right).\]
Letting 
\begin{multline*}
\epsilon_3\equiv\max\Big\{2\sqrt{2}e\max_j\|\log X_j\|_{\psi_1}\sqrt{\log 3+\log n+(\tau+1)\log m},\\
4e\max_j\|\log X_j\|_{\psi_1}(\log 3+\log n+(\tau+1)\log m)\Big\},
\end{multline*}
we get $0\leq -\log X_j^{(i)}\leq \max_k\mathbb{E}_0\log X_k+\epsilon_3$ for all $j=1,\ldots,m$ and $i=1,\ldots,n$ with probability at least $1-1/(3m^{\tau})$. The rest follows as in the proof of \refyu{Theorem 5.3 of \citet{yu21}} and Theorem \ref{theorem_simplex_nonlog_ggm}.
\end{proof}

%%% Local Variables:
%%% mode: latex
%%% TeX-master: "simplex_domain_ejs"
%%% End:

% Main text entry area

% \begin{acks}[Acknowledgments]
% \end{acks}

%%%%%%%%%%%%%%%%%%%%%%%%%%%%%%%%%%%%%%%%%%%%%%
%% Supplementary Material, if any, should   %%
%% be provided in {supplement} environment  %%
%% with title and short description.        %%
%%%%%%%%%%%%%%%%%%%%%%%%%%%%%%%%%%%%%%%%%%%%%%
%\begin{supplement}
%\stitle{???}
%\sdescription{???.}
%\end{supplement}

%% if your bibliography is in bibtex format, uncomment commands:
%\bibliographystyle{imsart-number} % Style BST file (imsart-number.bst or imsart-nameyear.bst)
%\bibliography{bibliography}       % Bibliography file (usually '*.bib')

%% or include bibliography directly:
% \begin{thebibliography}{}
% \bibitem{b1}
% \end{thebibliography}

\bibliographystyle{plainnat}
\bibliography{Paper}

\end{spacing}
\end{document}